\documentclass[aip,jmp,preprint,frafta4paper]{revtex4-1}
\usepackage{cancel}
\usepackage{amsmath}
\usepackage{amssymb}
\usepackage{amsthm}
\usepackage{graphicx}
\usepackage{dcolumn}
\usepackage{bm}
\usepackage{upgreek}
\usepackage{esint}

\usepackage{commands}

\theoremstyle{definition}
\newtheorem{definition}{Definition}

\theoremstyle{remark}
\newtheorem{remark}{Remark}

\theoremstyle{theorem}
\newtheorem{theorem}{Theorem}

\theoremstyle{lemma}
\newtheorem{lemma}{Lemma}

\theoremstyle{corollary}
\newtheorem{corollary}{Corollary}

\theoremstyle{proposition}
\newtheorem{proposition}{Proposition}

\begin{document}


\title{Variational nonlinear WKB in the Eulerian frame} 



\author{J. W. Burby}
\affiliation{Los Alamos National Laboratory, Los Alamos, New Mexico 87545, USA}
 \affiliation{Mathematical Sciences Research Institute, 17 Gauss Way, Berkeley, California 94720, USA}
\author{D. E. Ruiz}
 \affiliation{Sandia National Laboratories, P.O. Box 5800, Albuquerque, New Mexico 87185, USA}


\date{\today}

\begin{abstract}
Nonlinear WKB is a multiscale technique for studying locally-plane-wave solutions of nonlinear partial differential equations (PDE). Its application comprises two steps: (1) replacement of the original PDE with an extended system separating the large scales from the small, and (2) reduction of the extended system to its slow manifold. In the context of variational fluid theories with particle relabeling symmetry, nonlinear WKB in the \emph{mean} Eulerian frame is known to possess a variational structure. This much has been demonstrated using, for instance, the theoretical apparatus known as the generalized Lagrangian mean. On the other hand, the variational structure of nonlinear WKB in the conventional Eulerian frame remains mysterious. By exhibiting a variational principle for the extended equations from step (1) above, we demonstrate that nonlinear WKB in the Eulerian frame is in fact variational. Remarkably, the variational principle for the extended system admits loops of relabeling transformations as a symmetry group. Noether's theorem therefore implies that the extended Eulerian equations possess a \emph{family} of circulation invariants parameterized by $S^1$. As an illustrative example, we use our results to systematically deduce a variational model of high-frequency acoustic waves interacting with a larger-scale compressible isothermal flow.
%
\end{abstract}


\maketitle 

\section{Introduction}
Nonlinear WKB is a powerful tool for studying solutions of partial differential equations (PDE) whose local behavior about any point is well approximated by a plane wave. The method, which is a generalization of the usual WKB method for linear PDE, goes back at least to the mid 1960's, when it was used to study large-amplitude locally-plane-wave solutions of a variety of systems, including the Bousinesque equations\cite{Whitham_1965_eom} and the Kortweg-DeVries equation.\cite{Miura_1974} Generally speaking, given a (possibly nonlinear) PDE of the form
\begin{align}
F^b(\varphi^a(x),\partial_\mu\varphi^a(x),\partial^2_{\mu\nu}\varphi^a(x),\dots) = 0,\label{gen_PDE}
\end{align}
for the unknown multi-component field $\varphi^a$, application of nonlinear WKB comprises two steps. First Eq.\,\eqref{gen_PDE} is extended to a larger system of PDE using a procedure that we will refer to as ``nonlinear WKB extension."  Next, scale separation present in the original system of PDE, either in $F^b$ or the initial conditions, is leveraged to identify slow solutions of the extended system. The power of this procedure comes from the fact that rapidly oscillating locally-plane-wave solutions $\varphi^a$ of Eq.\,\eqref{gen_PDE} correspond to slowly-varying solutions of the extended system, which are easier to treat using asymptotic methods.

The nonlinear WKB extension procedure amounts to the following. First one introduces the nonlinear WKB ansatz 
\begin{align}
\varphi^a(x) = \tilde{\varphi}^a(x,S(x)),\label{WKB_ansatz}
\end{align}
where $\tilde{\varphi}^a(x,\theta)$ is $2\pi$-periodic in the second argument, and $S(x)$ is referred to as a phase function.
More explicitly, since $\tilde{\varphi}^a$ is periodic in the second argument, it can be written as a sum of Fourier harmonics in $S(x)$; that is,
\begin{align}
\tilde{\varphi}^a(x,\theta)
	&=\sum_{n=-\infty}^{\infty}
		\hat{\varphi}_n^a e^{inS(x)}.
\end{align}
Thus, the nonlinear WKB ansatz differs from the conventional WKB ansatz in that it contains all harmonics in $S$. The term ``nonlinear" is appropriate here because the ansatz \eqref{WKB_ansatz} can handle nonlinear terms appearing in the PDE \eqref{gen_PDE} that produce harmonic coupling.
The ansatz \eqref{WKB_ansatz} is then substituted into Eq.\,\eqref{gen_PDE} and the chain rule is applied to express $x$-derivatives of ${\varphi}$ in terms of $x$- and $\theta$-derivatives of $\tilde{\varphi}$ and $S$. Finally, the argument $S(x)$ in any of the derivatives of $\tilde{\varphi}^a$ is replaced with any arbitrary angle $\theta$ in order to obtain the extended system
\begin{align}
F^b(\tilde{\varphi}^a(x,\theta),\partial_\mu\tilde{\varphi}^a(x,\theta)+\partial_\mu S(x) \partial_\theta\tilde{\varphi}^a(x,\theta),\dots)=0.\label{WKB_extension_eom}
\end{align}
The dependent variables are now $\tilde{\varphi}^a(x,\theta)$ and $S(x)$, while the independent variables are $x$ and $\theta$. As is readily checked, each solution $(\tilde{\varphi}^a,S)$ of Eq.\,\eqref{WKB_extension_eom} yields a solution $\varphi^a$ of Eq.\,\eqref{gen_PDE}, with $\varphi^a$ given by Eq.\,\eqref{WKB_ansatz}. It is in this sense that Eq.\,\eqref{WKB_extension_eom} \emph{extends} the original equation \eqref{gen_PDE}.

In this paper we will study the nonlinear WKB extension procedure, i.e. the passage from Eq.\,\eqref{gen_PDE} to Eq.\,\eqref{WKB_extension_eom}, as an interesting mathematical construction in its own right, independent of any asymptotic methods. Specifically, we will be concerned with nonlinear WKB extension as it applies to a particular class of PDE from fluid mechanics known as Euler-Poincar\'e equations.\cite{Holm_1998} Such equations describe the evolution of ideal, i.e. dissipation-free, fluids. In the Euler-Poincar\'e setting, we will address the question of whether structural properties of the original system of PDE \eqref{gen_PDE} are inherited by the extended equations \eqref{WKB_extension_eom}. We will be particularly interested in the fate of variational structure and particle relabeling symmetry, the latter being the source of circulation invariants in ideal fluid models.


The methods of Whitham\cite{Whitham_1965_lag} are sufficient to study the fate of variational structure under nonlinear WKB extension when the system \eqref{gen_PDE} is equivalent to the Euler-Lagrange equations associated with a classical field theory. As we will review, Whitham's method of averaged Lagrangians provides a variational principle for the extended system in this case. However, conventional Euler-Poincar\'e variational principles for ideal fluid flow do not fit into the mold of variational principles used in classical field theory. Therefore Whitham's methods cannot be applied directly to show that the system \eqref{WKB_extension_eom} is variational when Eq.\,\eqref{gen_PDE} is an Euler-Poincar\'e fluid equation.

The essential difficulty can be understood through a close look at the ideal isothermal Euler equations
\begin{gather}
\rho (\partial_t\bm{u}+\bm{u}\cdot\del\bm{u}) = -c^2 \del\rho\label{euler_mom}\\
\partial_t\rho+\del\cdot(\rho\bm{u}) = 0,\label{euler_cont}
\end{gather}
where the unknown fields are the fluid velocity $\bm{u}(\bm{x},t) $ and the mass density $\rho(\bm{x},t)$, and $c$ is a constant representing the speed of small-amplitude sound waves. This system of equations, which has the form \eqref{gen_PDE}, arises from an Euler-Poincar\'e variational principle in the following sense.\cite{Holm_1998} Let $Q$ be a compact region in $\mathbb{R}^3$ that represents the fluid container, and let $Q_0$ be a diffeomorphic copy of $Q$ equipped with a non-vanishing function $\rho_0:Q_0\rightarrow \mathbb{R}$ that represents a reference configuration of fluid elements. A path $t\mapsto \bm{g}(t)\in\text{Diff}(Q_0,Q)$ in the space of diffeomorphisms $Q_0\rightarrow Q$ is a critical point of the functional 
\begin{align}
\mathcal{A}_{\rho_0}(\bm{g}) = \int_{t_1}^{t_2} \int_{Q_0} \frac{1}{2} |\dot{\bm{g}}(\bm{x}_0)|^2 \,\rho_0(\bm{x}_0)\,d\bm{x}_0-\int c^2 \rho_0(\bm{x}_0)\text{ln}\left(\frac{\rho_0(\bm{x}_0)}{\text{det}(\del_0\bm{g})(\bm{x}_0)}\right) \,d\bm{x}_0
\end{align}
if and only if Eq.\,\eqref{euler_mom} is satisfied with $\bm{u}$ and $\rho$ \emph{defined} according to
\begin{align}
\bm{u}(\bm{x}) &= \dot{\bm{g}}(\bm{g}^{-1}(\bm{x}))\label{eulerian_velocity}\\
\rho(\bm{x}) & = \frac{1}{\text{det}(\del_0\bm{g})(\bm{g}^{-1}(\bm{x}))}\rho_0(\bm{g}^{-1}(\bm{x})),\label{density_advection}
\end{align}
where we have suppressed the time argument $t$ for the sake of presentation. In particular, it is a consequence of these definitions that $\rho$ as defined in Eq.\,\eqref{density_advection} satisfies the continuity equation \eqref{euler_cont}. Thus, each critical point of $\mathcal{A}_{\rho_0}$ corresponds to a solution of the ideal isothermal Euler equations. Conversely, given a solution of the ideal isothermal Euler equations, there is some $\rho_0$ such that the time-dependent flow map of $\bm{u}$ is a critical point of $\mathcal{A}_{\rho_0}$. It is therefore appropriate to say that the system \eqref{euler_mom}-\eqref{euler_cont} is variational. However, the field that appears in the variational principle is $\bm{g}(\bm{x}_0,t)$ instead of $\bm{u}(\bm{x},t)$ or $\rho(\bm{x},t)$, as one might expect from experience with classical field theory. In fact, $\bm{g}$ is not even defined on the same \emph{domain} as $\bm{u}$ and $\rho$. It is therefore not at all obvious how, or if, Whitham's averaged Lagrangian technique can be applied to yield a variational principle for the nonlinear WKB extension of Eqs.\,\eqref{euler_mom}-\eqref{euler_cont},
\begin{gather}
\tilde{\rho}\left(\partial_t\tilde{\bm{u}}+\tilde{\bm{u}}\cdot\del\tilde{\bm{u}}+\Omega\partial_\theta\tilde{\bm{u}}\right) = -c^2\del\tilde{\rho}-c^2\del S\,\partial_\theta\tilde{\rho}\label{euler_ext_mom}\\
\partial_t\tilde{\rho}+\del\cdot(\tilde{\rho}\tilde{\bm{u}})+\partial_\theta(\Omega\tilde{\rho})= 0\label{euler_ext_cont}\\
\Omega = \pd_t{S}+\tilde{\bm{u}}\cdot\del S,
\end{gather}
where $\tilde{\rho} = \tilde{\rho}(\bm{x},t,\theta)$ and $\tilde{\bm{u}}(\bm{x},t,\theta)$ comprise the multi-component field $\tilde{\varphi}^a$ in Eq.\,\eqref{WKB_extension_eom} and $S = S(\bm{x},t)$ is the phase function.
For instance, one question that arises when attempting to apply Whitham's averaging to the action functional $\mathcal{A}_{\rho_0}$ is ``what is the appropriate nonlinear WKB ansatz for the mapping $\bm{g}$?'' The naive guess $\bm{g}(\bm{x}_0,t) = \tilde{\bm{g}}(\bm{x}_0,t,S(\bm{x}_0,t))$ does not make sense because the proper spatial domain of the phase function $S$ is $Q$ --- not $Q_0$. (In WKB theory, phases are assigned to spatial locations, not fluid element labels.)

We are by no means the first to consider the interplay between WKB theory and Euler-Poincar\'e variational principles. Before the terminology ``Euler-Poincar\'e variational principle" was even invented, Dewar\cite{Dewar_1970} (and independently Bretherton\cite{Bretherton_1971}) proposed the ansatz
\begin{align}
\bm{g}(\bm{x}_0,t) &= \overline{\bm{g}}(\bm{x}_0,t)+\bm{\xi}(\overline{\bm{g}}(\bm{x}_0,t),t)\label{old_lagrangian_wkb}\\
\bm{\xi}(\overline{\bm{x}},t) & = \text{Re}\left( \bm{a}(\overline{\bm{x}},t) \exp(i S(\overline{\bm{x}},t))\right)
\end{align}
for the fluid configuration map that appears in the Euler-Poincar\'e variational principle for magnetohydrodynamic (MHD) flow.\cite{Newcomb_1962} The intuition leading to \eqref{old_lagrangian_wkb} is that $\overline{\bm{g}}$ represents the ``mean" configuration of Lagrangian fluid elements. Under the assumptions that $\bm{\xi}$ is small and $S$ varies rapidly, this idea leads to a variational model of small-amplitude locally-plane waves interacting with a slowly-varying MHD background. In the context of purely hydrodynamic flow, Gjaja and Holm\cite{Gjaja_Holm_1996} explored this idea further, and uncovered the consequences of the \emph{mean} relabeling symmetry present in averaged Lagrangians built upon Eq.\,\eqref{old_lagrangian_wkb}. Here mean fluid particle relabeling symmetry refers to invariance of the averaged Lagrangian under the replacement
\begin{align}
\overline{\bm{g}}_t\mapsto \overline{\bm{g}}_t\circ \overline{\bm{\eta}},
\end{align}
where $\overline{\bm{\eta}}:Q_0\rightarrow Q_0$ is any diffeomorphism that preserves the density $\rho_0\, d\bm{x}_0$. Perhaps surprisingly, none of this previous work manages to provide a variational principle for the usual WKB extension of the fluid equations, e.g. Eqs.\,\eqref{euler_ext_mom}-\eqref{euler_ext_cont}. Instead, the ansatz \eqref{old_lagrangian_wkb} leads to an alternative, ostensibly inequivalent extension of the fluid equations,\cite{Gjaja_Holm_1996} and Whitham averaging produces a variational structure for this alternative system. It is reasonable to refer to this alternative extension as an extension \emph{in the mean Eulerian frame} because the quantity $\overline{\bm{x}} = \overline{\bm{g}}(\bm{x}_0,t)$ gives the \emph{phase average}  of the Eulerian fluid element location. Thus, nonlinear WKB extension in the mean Eulerian frame is known to be variational. However, the variational structure of nonlinear WKB extension in the conventional Eulerian frame has never been found.

In what follows, we will prove that nonlinear WKB extension in the Eulerian frame is in fact variational. The proof will make use of Whitham averaging, but will \emph{not} make use of the ansatz \eqref{old_lagrangian_wkb}. In fact, even the more general notion of separating quantities into mean and fluctuating parts will not play a role in the argument. Exploiting this fact, we will also prove that the extended equations admit as a symmetry group the space of \emph{loops} of particle relabeling transformations. Remarkably, this loop group\cite{Pressley_1988} is much larger than the group of mean relabeling transformations present in the work of Gjaja and Holm.\cite{Gjaja_Holm_1996} The presence of this loop group symmetry will allow us to prove that the nonlinear WKB extension of Euler-Poincar\'e fluid equations in the Eulerian frame admits a \emph{family} of circulation invariants parameterized by $S^1$. This result extends the circulation theorem of Gjaja and Holm, which may be seen as a consequence of invariance under the subgroup of constant loops. Finally, in order to demonstrate the utility of our results, we will apply them to derive a systematic, all-orders variational model of weakly-nonlinear high-frequency acoustic waves interacting with a longer-scale isothermal compressible flow. This example may be regarded as a fresh take on the analysis of Bretherton in Ref.\,\onlinecite{Bretherton_1971}.

Our discussion will be organized in the following manner. In Section \ref{sec:two} we will prove that applying Whitham averaging to the Lagrangian of a classical field theory is equivalent to applying the usual nonlinear WKB extension procedure directly to the Euler-Lagrange equations. In particular, we will show that Whitham's averaged Lagrangian is the Lagrangian for the nonlinear WKB extension of a classical field theory. In Section \ref{sec:three} we will show how fluid equations arising from Euler-Poincar\'e variational principles with local Lagrangians may be recast as classical field theories. In Section \ref{sec:four} we will then combine the results of Section \ref{sec:two} and \ref{sec:three} to produce the variational structure underlying the nonlinear WKB extension (in the Eulerian frame) of Euler-Poincar\'e fluid equations. We will investigate the relabeling symmetries of this new variational principle in Section \ref{sec:five}. In particular, we will prove that the symmetry group of the WKB extension includes the space of loops of particle relabeling transformations, and identify the corresponding momentum map using Noether's theorem. Finally, we will apply our results to high-frequency acoustic waves interacting with longer-scale compressible isothermal flow in Section \ref{sec:six}. After presenting our results, we will discuss the relationship of our work with existing literature, in particular with Ref.\,\onlinecite{Gjaja_Holm_1996} and the theory of generalized Lagrangian means, in Section \ref{sec:seven}.

{As a forewarning remark, unless indicated otherwise, we will assume in this paper that all mappings are $C^\infty$. We make this assumption in spite of the fact that some of the PDEs we will encounter may not have a good existence and uniqueness theory in the smooth setting. In addition, the discussion contained in Section \ref{sec:six} will proceed at the level of formal asymptotics. Regularity assumptions mentioned in Section \ref{sec:six} are merely included to ensure that coefficients in various asymptotic expansion may be computed.
}

\section{A basic theorem on Whitham averaging\label{sec:two}}
In this section we provide an anachronistic review of the nonlinear WKB extension procedure as it applies to general first-order classical field theories. Our goal is to prove that the nonlinear WKB extension of a field theory satisfies a variational principle. We will build upon this result in subsequent sections when uncovering the variational structure of nonlinear WKB extensions of Euler-Poincar\'e fluid equations. Essentially all of the ideas in this section can be found in the work of Whitham.\cite{Whitham_1965_lag} 

For the purposes of our discussion, a first-order classical field theory will be defined as follows. 
\begin{definition}\label{def_ft}
A \emph{first-order classical field theory} is a triple $(M,\mathcal{C},\mathcal{L})$ comprising a manifold $M$, a space of functions $\mathcal{C}$, and a function $\mathcal{L}$ with the following properties.
\begin{itemize}
\item The \emph{spacetime} $M$ is an $m$-dimensional space presented as the product of a vector space with a torus of some dimension. The natural coordinates on $M$ are denoted $x^\mu$, $\mu\in\{1,\dots,m\}$.
\item The \emph{space of fields} $\mathcal{C}$ is a vector space of functions $\varphi:M\rightarrow F$, where the \emph{fiber} $F$ is an $f$-dimensional space of the same type as $M$, but with a possibly different dimension. The natural coordinates on $F$ are denoted $\varphi^a$, $a\in\{1,\dots,f\}$.
\item The \emph{Lagrangian density} $\mathcal{L}$ is a real-valued function on $M\times F\times D$, where $D$ is the space of $f\times m$ matrices with components $v^a_\mu$, $a\in \{1,\dots,f\}$, $\mu\in\{1,\dots,m\}$.
\end{itemize}
\end{definition}
\noindent To each first-order classical field theory and compact subset $U\subset M$, we associate the local action functional
\begin{align}
A_U(\varphi) = \int_U \mathcal{L}(x,\varphi(x), \partial\varphi(x) )\,dx.\label{eq_16}
\end{align}
Here $\partial\varphi(x)\in D$ has entries $[\partial\varphi(x)]^a_\mu = \partial_\mu\varphi^a(x)$. We say that a field $\varphi$ is a critical point of $A_U$ if 
\begin{gather}
\frac{d}{d\epsilon}\bigg|_0A_U(\varphi+\epsilon \delta\varphi) = 0
\end{gather}
for all $\delta\varphi\in \mathcal{C}$ that vanish on $\partial U$. 

Suppose that $\mathcal{C}$ contains all smooth fields with compact support. Then it is a standard result in the calculus of variations that $\varphi$ is a critical point of $A_U$ for all $U\subset M$ if and only if $\varphi$ satisfies the system of second-order PDE known as the Euler-Lagrange equations:
\begin{align}\label{ele_general}
\frac{\partial\mathcal{L}}{\partial \varphi^a}(x,\varphi(x),\partial\varphi(x)) = \frac{\partial}{\partial x^\mu} \left(\frac{\partial\mathcal{L}}{\partial v^a_\mu}(x,\varphi(x),\partial\varphi(x))\right).
\end{align}
In this setting, we refer to the first-order classical field theory $(M,\mathcal{C},\mathcal{L})$ as \emph{ordinary}.
\begin{definition}
An \emph{ordinary} first-order classical field theory is a first-order classical field theory $(M,\mathcal{C},\mathcal{L})$ where the space of fields $\mathcal{C}$ contains all smooth fields $\varphi: M\rightarrow F$ with compact support.
\end{definition}

Given an ordinary first-order classical field theory $(M,\mathcal{C},\mathcal{L})$, the nonlinear WKB extension procedure described in the introduction may be applied to the theory's Euler-Lagrange equations. It will be convenient to refer to the resulting extended system as the nonlinear WKB extension of $(M,\mathcal{C},\mathcal{L})$.
\begin{definition}\label{def_2}
\emph{The nonlinear WKB (NL-WKB) extension} of the ordinary first-order classical field theory $(M,\mathcal{C},\mathcal{L})$ is the nonlinear WKB extension of the field theory's Euler-Lagrange equations \eqref{ele_general}. That is, the NL-WKB extension of $(M,\mathcal{C},\mathcal{L})$ is the system of partial differential equations
\begin{align}
\label{pre_extension}
\frac{\partial\mathcal{L}}{\partial \varphi^a}(j(x,\theta)) =& \left(\frac{\partial}{\partial x^\mu}+\partial_\mu S(x)\frac{\partial}{\partial\theta}\right)\left(\frac{\partial\mathcal{L}}{\partial v^a_\mu}(j(x,\theta))\right),
\end{align}
where
\begin{align}
j(x,\theta)= (x,\tilde{\varphi}(x,\theta),\partial\tilde{\varphi}(x,\theta) + \partial_\theta\tilde{\varphi}(x,\theta)\partial S(x))\label{shorthand_j}
\end{align}
is convenient shorthand notation, and $\tilde{\varphi}:M\times S^1\rightarrow F$ and $S:M\rightarrow S^1$ are the unknown fields in the extended system. We regard $\partial_\theta\tilde{\varphi}(x,\theta)$ and $\partial S(x)$ as $f\times 1$ and $1\times m$ matrices, respectively.
\end{definition}
\noindent Our goal is to describe elements of the relationship between the ordinary field theory $(M,\mathcal{C},\mathcal{L})$ and its NL-WKB extension. Due to the following Lemma, we expect this relationship to be strong.

\begin{lemma}\label{lemma_1}
If $(\tilde{\varphi},S)$ is a solution of the NL-WKB extension of $(M,\mathcal{C},\mathcal{L})$, then $\varphi(x) = \tilde{\varphi}(x,S(x))$ is a solution of the Euler-Lagrange equations associated with $\mathcal{L}$. Conversely, if $\varphi$ is a solution of $\mathcal{L}$'s Euler-Lagrange equations, then $\tilde{\varphi}(x,\theta) = \varphi(x)$, $S(x) = 0$ is a solution of the NL-WKB extension of $(M,\mathcal{C},\mathcal{L})$.
\end{lemma}
\begin{proof}
That $\varphi(x) = \tilde{\varphi}(x,S(x))$ satisfies the Euler-Lagrange equations associated with $\mathcal{L}$ is a straightforward application of the chain rule. The converse statement follows from the fact that Eq.\,\eqref{pre_extension} reduces to Eq.\,\eqref{ele_general} when $\tilde{\varphi}(x,\theta) = \varphi(x)$ and $S(x) = 0$.
\end{proof}


In order to go beyond Lemma \ref{lemma_1} in our description of the relationship between an ordinary classical field theory and its nonlinear WKB extension, it is useful to understand the heuristic origins of the nonlinear WKB extension procedure. The key idea is scale separation. Suppose $\varphi$ is a solution of the Euler-Lagrange equations associated with $(M,\mathcal{C},\mathcal{L})$ that locally has the appearance of a plane wave. Formally, we may then write $\varphi(x) = \tilde{\varphi}(x,S(x))$, where $\tilde{\varphi}:M\times S^1\rightarrow F$ is a profile and $S$ is a rapidly oscillating phase function. The derivatives of $\varphi$ are apparently given by
\begin{align}
\partial_\mu\varphi^a(x) = \partial_\mu\tilde{\varphi}^a(x,S(x)) + \partial_\mu S(x) \,\partial_\theta\tilde{\varphi}^a(x,S(x)).
\end{align}
In light of the Euler-Lagrange equations \eqref{ele_general}, the profile and phase function must therefore satisfy
\begin{align}\label{pre_eom}
\frac{\partial\mathcal{L}}{\partial \varphi^a}(j(x,S(x))) &= \frac{\partial}{\partial x^\mu}\left(\frac{\partial\mathcal{L}}{\partial v^a_\mu}(j(x,S(x)))\right)\nonumber\\
&=\left(\frac{\partial}{\partial x^\mu}+\partial_\mu S(x)\frac{\partial}{\partial\theta}\right)\left(\frac{\partial\mathcal{L}}{\partial v^a_\mu}(j(x,\theta))\right)\bigg|_{\theta = S(x)},
\end{align}
where we have used the shorthand notation $j(x,\theta)$ introduced in Definition \ref{def_2}.
Because the phase function is, by hypothesis, rapidly rotating, we can extract more information from Eq.\,\eqref{pre_eom} by considering the latter in a spacetime region that is small compared with the long spacetime scale, but large compared with the short spacetime scale. In such a region, we may regard the argument $x$ in $j(x,S(x))$ as being fixed, while the argument $S(x)$ retains its rapidly oscillating character. If we make the (very weak) assumption that $S(x)$ makes at least one complete rotation in our intermediate-scale region, we may therefore conclude that the following strengthened version of Eq.\,\eqref{pre_eom} must be satisfied: 
\begin{align}\label{extended_ele_heuristic}
\frac{\partial\mathcal{L}}{\partial \varphi^a}(j(x,\theta)) =& \left(\frac{\partial}{\partial x^\mu}+\partial_\mu S(x)\frac{\partial}{\partial\theta}\right)\left(\frac{\partial\mathcal{L}}{\partial v^a_\mu}(j(x,\theta))\right),
\end{align}
where $\theta\in S^1$ is now arbitrary. Equation \eqref{extended_ele_heuristic} reproduces the NL-WKB extension of $(M,\mathcal{C},\mathcal{L})$. Thus, the NL-WKB extension of $(M,\mathcal{C},\mathcal{L})$ can be deduced by applying heuristic arguments based on scale separation to the Euler-Lagrange equations associated with $\mathcal{L}$.

Now consider the application of similar heuristic arguments to the variational principle associated with $(M,\mathcal{C},\mathcal{L})$. Suppose once more that $\varphi$ is a solution of the field theory that is locally a plane wave. Then, as before, we may write $\varphi(x) = \tilde{\varphi}(x,S(x))$, where $\tilde{\varphi}:M\times S^1\rightarrow F$ is a profile and $S$ is a rapidly oscillating phase function. Moreover, the action $A_U$ evaluated on this special $\varphi$ can be written
\begin{align}
A_U(\varphi) = \int_U \mathcal{L}(j(x,S(x)))\,dx.
\end{align}
Because the phase function $S$ is rapidly oscillating by hypothesis, we may partition the integration domain $U = \cup_i U_i$ into cells with diameters that are large compared with the short scale and short compared with the large scale, and then write $A_U(\varphi) = \sum_i A_{U_i}(\varphi)$. In each of the integrals $A_{U_i}$ the first argument of $j(x,S(x))$ may be replaced with the center $x_i$ of cell $U_i$ without appreciably changing the value of the integral. Moreover, because $S(x)$ varies rapidly in $U_i$, the dominant contribution to the integral $A_{U_i}$ is given by averaging over $S(x)$ according to
\begin{align}
A_{U_i}\approx \frac{1}{2\pi} \int_0^{2\pi}\int_{U_i} \mathcal{L}(j(x_i,\theta))\,dx\,d\theta.
\end{align}
If we now interpret the previously established formula $A_U(\varphi) = \sum_i A_{U_i}(\varphi)$ as a Riemann sum, we conclude that the action functional evaluated on a locally-plane $\varphi$ is given approximately by
\begin{align}
A_U(\varphi)&\approx \frac{1}{2\pi}\int_0^{2\pi}\int_{U} \mathcal{L}(j(x,\theta))\,dx\,d\theta\nonumber\\
&\equiv \tilde{A}_{U\times S^1}(\tilde{\varphi},S),
\end{align}
where we have introduced the extended action functional $\tilde{A}_{U\times S^1}(\tilde{\varphi},S)$. Moreover, because $\varphi$ is by assumption a critical point of $A_U$, this argument suggests that 
\begin{align}
\frac{d}{d\epsilon}\bigg|_0 \tilde{A}_{U\times S^1}(\tilde{\varphi}+\epsilon \delta\tilde{\varphi},S+\epsilon \delta S) \approx 0,
\end{align}
where $\delta\tilde{\varphi}(x,\theta)$ and $\delta S(x)$ are arbitrary functions that vanish when $x\in \partial U$. That is, the NL-WKB extension of $(M,\mathcal{C},\mathcal{L})$ at least approximately satisfies a variational principle. 

Somewhat surprisingly, the result suggested by the previous heuristic argument is correct. The extended action functional provides the NL-WKB extension of an ordinary first-order classical field theory with the following variational formulation.

\begin{definition}\label{looping_of_field_theory}
Given an ordinary first-order classical field theory $(M,\mathcal{C},\mathcal{L})$, the \emph{looping} of $(M,\mathcal{C},\mathcal{L})$ is the first-order classical theory $(\tilde{M},\tilde{\mathcal{C}},\tilde{\mathcal{L}})$ prescribed as follows.
\begin{itemize}
\item The \emph{looped spacetime} $\tilde{M}=M\times S^1$ is the trivial $S^1$ bundle over $M$.
\item The \emph{looped space of fields} $\tilde{\mathcal{C}}$ comprises maps $\Phi:\tilde{M}\rightarrow \tilde{F}$ of the form $\Phi(x,\theta) = (\tilde{\varphi}(x,\theta),S(x))$ with $\tilde{\varphi}:\tilde{M}\rightarrow F$, $S:M\rightarrow S^1$, and $\tilde{F} = F\times S^1$.
\item Set $\tilde{D}=M_{(f+1)\times(m+1)}(\mathbb{R})$, the space of real-valued $(f+1)\times(m\times 1)$ matrices. Let $X=(x,\theta)\in\tilde{M}$ and $\Phi=(\tilde{\varphi},S)\in \tilde{F}$. The \emph{looped Lagrangian density} $\tilde{\mathcal{L}}:\tilde{M}\times \tilde{F}\times \tilde{D}\rightarrow\mathbb{R}$ is given by
\begin{align}\label{extended_lag_gen}
\tilde{\mathcal{L}}(X,\Phi,V) = \frac{1}{2\pi} \mathcal{L}(x,\tilde{\varphi},\tilde{v}+\zeta\kappa),
\end{align}
where the matrix $V\in \tilde{D}$ has the block structure
\begin{align}
V=\left(\begin{array}{cc}
\tilde{v} & \zeta\\
\kappa & \alpha
\end{array}\right),\quad \begin{array}{c} \tilde{v}\in M_{f\times m}(\mathbb{R})\\ \zeta\in M_{f\times 1}(\mathbb{R})\end{array}\quad\begin{array}{c}  \kappa\in M_{1\times m}(\mathbb{R})\\ \alpha\in\mathbb{R}.\end{array}
\end{align}
\end{itemize}
\end{definition}

\begin{remark}
The last component of $\Phi = (\tilde{\varphi},S)\in\tilde{\mathcal{C}}$, being a function of $x$ alone, cannot be localized near an arbitrary angle $\theta$. Thus, $\tilde{\mathcal{C}}$ does not contain all functions with compact support, meaning $(\tilde{M},\tilde{\mathcal{C}},\tilde{\mathcal{L}})$ is \emph{not} ordinary. The Euler-Lagrange equations therefore do not take the standard form \eqref{ele_general}. The appropriate modification of the Euler-Lagrange equations will be found in the process of proving the next theorem.
\end{remark}

\begin{theorem}[Whitham averaging]\label{whitham_thm}
Let $(\tilde{M},\tilde{\mathcal{C}},\tilde{\mathcal{L}})$ be the looping of the ordinary first-order classical field theory  $(M,\mathcal{C},\mathcal{L})$. If $\Phi = (\tilde{\varphi},S)$ is a solution of $(M,\mathcal{C},\mathcal{L})$'s NL-WKB extension, then $\Phi$ is a critical point of $(\tilde{M},\tilde{\mathcal{C}},\tilde{\mathcal{L}})$'s local action functional $\tilde{A}_{U\times S^1}$ for each $U\subset M$. Conversely, if $\Phi= (\tilde{\varphi},S)\in\tilde{\mathcal{C}}$ is a critical point of $\tilde{A}_{U\times S^1}$ for each $U\subset M$, then $(\tilde{\varphi},S)$ is a solution of $(M,\mathcal{C},\mathcal{L})$'s NL-WKB extension.
\end{theorem}

\begin{proof}
First suppose that $\Phi=(\tilde{\varphi},S)\in\tilde{\mathcal{C}}$ is a critical point of $\tilde{A}_{U\times S^1}$ for each $U\subset M$. Introduce the indices $\tilde{a}\in\{1,\dots,f+1\}$, $\tilde{\mu}\in\{1,\dots,m+1\}$, as well as the shorthand notation
\begin{align}
J(X) = (X,\Phi(X),\partial\Phi(X))\in \tilde{M}\times \tilde{F}\times\tilde{D}.
\end{align} 
(We refer the reader to the text below Eq.\,\eqref{eq_16} for the definition of $\partial\Phi$.) Then it must be true that
\begin{align}\label{critical}
\int_{U\times S^1} \left[\frac{\partial\tilde{\mathcal{L}}}{\partial \Phi^{\tilde{a}}}(J(X))-\frac{\partial}{\partial X^{\tilde{\mu}}}\left(\frac{\partial\tilde{\mathcal{L}}}{\partial V_{\tilde{\mu}}^{\tilde{a}}}(J(X))\right)\right]\delta\Phi^{\tilde{a}}\,dX = 0,
\end{align}
for all $\delta\Phi\in \tilde{\mathcal{C}}$ that vanish on $\partial U\times S^1$. Because $\delta\Phi^{a} = \delta\tilde{\varphi}^{a}$ when $a\in\{1,\dots,f\}$ and $\delta\tilde{\varphi}(x,\theta)$ is arbitrary away from the boundary, Eq.\,\eqref{critical} implies
\begin{align}\label{extended_ele_x}
\frac{\partial\tilde{\mathcal{L}}}{\partial \tilde{\varphi}^{a}}(J(X))=\frac{\partial}{\partial x^{\mu}}\left(\frac{\partial\tilde{\mathcal{L}}}{\partial \tilde{v}_{\mu}^{a}}(J(X))\right)+\frac{\partial}{\partial\theta}\left(\frac{\partial\tilde{\mathcal{L}}}{\partial \zeta^a}(J(X))\right).
\end{align}
Likewise, because $\delta\Phi^{f+1} = \delta S$ and $\delta S(x)$ is arbitrary away from the boundary, we have
\begin{align}\label{extended_ele_s}
\frac{\partial}{\partial x^{\mu}}\int_{S^1}\left(\frac{\partial\tilde{\mathcal{L}}}{\partial \kappa_{\mu}}(J(X))\right)\,d\theta = 0.
\end{align} 
Here we have used $\frac{\partial\tilde{\mathcal{L}}}{\partial S}=0$, which follows from Eq.\,\eqref{extended_lag_gen}. We will refer to Eqs.\,\eqref{extended_ele_x} and \eqref{extended_ele_s} as the Euler-Lagrange equations associated with the looping $(\tilde{M},\tilde{\mathcal{C}},\tilde{\mathcal{L}})$, or the \emph{looped Euler-Lagrange equations}. By reading the previous argument in reverse, we see that $\Phi$ is a critical point of $\tilde{A}_{U\times S^1}$ for all $U\subset M$ if and only if $\Phi$ satisfies the looped Euler-Lagrange equations.

The derivatives of $\tilde{\mathcal{L}}$ that appear in Eqs.\,\eqref{extended_ele_x} and \eqref{extended_ele_s} may be expressed in terms of derivatives of $\mathcal{L}$ using the definition \eqref{extended_lag_gen} according to
\begin{align}
\frac{\partial\tilde{\mathcal{L}}}{\partial \tilde{\varphi}^a}(J(X))& =\frac{1}{2\pi} \frac{\partial\mathcal{L}}{\partial \varphi^a}(j(x,\theta))\\
\frac{\partial\tilde{\mathcal{L}}}{\partial \tilde{v}^a_\mu}(J(X)) &=\frac{1}{2\pi} \frac{\partial\mathcal{L}}{\partial v^a_\mu}(j(x,\theta))\\
\frac{\partial\tilde{\mathcal{L}}}{\partial \zeta^a}(J(X)) &=\frac{1}{2\pi}\partial_\mu S(x)\frac{\partial\mathcal{L}}{\partial v^a_\mu}(j(x,\theta))\\
\frac{\partial\tilde{\mathcal{L}}}{\partial \kappa_\mu}(J(X)) &= \frac{1}{2\pi}\frac{\partial\tilde{\varphi}^a}{\partial\theta}(x,\theta)\frac{\partial\mathcal{L}}{\partial v^a_\mu}(j(x,\theta)).
\end{align}
Therefore the looped Euler-Lagrange equations are equivalent to
\begin{gather}
 \frac{\partial\mathcal{L}}{\partial \varphi^a}(j(x,\theta)) = \left(\frac{\partial}{\partial x^\mu}+\partial_\mu S(x)\frac{\partial}{\partial\theta}\right)\left(\frac{\partial\mathcal{L}}{\partial v^a_\mu}(j(x,\theta))\right)\label{looped_ele_x}\\
 \frac{\partial}{\partial x^\mu} \int_{S^1}\frac{\partial\tilde{\varphi}^a}{\partial\theta}(x,\theta)\frac{\partial\mathcal{L}}{\partial v^a_\mu}(j(x,\theta))\,d\theta = 0,\label{looped_ele_s}
\end{gather}
where $j(x,\theta)$ was defined in Eq.\,\eqref{shorthand_j}. In particular, Eq.\,\eqref{looped_ele_x} reproduces the NL-WKB extension of $(M,\mathcal{C},\mathcal{L})$, i.e. Eq.\,\eqref{pre_extension}. This proves that each $\Phi$ that is a critical point of $\tilde{A}_{U\times S^1}$ for all $U\subset M$ is also a solution of the NL-WKB extension of $(M,\mathcal{C},\mathcal{L})$.

Now suppose conversely that $\Phi = (\tilde{\varphi},S)$ is a solution of the NL-WKB extension of $(M,\mathcal{C},\mathcal{L})$. The first of the looped Euler-Lagrange equations, i.e. Eq.\,\eqref{looped_ele_x}, is then clearly satisfied. However it is not immediately clear that the second equation \eqref{looped_ele_s} is also satisfied. To establish Eq.\,\eqref{looped_ele_s} first note that we have the following identity:
\begin{align}
\frac{\partial}{\partial\theta} \mathcal{L}(j(x,\theta)) =& \frac{\partial\tilde{\varphi}^a}{\partial\theta}(x,\theta) \frac{\partial\mathcal{L}}{\partial \varphi^a}(j(x,\theta))+\frac{\partial}{\partial\theta}\left(\frac{\partial\tilde{\varphi}^a}{\partial x^\mu}(x,\theta)+\partial_\mu S(x)\frac{\partial\tilde{\varphi}^a}{\partial\theta}\right)\frac{\partial\mathcal{L}}{\partial v^a_\mu}(j(x,\theta))\nonumber\\
=&\frac{\partial\tilde{\varphi}^a}{\partial\theta}(x,\theta) \left(\frac{\partial}{\partial x^\mu}+\partial_\mu S(x)\frac{\partial}{\partial\theta}\right)\left(\frac{\partial\mathcal{L}}{\partial v^a_\mu}(j(x,\theta))\right)\nonumber\\
&+\frac{\partial}{\partial\theta}\left(\frac{\partial\tilde{\varphi}^a}{\partial x^\mu}(x,\theta)+\partial_\mu S(x)\frac{\partial\tilde{\varphi}^a}{\partial\theta}\right)\frac{\partial\mathcal{L}}{\partial v^a_\mu}(j(x,\theta)),\label{proof_id_1}
\end{align}
where we have used the chain rule on the first line, and the NL-WKB extension of $(M,\mathcal{C},\mathcal{L})$ (cf. Eq.\,\eqref{pre_extension}) on the second line. Next integrate Eq.\,\eqref{proof_id_1} over $S^1$ and apply integration by parts as follows:
\begin{align}
0 = &\int_{S^1} \frac{\partial\tilde{\varphi}^a}{\partial\theta}(x,\theta) \left(\frac{\partial}{\partial x^\mu}+S_{,\mu}(x)\frac{\partial}{\partial\theta}\right)\left(\frac{\partial\mathcal{L}}{\partial v^a_\mu}(j(x,\theta))\right)\,d\theta\nonumber\\
&+\int_{S^1}\frac{\partial}{\partial\theta}\left(\frac{\partial\tilde{\varphi}^a}{\partial x^\mu}(x,\theta)+S_{,\mu}(x)\frac{\partial\tilde{\varphi}^a}{\partial\theta}\right)\frac{\partial\mathcal{L}}{\partial v^a_\mu}(j(x,\theta))\,d\theta\nonumber\\
=&\int_{S^1} \frac{\partial\tilde{\varphi}^a}{\partial\theta}(x,\theta) \left(\frac{\partial}{\partial x^\mu}+S_{,\mu}(x)\frac{\partial}{\partial\theta}\right)\left(\frac{\partial\mathcal{L}}{\partial v^a_\mu}(j(x,\theta))\right)\,d\theta\nonumber\\
&-\int_{S^1}\left(\frac{\partial\tilde{\varphi}^a}{\partial x^\mu}(x,\theta)+S_{,\mu}(x)\frac{\partial\tilde{\varphi}^a}{\partial\theta}\right)\frac{\partial}{\partial\theta}\frac{\partial\mathcal{L}}{\partial v^a_\mu}(j(x,\theta))\,d\theta\nonumber\\
=&\int_{S^1} \frac{\partial\tilde{\varphi}^a}{\partial\theta}(x,\theta) \frac{\partial}{\partial x^\mu}\left(\frac{\partial\mathcal{L}}{\partial v^a_\mu}(j(x,\theta))\right)\,d\theta-\int_{S^1}\frac{\partial\tilde{\varphi}^a}{\partial x^\mu}(x,\theta)\frac{\partial}{\partial\theta}\frac{\partial\mathcal{L}}{\partial v^a_\mu}(j(x,\theta))\,d\theta\nonumber\\
=&\frac{\partial}{\partial x^\mu}\int_{S^1} \frac{\partial\tilde{\varphi}^a}{\partial\theta}(x,\theta) \frac{\partial\mathcal{L}}{\partial v^a_\mu}(j(x,\theta))\,d\theta - \int_{S^1} \frac{\partial^2\tilde{\varphi}^a}{\partial x^\mu\partial\theta}(x,\theta) \left(\frac{\partial\mathcal{L}}{\partial v^a_\mu}(j(x,\theta))\right)\,d\theta\nonumber\\
&+\int_{S^1}\frac{\partial^2\tilde{\varphi}^a}{\partial x^\mu\partial\theta}(x,\theta)\frac{\partial\mathcal{L}}{\partial v^a_\mu}(j(x,\theta))\,d\theta\nonumber\\
=&\frac{\partial}{\partial x^\mu}\int_{S^1} \frac{\partial\tilde{\varphi}^a}{\partial\theta}(x,\theta) \frac{\partial\mathcal{L}}{\partial v^a_\mu}(j(x,\theta))\,d\theta,
\end{align}
where have used the equality of mixed partial derivatives. This completes the proof. 
\end{proof}
\begin{remark}
The preceding proof demonstrated the presence of a redundancy in the looped Euler-Lagrange equations \eqref{looped_ele_x}-\eqref{looped_ele_s}. The reason for this redundancy is the presence of a gauge symmetry. The gauge group is given by smooth functions $\psi:M\rightarrow S^1$ with addition as the group composition law. The action of $\psi$ on $(\tilde{\varphi},S)$ is given by $\psi\cdot (\tilde{\varphi},S) = (\tilde{\varphi}^\prime,S^\prime)$, where
\begin{align}
\tilde{\varphi}^\prime(x,\theta) =& \tilde{\varphi}(x,\theta-\psi(x))\\
S^\prime(x)=&S(x)+\psi(x).
\end{align}
The redundancy in the looped Euler-Lagrange equations, i.e. the fact that \Eq{looped_ele_x} implies \Eq{looped_ele_s}, may be seen as a consequence of gauge symmetry by applying Noether's second theorem. The presence of this gauge symmetry could have been anticipated by noting that the nonlinear WKB ansatz does not uniquely specify the phase function $S$.
\end{remark}

\section{Euler-Poincar\'e fluids as classical field theories\label{sec:three}}
While the results from Section \ref{sec:two} are useful for identifying variational principles that govern the nonlinear WKB extension of a large class of dissipation-free PDE, they are not immediately applicable to many of the PDEs that appear in fluid dynamics. In particular, they cannot be applied directly to the fluid-mechanical PDEs that arise from Euler-Poincar\'e variational principles.\cite{Holm_1998} The essential issue is that, as we will review, Euler-Poincar\'e variational principles do not fit into the mold of classical field theory. The purpose of this section is to construct an alternative variational principle for Euler-Poincar\'e fluid equations to which the Whitham averaging (i.e. Theorem \ref{whitham_thm}) can be profitably applied. We will show how to use Whitham averaging to identify a variational principle for the NL-WKB extension of Euler-Poincar\'e fluid equations in the following section.

We will restrict our attention to a large subclass of Euler-Poincar\'e fluid equations defined as follows.

\begin{definition}[LBEP equations]
Given a function $\mathcal{L}_{\text{EP}}:\mathbb{R}^3\times\mathbb{R}\times\mathbb{R}^3\rightarrow \mathbb{R}:(\bm{u},\rho,\del\rho)\mapsto \mathcal{L}(\bm{u},\rho,\del\rho)$ such that $\bm{u}\mapsto (\partial\mathcal{L}_{\text{EP}}/\partial\bm{u})(\bm{u},\rho,\del\rho)$ is a diffeomorphism for each $(\rho,\del\rho)\in\mathbb{R}\times\mathbb{R}^3$, the associated \emph{local barotropic Euler-Poincar\'e fluid equations} (LBEP equations) are the system of PDEs
\begin{gather}
\partial_t\rho+\del\cdot(\rho\bm{u}) = 0\\
\partial_t\frac{\partial\mathcal{L}_{\text{EP}}}{\partial\bm{u}}+\del\cdot\left(\bm{u}\otimes \frac{\partial\mathcal{L}_{\text{EP}}}{\partial\bm{u}}\right)  \nonumber\\=\del\cdot\left(\left[\rho\frac{\partial\mathcal{L}_{\text{EP}}}{\partial\rho}-\rho\del\cdot\frac{\partial\mathcal{L}_{\text{EP}}}{\partial\del\rho} - \mathcal{L}_{\text{EP}}\right]\mathbb{I}+\frac{\partial\mathcal{L}_{\text{EP}}}{\partial\del\rho}\otimes\del\rho \right),\label{mom_eqn_LBEP_def}
\end{gather} 
where the unknown fields are $(\rho(\bm{x},t),\bm{u}(\bm{x},t))\in\mathbb{R}\times\mathbb{R}^3$ and all derivatives of $\mathcal{L}_{\text{EP}}$ are evaluated at$(\bm{u}(\bm{x},t),\rho(\bm{x},t),\del\rho(\bm{x},t))$. The function $\mathcal{L}_{\text{EP}}$ is called the \emph{Euler-Poincar\'e Lagrange density} and our convention for the tensor divergence is $(\del\cdot T)_j = \partial_i T_{ij}$. We have also introduced the notation $\otimes$ for the point-wise tensor product, i.e. the tensor product over the ring $C^\infty(Q)$.
\end{definition}
\noindent Upon introducing the EP Hamiltonian density $\mathcal{H}_{\text{EP}}$, the LBEP equations may also be conveniently written in terms of the momentum density $\bm{p} = \partial\mathcal{L}_{\text{EP}}/\partial\bm{u}$ as follows.
\begin{definition}
Let $\underline{\bm{u}}:\mathbb{R}^3\times\mathbb{R}\times\mathbb{R}^3\rightarrow\mathbb{R}^3$ be defined implicitly by the formula
\begin{align}
\frac{\pd \mc{L}_{\text{EP}}}{\pd \bm{u}}(\underline{\bm{u}}(\bm{p},\rho,\del\rho),\rho,\del\rho) = \bm{p}.\label{u_underline}
\end{align} 
The \emph{EP Hamiltonian density} $\mathcal{H}_{\text{EP}}:\mathbb{R}^3\times\mathbb{R}\times \mathbb{R}^3\rightarrow \mathbb{R}$ is defined by
\begin{gather}
	\mathcal{H}_{\text{EP}}(\vec{p},\rho, \del \rho) 
		\doteq  \underline{\vec{u}}(\bm{p},\rho,\del\rho)\cdot \vec{p} 
					- \mathcal{L}_{\text{EP}} (\underline{\vec{u}}(\bm{p},\rho,\del\rho),\rho, \del \rho).\label{eq:EP:Ham}
\end{gather}
\end{definition}
\begin{lemma}[LBEP equations, momentum form]\label{mLBEP_eqs}
The LBEP equations for the unknown fields $(\rho(\bm{x},t),\bm{u}(\bm{x},t))\in\mathbb{R}\times\mathbb{R}^3$ are equivalent to the following system of PDEs for the unknown fields $(\rho(\bm{x},t),\bm{p}(\bm{x},t))\in\mathbb{R}\times\mathbb{R}^3$:
\begin{gather}
\partial_t\rho+\del\cdot\left(\rho\frac{\pd\mc{H}_{\text{EP}}}{\pd \bm{p}}\right) = 0\label{ham_cont_def}\\
\partial_t\vec{p} 
		+ 	\del \cdot \left( \frac{\pd\mc{H}_{\text{EP}}}{\pd\bm{p}}\otimes \bm{p} +  \frac{\pd \mathcal{H}_{\text{EP}}}{ \pd \del \rho}  \otimes \del \rho \right)\nonumber\\ 
		=	-\del \left( \rho   \frac{\pd \mathcal{H}_{\text{EP}}}{\pd \rho}    
				-  \rho\del \cdot \frac{\pd \mathcal{H}_{\text{EP}}}{ \pd \del \rho} 
				+ 	\bm{p}\cdot\frac{\partial\mathcal{H}_{\text{EP}}}{\partial\bm{p}}-\mathcal{H}_{\text{EP}} \right).\label{ham_mom_def}
\end{gather}
We will refer to Eqs.\,\eqref{ham_cont_def} and \eqref{ham_mom_def} as the \emph{momentum form} of the LBEP equations, or mLBEP equations for brevity.
\end{lemma}
\begin{proof}
By differentiating the definition \eqref{eq:EP:Ham} of $\mc{H}_{\text{EP}}$ and substituting the definition \eqref{u_underline} of $\underline{\bm{u}}$, we obtain
\begin{align}
\frac{\pd\mc{H}_{\text{EP}}}{\pd \bm{p}} (\bm{p},\rho,\del\rho)&= \underline{\bm{u}}(\bm{p},\rho,\del\rho)\label{eq:implicit_diff_first}\\
\frac{\pd\mc{H}_{\text{EP}}}{\pd\rho}(\bm{p},\rho,\del\rho) &= -\frac{\pd\mc{L}_{\text{EP}}}{\pd\rho}(\underline{\bm{u}}(\bm{p},\rho,\del\rho),\rho,\del\rho) \\
\frac{\pd\mc{H}_{\text{EP}}}{\pd\del\rho}(\bm{p},\rho,\del\rho) &= -\frac{\pd\mc{L}_{\text{EP}}}{\pd\del\rho}(\underline{\bm{u}}(\bm{p},\rho,\del\rho),\rho,\del\rho).\label{eq:implicit_diff_last}
\end{align}
If $(\rho,\bm{p})$ is a solution of the mLBEP equations, then the identities \eqref{eq:implicit_diff_first}-\eqref{eq:implicit_diff_last} imply that $(\rho,\underline{\bm{u}}(\bm{p},\rho,\del\rho))$ is a solution of the LBEP equations. This shows that $I:(\rho,\bm{p})\mapsto (\rho,\underline{\bm{u}}(\bm{p},\rho,\del\rho))$ maps solutions of the mLBEP equations into solutions of the LBEP equations. If $(\rho,\bm{u})$ is a solution of the LBEP equations, then the identities \eqref{eq:implicit_diff_first}-\eqref{eq:implicit_diff_last} imply $(\rho,\pd\mc{L}_{\text{EP}}/\pd\bm{u}(\bm{u},\rho,\del\rho))$ is a solution of mLBEP equations. This shows that the mapping $I$ is surjective. Injectivity of $I$ follows from the hypothesis that $\bm{u}\mapsto \pd\mc{L}_{\text{EP}}/\pd \bm{u}$ is a diffeomorphism. The mapping $I$ therefore establishes a bijection between solutions of the LBEP equations and solutions of the mLBEP equations.
\end{proof}

 We aim to identify an ordinary first-order classical field theory whose associated Euler-Lagrange equations reproduce the LBEP equations. 
In order to illustrate why this task is non-trivial, let us briefly review the Euler-Poincar\'e variational formulation of the LBEP equations described in Ref.\,\onlinecite{Holm_1998}. The basic idea is to introduce the space of Lagrangian configuration maps $\bm{g}:Q_0\rightarrow Q$. A Lagrangian configuration map is a diffeomorphism that assigns to each \emph{particle label} $\bm{x}_0\in Q_0$ the current Eulerian position of that fluid particle $\bm{x} = \bm{g}(\bm{x}_0)$. The space $Q_0$ is referred to as the label space, while the space $Q$ is the fluid container. 
Set $Q=(S^1)^3$, fix a positive function $\rho_0:Q_0\rightarrow \mathbb{R}$, and consider the action
\begin{equation}
	\mc{A}_{\rho_0}(\bm{g})=\int_{t_1}^{t_2}  L_{\rho_0}(\bm{g}(t),\dot{\bm{g}}(t))\, \mathrm{d}t,
	\label{eq:EP:act_original}
\end{equation}
where $\bm{g}:[t_1,t_2]\rightarrow \text{Diff}(Q_0,Q)$. Define the fluid velocity $\bm{v}$ and the mass density $\rho$ according to
\begin{align}
\bm{v}(\bm{x}) & = \dot{\bm{g}}(\bm{g}^{-1}(\bm{x}))\\
\rho(\bm{x}) &= \frac{1}{\text{det}(\del_0\bm{g})(\bm{g}^{-1}(\bm{x}))}\rho_0(\bm{g}^{-1}(\bm{x})).
\end{align}
When the Lagrangian $L_{\rho_0}:T\text{Diff}(Q_0,Q)\rightarrow \mathbb{R}$ is given by
\begin{align}
L_{\rho_0}(\bm{g},\dot{\bm{g}}) = \int_Q \mathcal{L}_{\text{EP}}(\bm{v}(\bm{x}),\rho(\bm{x}),\del\rho(\bm{x}))\,d^3\bm{x},
\end{align}
we can establish a close relationship between the LBEP equations and the Euler-Lagrange equations associated with $L_{\rho_0}$. To see this, observe first that
because $\vec{v}$ and $\rho$ are defined in terms of the configuration map $\bm{g}$, they cannot be varied independently. Instead, variations of $\bm{g}$ \emph{induce} variations of $\bm{v}$ and $\rho$ as follows. Given a variation $\delta\bm{g}$ of $\bm{g}$, we may construct an ``Eulerianized" variation $\bm{\xi} = \delta\bm{g}\circ \bm{g}^{-1}$. The variations of $\bm{v}$ and $\rho$ may then be computed as
\begin{gather}
	\delta\vec{v}= \dot{\vec{\xi}} + ( \vec{v} \cdot \del ) \vec{\xi} - (\vec{\xi} \cdot \del)  \vec{v},
	\label{eq:EP:variation_v}  \\
	\delta \rho = - \del \cdot ( \rho\vec{\xi} ) .
	\label{eq:EP:variation_rho}
\end{gather}
The Euler--Lagrange (EL) equations may therefore be obtained by varying the action $\mc{A}_{\rho_0}$ with respect to $t\mapsto \bm{g}(t)$ and making judicious use of the induced variation formulas \eqref{eq:EP:variation_v}-\eqref{eq:EP:variation_rho}. This leads to
\begin{equation}
	0=	\delta \mc{A}_{\rho_0}
		=	\int_{t_1}^{t_2} \int_Q \left( 
		\frac{\partial \mathcal{L}_{\text{EP}}}{\partial \vec{v}} \cdot \delta \vec{v}
		+ \frac{\partial \mathcal{L}_{\text{EP}}}{\partial \rho} \delta \rho
		+ \frac{\partial \mathcal{L}_{\text{EP}}}{ \partial \del \rho} \cdot \del \delta \rho \right)   \, \mathrm{d}^3 \vec{x} \, \mathrm{d}t\,.
	\label{eq:EP:variation_action}
\end{equation}
Substituting \Eqs{eq:EP:variation_v} and \eq{eq:EP:variation_rho} and integrating by parts leads to
\begin{equation}
	\frac{\pd}{\pd t} \vec{p} + \frac{\pd}{\pd x^i }( v^i \vec{p}) + p_i \del v^i  
		= \rho \del \left( \frac{\partial \mathcal{L}_{\text{EP}}}{\partial \rho} - \del \cdot \frac{\partial \mathcal{L}_{\text{EP}}}{ \partial \del \rho}  \right),
	\label{eq:EP:variation_mom}
\end{equation}
where $\vec{p}$ is the momentum density
\begin{equation}
	\vec{p} \doteq \frac{\partial \mathcal{L}_{\text{EP}}}{\partial \vec{v}}(\bm{v},\rho,\del\rho).
	\label{eq:EP:mom}
\end{equation}
Equation \eq{eq:EP:variation_mom} is a generalized momentum conservation equation. It can be written in conservative form as well. A simple calculation leads to
\begin{equation}
	\frac{\pd}{\pd t} \vec{p} 
		+ 	\del \cdot \left(  \vec{v} \otimes \vec{p} -  \frac{\partial \mathcal{L}_{\text{EP}}}{ \partial \del \rho}  \otimes \del \rho \right) 
		=	\del \left( \rho   \frac{\partial \mathcal{L}_{\text{EP}}}{\partial \rho}    
				-  \rho\del \cdot \frac{\partial \mathcal{L}_{\text{EP}}}{ \partial \del \rho} 
				- 	\mc{L}_{\text{EP}} \right).
	\label{eq:EP_mom_conservative}
\end{equation}
Note that by substituting Eq.\,\eqref{eq:EP:mom} and identifying $\vec{v}$ with $\vec{u}$, Eq.\,\eqref{eq:EP_mom_conservative} becomes the momentum equation \eqref{mom_eqn_LBEP_def} of the LBEP equations. In addition, since the fluid mass density $\rho$ is defined in terms of $\bm{g}$ by \Eq{density_advection}, it satisfies by construction the continuity equation
\begin{equation}
	\pd_t \rho + \del \cdot ( \vec{v} \rho ) =0.
	\label{eq:EP:continuity_aux}
\end{equation}
In particular, Eq.\,\eqref{eq:EP:continuity_aux} is \emph{not} a consequence of the Euler-Lagrange equations. It follows that the Euler-Lagrange equations associated with the action $\mathcal{A}_{\rho_0}$ may be regarded as a second-order ordinary differential equation in the variable $\bm{g}\in \text{Diff}(Q_0,Q)$, or, equivalently, a first-order ordinary differential equation in the variables $(\bm{g},\bm{u})\in \text{Diff}(Q_0,Q)\times \mathfrak{X}(Q)\approx T\text{Diff}(Q_0,Q)$; the evolution equation for $\bm{g}$ is $\bm{v}=\dot{\bm{g}}\circ\bm{g}^{-1} = \bm{u}$ and the evolution equation for $\bm{u}$ is given by substituting $\bm{v} = \bm{u}$ in Eq.\,\eqref{eq:EP_mom_conservative}. Note in particular that if $t\mapsto (\bm{g}(t),\bm{u}(t))$ is a solution of the Euler-Lagrange equations, then $t\mapsto (\bm{u}(t),\rho(t))$ is a solution of the LBEP equations when Eq.\,\eqref{eq:EP:continuity_aux} is used to define the mass density $\rho$.

Following the discussion in \Sec{sec:two}, we would like to find a variational formulation for the NL-WKB extension of the LBEP equations. Whitham averaging would seem to be a natural tool for this task. However, there are certain idiosyncrasies of the EP action principle $\delta\mc{A}_{\rho_0}=0$ that prevent us from directly applying the Whitham averaging theorem (Theorem \ref{whitham_thm}). These are the following. First, the Lagrangian $L_{\rho_0}$ depends on the configuration map $\bm{g}$, whose domain is the label space $Q_0\neq Q$. 
In contrast, Theorem \ref{whitham_thm} applies to Lagrangians defined on spaces of fields over spacetime $M = Q\times\mathbb{R}$. Second, for any given $\rho_0$, not \emph{all} of the solutions of the LBEP can be recovered from the EP action principle $\delta\mathcal{A}_{\rho_0} = 0$.  Indeed, while the space of solutions of the LBEP equations on $Q= (S^1)^3$ is parameterized by initial data $(\rho,\bm{u})\in C_+^\infty(Q)\times \mathfrak{X}(Q)$, solutions of the Euler-Lagrange equations associated with $\mathcal{A}_{\rho_0}$ cannot accommodate all possible initial $\rho$.  In fact, solutions of the Euler-Lagrange equations for $\bm{g}$ can only recover solutions of the LBEP equations with initial $\rho$ that may be related to the parameter $\rho_0$ by some diffeomorphism $\bm{g}:Q_0\rightarrow Q$ using the formula \eqref{density_advection}. Note in particular that initial $\rho$ with $\int \rho\,d^3\bm{x}\neq \int \rho_0 \,d^3\bm{x}_0$ cannot be obtained in this manner. In contrast, the PDEs addressed by the Whitham averaging theorem all have the property that solutions of the PDE are precisely the critical points of a single action functional, rather than a family of action functionals like $\mathcal{A}_{\rho_0}$. Finally, and most superficially, since the mass density $\rho$ is defined using the Jacobian of the mapping $\bm{g}$ (see \Eq{density_advection}), second derivatives of $\bm{g}$ appear in the Lagrangian $L_{\rho_0}$. This suggests that theory of Whitham averaging described in Section \ref{sec:two} for first-order field theories cannot be applied.

In order to eventually bring Whitham averaging to bear on the problem of NL-WKB extension of the LBEP equations, we will construct an alternative variational formulation of the LBEP equations that fits into the framework of first-order classical field theory. We proceed as follows. (1) First, let us introduce the inverse of the configuration map, $\bm{h}\doteq  \bm{g}^{-1}$, which is also known as the \textit{back-to-labels map}. Conveniently, $\bm{h}$ is a mapping from the spatial domain $Q$ to the label space $Q_0$, and may therefore be regarded as a $Q_0$-valued field. The motivation here is that while the ansatz $\bm{g}(\bm{x}_0)=\tilde{\bm{g}}(\bm{x}_0,S(\bm{x}_0))$ requires evaluating the phase function on points in the label space, the ansatz $\bm{h}(\bm{x}) = \tilde{\bm{h}}(\bm{x},S(\bm{x}))$ does not.
We then substitute $\bm{g} = \bm{h}^{-1}$ into the action \Eq{eq:EP:act_original}. The velocity field $\vec{v}$, which was originally given by $\vec{v} \doteq \dot{\bm{g}} \circ \bm{g}^{-1}$, is now written in terms of $\bm{h}$ as
\begin{equation}
	\vec{v} = - \dot{\bm{h}} \cdot  (\del \bm{h})^{-1}  .
	\label{eq:EP_velocity_h}
\end{equation}
It can be shown that variations of $\vec{v}$ and $\rho$ with respect to $\bm{h}$ are still given by \Eq{eq:EP:variation_v} and \Eq{eq:EP:variation_rho}, but now the vector field $\vec{\xi}$ is written as $\vec{\xi} = - \delta \bm{h} \cdot  (\del \bm{h})^{-1}$. (2) Next we  construct the parameter-dependent \emph{phase space Lagrangian} $\mathsf{L}_{\rho_0}(\bm{h},\dot{\bm{h}},\bm{p},\dot{\bm{p}})$ given by
\begin{align}
\mathsf{L}_{\rho_0}(\bm{h},\dot{\bm{h}},\bm{p},\dot{\bm{p}}) = \int_Q \bm{p}\cdot\bm{v}\,d^3\bm{x} - \int_Q \mathcal{H}_{\text{EP}}(\bm{p}(\bm{x}),\rho(\bm{x}),\del\rho(\bm{x}))\,d^3\bm{x}.
\end{align}
The associated parameter-dependent phase space action functional, 
\begin{align}
\mathsf{A}_{\rho_0}(\bm{h},\bm{p}) =\int_{t_1}^{t_2} \mathsf{L}_{\rho_0}(\bm{h}(t),\dot{\bm{h}}(t),\bm{p}(t),\dot{\bm{p}}(t))\,dt,
\end{align}
is defined on the space of paths $[t_1,t_2]\rightarrow \text{Diff}(Q,Q_0)\times \mathfrak{X}(Q)$. This implies that variations are to be applied to $\bm{h}$ and $\bm{p}$ independently while holding the values of $\bm{h}$ and $\bm{p}$ fixed at $t_1$ and $t_2$. (3) Finally, we introduce a scalar function $\chi:Q\rightarrow\mathbb{R}$ as a Lagrange multiplier that enforces the continuity equation as in Section 4.2 of Ref.\,\onlinecite{Cotter_Holm_2012}. This leads to the parameter-independent phase space Lagrangian $\mathsf{L}(\bm{h},\dot{\bm{h}},\bm{p},\dot{\bm{p}},\rho,\dot{\rho},\chi,\dot{\chi})$ given by
\begin{align}
\mathsf{L}(\bm{h},\dot{\bm{h}},\bm{p},\dot{\bm{p}},\rho,\dot{\rho},\chi,\dot{\chi}) = \int_Q \bm{p}\cdot\bm{v}\,d^3\bm{x} + \int_Q (\dot{\chi}+\bm{v}\cdot\del\chi)\,\rho\,d^3\bm{x}-\int_Q \mathcal{H}_{\text{EP}}(\bm{p},\rho,\del\rho)\,d^3\bm{x},\label{psl_no_param}
\end{align}
where the velocity $\bm{v}$ is defined in terms of $\bm{h}$ as in \Eq{eq:EP_velocity_h}. The Lagrangian $\mathsf{L}$ is intrinsically a function on $T\mathcal{C}_0$, where $\mathcal{C}_0$ the space of \emph{frozen field configurations}.
\begin{definition}\label{frozen_field_def}
The space of frozen field configurations is the infinite-dimensional manifold $\mathcal{C}_0 = \text{Diff}(Q,Q_0)\times \mathfrak{X}(Q)\times C^\infty_+(Q)\times C^\infty(Q)$, where $C^\infty_+(Q)$ is the set of smooth positive functions on $Q$, and $\mathfrak{X}(Q)$ is the set of vector fields on $Q$. That is, $\mc{C}_0$ comprises maps $Q\ni \bm{x}\mapsto (\bm{h}(\bm{x}),\bm{p}(\bm{x}),\rho(\bm{x}),\chi(\bm{x}))\in Q_0\times\mathbb{R}^3\times\mathbb{R}\times\mathbb{R}$, where $\bm{h}$ is a diffeomorphism and $\rho(\bm{x})>0$ for all $\bm{x}\in Q$.
\end{definition} 
\noindent Correspondingly, the parameter-independent phase space action functional,
\begin{align}
\mathsf{A}(\bm{h},\bm{p},\rho,\chi) = \int_{t_1}^{t_2}\mathsf{L}(\bm{h}(t),\dot{\bm{h}}(t),\bm{p}(t),\dot{\bm{p}}(t),\rho(t),\dot{\rho}(t),\chi(t),\dot{\chi}(t))\,dt,\label{eq:EP:act}
\end{align}
is defined on the space of paths $[t_1,t_2]\rightarrow \mathcal{C}_0$, which implies that variations should be applied to $\bm{h}$, $\bm{p}$, $\rho$, and $\chi$ independently.

%

The following proposition shows that the Euler-Lagrange equations associated with $\mathsf{A}$ define a system of PDEs that completely recover the LBEP equations.

\begin{proposition}\label{extended_vp}
A path $t\mapsto (\bm{h}(t),\bm{p}(t),\rho(t),\chi(t))\in \mathcal{C}_0$ is a critical point of the action functional $\mathsf{A}$ in Eq.\,\eqref{eq:EP:act} if and only $\bm{h}$, $\bm{p}$, $\rho$, and $\chi$ satisfy the following system of PDEs:
\begin{gather}
\partial_t\bm{h}=-\frac{\partial\mathcal{H}_{\text{EP}}}{\partial\bm{p}}\cdot\del\bm{h}\label{eq:extended_first}\\
\partial_t\vec{p} 
		+ 	\del \cdot \left(\frac{\pd\mc{H}_{\text{EP}}}{\pd\bm{p}}\otimes \bm{p}  +  \frac{\pd \mathcal{H}_{\text{EP}}}{ \pd \del \rho}  \otimes \del \rho \right) \nonumber\\
		=	-\del \left( \rho   \frac{\pd \mathcal{H}_{\text{EP}}}{\pd \rho}    
				-  \rho\del \cdot \frac{\pd \mathcal{H}_{\text{EP}}}{ \pd \del \rho} 
				+ 	\bm{p}\cdot\frac{\partial\mathcal{H}_{\text{EP}}}{\partial\bm{p}}-\mathcal{H}_{\text{EP}} \right)\label{eq:ham_momentum}\\
				\partial_t\rho +\del\cdot\left(\rho\frac{\partial\mathcal{H}_{\text{EP}}}{\partial\bm{p}}\right)= 0\label{eq:ham_cont}\\
				\partial_t \chi + \frac{\pd\mc{H}_{\text{EP}}}{\pd\bm{p}} \cdot \del  \chi 
		=  \frac{\partial \mathcal{H}_{\text{EP}}}{\partial \rho}  
					-\del \cdot \frac{\partial \mc{H}_{\text{EP}}}{\partial \del \rho}.\label{eq:extended_last}
\end{gather}
Moreover, every solution $t\mapsto (\bm{u}(t),\rho(t))$ of the LBEP equations may be obtained from some solution $t\mapsto (\bm{h}(t),\bm{p}(t),\rho(t),\chi(t))$ of the Euler-Lagrange equations associated with $\mathsf{A}$ by defining $\bm{u}(t)$ using the Legendre transform
\begin{align}
\frac{\partial\mathcal{L}_{\text{EP}}}{\partial\bm{u}}(\bm{u}(\bm{x},t),\rho(\bm{x},t),\del\rho(\bm{x},t)) = \bm{p}(\bm{x},t).
\end{align}
\end{proposition}

\begin{proof}
The EL equations associated with the variational principle $\delta \mathsf{A} =0$ may be derived as follows. When varying the momentum density $\vec{p}$, one immediately finds
\begin{equation}
	\vec{v} = \frac{\partial \mathcal{H}_{\text{EP}}}{\partial \vec{p}}.
	\label{eq:EP_p}
\end{equation}
Thus, just as in finite-dimensional Hamiltonian systems, the fluid velocity is given by the partial derivative of the Hamiltonian with respect to the momentum density. Varying the action with respect to the scalar field $\chi$ leads to the continuity equation
\begin{equation}
	\pd_t \rho + \del \cdot ( \vec{v} \rho ) =0.
	\label{ep:EP_continuity}
\end{equation}
Varying the action with respect to $\rho$ gives
\begin{equation}
	\pd_t \chi + (\vec{v} \cdot \del)  \chi 
		=  \frac{\partial \mc{H}_{\text{EP}}}{\partial \rho}  
					-\del \cdot \frac{\partial \mc{H}_{\text{EP}}}{\partial \del \rho} .
	\label{eq:EP_psi}
\end{equation}
As before, since $\vec{v}$ depends on $\bm{h}$, the induced variations of $\bm{v}$ are given by \Eq{eq:EP:variation_v}. Therefore variations of $\bm{h}$ lead to
\begin{equation}
	\pd_t \vec{p}+ \del \cdot  (\vec{v} \otimes \vec{p}) + (\del \bm{v})\cdot\bm{p} 
		=	-\pd_t (\rho \del \chi) - \del \cdot ( \rho \vec{v}  \otimes \del \chi) - \rho ( \del \bm{v})\cdot \del\chi.
	\label{eq:EP_mom_aux}
\end{equation}

Equations \eq{eq:EP_p}--\eq{eq:EP_mom_aux} are the EL equations associated with the action \eq{eq:EP:act}. In order to see that they are equivalent to Eqs.\,\eqref{eq:extended_first}-\eqref{eq:extended_last}, first substitute \Eqs{ep:EP_continuity} and \eq{eq:EP_psi} into \Eq{eq:EP_mom_aux} in order to obtain
\begin{equation}
	\partial_t\vec{p} + \del\cdot (\vec{v} \otimes  \vec{p}) + ( \del \bm{v})\cdot\bm{p} 
		=	-\rho \del \left( 
				\frac{\partial \mathcal{H}_{\text{EP}}}{\partial \rho} 
				- \del \cdot \frac{\partial \mathcal{H}_{\text{EP}}}{\partial \del \rho}
			\right).
	\label{eq:EP_mom}
\end{equation}
Then use the identity 
\begin{align}
\del(\mathcal{H}_{\text{EP}}(\bm{p},\rho,\del\rho)) =& \del\cdot\left(\frac{\partial\mathcal{H}_{\text{EP}}}{\partial\del\rho}\otimes\del\rho+\mathbb{I}\left[\bm{p}\cdot\frac{\partial\mathcal{H}_{\text{EP}}}{\partial\bm{p}}+\rho\frac{\partial\mathcal{H}_{\text{EP}}}{\partial\rho} - \rho\del\cdot\frac{\partial\mathcal{H}_{\text{EP}}}{\partial\del\rho}\right]\right)\nonumber\\
&-\left(\del\frac{\partial\mathcal{H}_{\text{EP}}}{\partial\bm{p}}\right)\cdot \bm{p}-\rho\del\left(\frac{\partial\mathcal{H}_{\text{EP}}}{\partial\rho}-\del\cdot\frac{\partial\mathcal{H}_{\text{EP}}}{\partial\del\rho}\right),
\end{align}
together with \Eq{eq:EP_p} to write \Eq{eq:EP_mom} as \Eq{eq:ham_momentum}. Equations \eqref{eq:extended_first}, \eqref{eq:ham_cont}, and \eqref{eq:extended_last} are finally seen to be equivalent to Eqs.\,\eqref{eq:EP_p}, \eqref{ep:EP_continuity}, and \eqref{eq:EP_psi} in light of the relations $\bm{v} = \pd\mc{H}_{\text{EP}}/\pd \bm{p} = -\partial_t\bm{h}\cdot(\del\bm{h})^{-1}$. 

In order to see that every solution of the LBEP equations may be obtained from solutions of \Eqs{eq:extended_first}-\eqref{eq:extended_last}, we merely observe that Eqs.\,\eqref{eq:ham_momentum}-\eqref{eq:ham_cont} are precisely the momentum form of the LBEP equations. This system of PDEs was shown to be equivalent to the LBEP equations in Lemma \ref{mLBEP_eqs}. We may say that the mLBEP equations are embedded within the Euler-Lagrange equations associated with $\mathsf{A}$.

\end{proof}

\begin{remark}
It is to be noted that, in order to apply Whitham's averaging to the EP action principle, it is not entirely necessary to  construct the Hamiltonian formulation given in \Eqs{eq:EP:Ham} and \eq{eq:EP:act} with $\vec{p}$ as an additional dynamical variable. One could have simply introduced the back-to-labels map $\bm{h}$ and added the term involving the Lagrange multiplier in order to construct the  action. However, one of the advantages that we shall obtain after applying the NL--WKB extension to the generalized fluid system is that having $\vec{p}$ as an as an argument of the action functional is convenient when performing WKB asymptotics. This will be further discussed in \Sec{sec:six}, when we will apply our results to high-frequency acoustic waves interacting with a compressible isothermal flow.
\end{remark}

A simple corollary of Proposition \ref{extended_vp} is that the LBEP equations may be formulated as the Euler-Lagrange equations associated with an ordinary first-order classical field theory.

\begin{theorem}[LBEP field theory]\label{LBEP_field_theory}
Let $(M_{\text{EP}},\mc{C}_{\text{EP}},\mathfrak{L}_{\text{EP}})$ be the ordinary first-order classical field theory defined as follows.
\begin{itemize}
\item $M_{\text{EP}}=Q\times\mathbb{R}$.
\item $\mathcal{C}_{\text{EP}}$ is the space of smooth functions $\varphi:M\rightarrow F$, where $F = Q_0\times\mathbb{R}^3\times\mathbb{R}\times\mathbb{R}$.
\item Write a general element $\partial\varphi\in D = M_{8\times 4}(\mathbb{R})$ as
\begin{align}
\partial\varphi = \left(\begin{array}{cc}
(\del\bm{h})^T & \partial_t\bm{h}\\
(\del\bm{p})^T & \partial_t\bm{p}\\
(\del\rho)^T & \partial_t\rho\\
(\del\chi)^T & \partial_t\chi
\end{array}\right),\quad \begin{array}{c} \del\bm{h},\del\bm{p}\in M_{3\times 3}(\mathbb{R})\\
\partial_t\bm{h},\partial_t\bm{p},\del\rho,\del\chi\in M_{3\times1}(\mathbb{R})\\
\partial_t\rho,\partial_t\chi \in \mathbb{R}, \end{array} 
\end{align}
and a general element $\varphi\in F$ as $(\bm{h},\bm{p},\rho,\chi)\in Q_0\times\mathbb{R}^3\times\mathbb{R}\times\mathbb{R}$. The Lagrangian density $\mathfrak{L}_{\text{EP}}:M\times F\times D\rightarrow\mathbb{R}$ is given by
\begin{align}
\mathfrak{L}_{\text{EP}}(\bm{x},t,\varphi,\partial\varphi) = -  (\partial_t\bm{h})\cdot(\del\bm{h})^{-1}\cdot\bm{p} + \rho\,(\partial_t\chi-\partial_t\bm{h}\cdot(\del\bm{h})^{-1}\cdot\del\chi) - \mathcal{H}_{\text{EP}}(\bm{p},\rho,\del\rho).
\end{align}
\end{itemize}
The Euler-Lagrange equations associated with $(M_{\text{EP}},\mathcal{C}_{\text{EP}},\mathfrak{L}_{\text{EP}})$ are equivalent to \Eqs{eq:extended_first}-\eqref{eq:extended_last}.
\end{theorem}
\begin{remark}
The Euler-Lagrange equations associated with $(M_{\text{EP}},\mathcal{C}_{\text{EP}},\mathfrak{L}_{\text{EP}})$ comprise a larger system of PDEs than the LBEP equations. However, Proposition \ref{extended_vp} shows that the LBEP equations are embedded within the Euler-Lagrange equations associated with $\mathfrak{L}_{\text{EP}}$. In this sense it seems reasonable to attempt to uncover properties of the LBEP equations by studying $(M_{\text{EP}},\mathcal{C}_{\text{EP}},\mathfrak{L}_{\text{EP}})$. On the other hand, it also seems plausible that the additional variables present in $(M_{\text{EP}},\mathcal{C}_{\text{EP}},\mathfrak{L}_{\text{EP}})$'s Euler-Lagrange equations might render the study of $(M_{\text{EP}},\mathcal{C}_{\text{EP}},\mathfrak{L}_{\text{EP}})$ even more complicated than studying the LBEP equations directly. We will show in Section \ref{sec:four} that $(M_{\text{EP}},\mathcal{C}_{\text{EP}},\mathfrak{L}_{\text{EP}})$ does provide useful information about the LBEP equations because it can be can be combined with Whitham averaging in order to identify a variational formulation for the NL-WKB extension of the LBEP equations. In Section \ref{sec:five} we will show that symmetries of $(M_{\text{EP}},\mathcal{C}_{\text{EP}},\mathfrak{L}_{\text{EP}})$ and $(M_{\text{EP}},\mathcal{C}_{\text{EP}},\mathfrak{L}_{\text{EP}})$'s looping explain why the additional fields present in $(M_{\text{EP}},\mathcal{C}_{\text{EP}},\mathfrak{L}_{\text{EP}})$ do not spoil the utility of $(M_{\text{EP}},\mathcal{C}_{\text{EP}},\mathfrak{L}_{\text{EP}})$.
\end{remark}

\section{Variational structure of nonlinear WKB in the Eulerian frame\label{sec:four}}

In this section we will use the results from the previous two sections to identify a variational principle for the NL-WKB extension of the LBEP equations. We will frame our discussion in terms of the momentum form of the LBEP equations.
\begin{definition}
Introduce the operators $\partial_t^S = \partial_t+\partial_tS\, \partial_\theta$ and $\del^S=\del+\del S \,\partial_\theta$. The \emph{NL-WKB extension of the LBEP equations} is the system of PDEs
\begin{gather}
\partial_t^S\trho+\del^S\cdot\left(\trho\frac{\pd\mc{H}_{\text{EP}}}{\pd\bm{p}}\right) = 0\label{extLBEP_cont}\\
\partial_t^S\tp 
		+ 	\del^S \cdot \left(   \frac{\pd \mathcal{H}_{\text{EP}}}{\pd \bm{p}} \otimes \tp+  \frac{\pd \mathcal{H}_{\text{EP}}}{ \pd \del \rho}  \otimes \del^S \trho \right)\nonumber\\ 
		=	-\del^S \left( \trho \,  \frac{\pd \mathcal{H}_{\text{EP}}}{\pd \rho}    
				-  \trho\,\del^S \cdot \frac{\pd \mathcal{H}_{\text{EP}}}{ \pd \del \rho} 
				+ 	\tp\cdot\frac{\partial\mathcal{H}_{\text{EP}}}{\partial\bm{p}}-\mathcal{H}_{\text{EP}} \right),\label{extLBEP_mom}
\end{gather}
where the derivatives of the EP Hamiltonian density are evaluated at $(\tp,\trho,\del^S\trho)$. The unknown fields are $(\trho(\bm{x},\theta,t),\tp(\bm{x},\theta,t),S(\bm{x},t))\in \mathbb{R}\times\mathbb{R}^3\times S^1$. For the sake of brevity, we will refer to this system of PDEs as the \emph{extLBEP equations}.
\end{definition}

The rationale behind the existence of a variational formulation for the extLBEP equations is as follows. According to Theorem \ref{LBEP_field_theory}, the LBEP equations may be realized as a subset of the Euler-Lagrange equations arising from the ordinary first-order classical field theory $(M_{\text{EP}},\mathcal{C}_{\text{EP}},\mathfrak{L}_{\text{EP}})$. Because the LBEP equations are a subset of $(M_{\text{EP}},\mathcal{C}_{\text{EP}},\mathfrak{L}_{\text{EP}})$'s Euler-Lagrange equations, the extLBEP equations must be a subset of $(M_{\text{EP}},\mathcal{C}_{\text{EP}},\mathfrak{L}_{\text{EP}})$'s NL-WKB extension. Indeed, applying the NL-WKB extension procedure to Eqs.\,\eqref{eq:extended_first}-\eqref{eq:extended_last} involves applying NL-WKB extension to Eqs.\,\eqref{eq:ham_momentum}-\eqref{eq:ham_cont}, the latter of which are equivalent to the momentum form of the LBEP equations. But by Theorem \ref{whitham_thm} the NL-WKB extension of $(M_{\text{EP}},\mathcal{C}_{\text{EP}},\mathfrak{L}_{\text{EP}})$ arises as the Euler-Lagrange equations associated with the looping of $(M_{\text{EP}},\mathcal{C}_{\text{EP}},\mathfrak{L}_{\text{EP}})$, i.e. $(\widetilde{M}_{\text{EP}},\widetilde{\mathcal{C}}_{\text{EP}},\widetilde{\mathfrak{L}}_{\text{EP}})$. (See Definition \ref{looping_of_field_theory}.) Therefore the variational principle furnished by $(\widetilde{M}_{\text{EP}},\widetilde{\mathcal{C}}_{\text{EP}},\widetilde{\mathfrak{L}}_{\text{EP}})$'s action functional must serve as a variational principle for the extLBEP equations. In summary, we have proved the following.

\begin{proposition}\label{extLBEP_vp}
Let $(M_{\text{EP}},\mathcal{C}_{\text{EP}},\mathfrak{L}_{\text{EP}})$ be defined as in Theorem \ref{LBEP_field_theory}, and let $(\widetilde{M}_{\text{EP}},\widetilde{\mathcal{C}}_{\text{EP}},\widetilde{\mathfrak{L}}_{\text{EP}})$ be the looping of $(M_{\text{EP}},\mathcal{C}_{\text{EP}},\mathfrak{L}_{\text{EP}})$. Consider $\Phi\in \widetilde{\mathcal{C}}_{\text{EP}}$ with components $\Phi(\bm{x},t,\theta) = (\tih(\bm{x},t,\theta),\tp(\bm{x},t,\theta),\trho(\bm{x},t,\theta),\tchi(\bm{x},t,\theta),S(\bm{x},t))\in Q_0\times \mathbb{R}^3\times \mathbb{R}\times\mathbb{R}\times S^1$. The field $\Phi$ is a critical point of $(\widetilde{M}_{\text{EP}},\widetilde{\mathcal{C}}_{\text{EP}},\widetilde{\mathfrak{L}}_{\text{EP}})$'s local action functional $\tilde{A}_{U\times S^1}$ for each $U\subset M$ if and only if $\Phi$'s component functions satisfy the system of PDEs
\begin{gather}
\partial_t^S\tih = -\frac{\pd \mc{H}_{\text{EP}}}{\pd \bm{p}}\cdot \del^S\tih\label{extension_labels}\\
\partial_t^S\tp 
		+ 	\del^S \cdot \left(  \frac{\pd \mathcal{H}_{\text{EP}}}{\pd \bm{p}} \otimes \tp +  \frac{\pd \mathcal{H}_{\text{EP}}}{ \pd \del \rho}  \otimes \del^S \trho \right)\nonumber\\ 
		=	-\del^S \left( \trho \,  \frac{\pd \mathcal{H}_{\text{EP}}}{\pd \rho}    
				-  \trho\,\del^S \cdot \frac{\pd \mathcal{H}_{\text{EP}}}{ \pd \del \rho} 
				+ 	\tp\cdot\frac{\partial\mathcal{H}_{\text{EP}}}{\partial\bm{p}}-\mathcal{H}_{\text{EP}} \right)\label{extension_momentum}\\
\partial_t^S\trho+\del^S\cdot\left(\trho\frac{\pd\mc{H}_{\text{EP}}}{\pd\bm{p}}\right) = 0\label{extension_cont}\\		
\partial_t^S\tchi +\frac{\pd \mc{H}_{\text{EP}}}{\pd \bm{p}}\cdot \del^S\tchi = \frac{\pd\mc{H}_{\text{EP}}}{\pd \rho}-\del^S\cdot\frac{\pd\mc{H}_{\text{EP}}}{\pd\del\rho},\label{extension_chi}
\end{gather}
where the derivatives of $\mc{H}_{\text{EP}}$ are evaluated at $(\tp,\trho,\del^S\trho)$. In particular, the extLBEP equations are recovered as a subset of the Euler-Lagrange equations associated with $(\widetilde{M}_{\text{EP}},\widetilde{\mathcal{C}}_{\text{EP}},\widetilde{\mathfrak{L}}_{\text{EP}})$.
\end{proposition}

Proposition \ref{extLBEP_vp} gives the variational structure of the extLBEP equations in a manner that treats space and time on an equal footing. In order to analyze the extLBEP equations as a dynamical system, it is also important to formulate Proposition \ref{extLBEP_vp} in terms of evolving fields on space instead of ``static" fields on spacetime. For this purpose, it is useful to introduce the looped frozen field configurations and the extLBEP Lagrangian.

\begin{definition}
The space of \emph{looped frozen field configurations} $\ell\mc{C}_0$ is the collection of smooth mappings $S^1\rightarrow \mc{C}_0$. (cf. Definition \ref{frozen_field_def}.) We will identify elements of $\ell\mc{C}_0$ with mappings $Q\times S^1\ni (\bm{x},\theta)\mapsto (\tih(\bm{x},\theta),\tp(\bm{x},\theta),\trho(\bm{x},\theta),\tchi(\bm{x},\theta))\in Q_0\times \mathbb{R}^3\times\mathbb{R}\times \mathbb{R} $, where $\bm{x}\mapsto \tih(\bm{x},\theta)$ is a diffeomorphism for each $\theta$, and $\rho(\bm{x},\theta)>0$ for all $(\bm{x},\theta)\in Q\times S^1$.
\end{definition}

\begin{definition}
The \emph{extLBEP Lagrangian} is the functional $\widetilde{\mathsf{L}}:T(\ell\mc{C}_0\times C^\infty(Q,S^1))\rightarrow \mathbb{R}$ whose value at $(\tih,\tp,\trho,\tchi,S,\dot{\tih},\dot{\trho},\dot{\tchi},\dot{S})\in T(\ell\mc{C}_0\times C^\infty(Q,S^1))$ is given by
\begin{align}
\widetilde{\mathsf{L}}(\tih,\tp,\trho,\tchi,S,\dot{\tih},\dot{\trho},\dot{\tchi},\dot{S}) =& \fint\int_Q \left(\tp\cdot \tv + \trho\,\dot{\tchi}+\trho\,\dot{S}\partial_\theta\tchi + \trho\tv\cdot\del^S\tchi \right)\,d^3\bm{x}\,d\theta\nonumber\\
&-\fint\int_Q \mc{H}_{\text{EP}}(\tp,\trho,\del^S\trho)\,d^3\bm{x}\,d\theta,\label{extended_lag}
\end{align}
where
\begin{align}
\tv = -(\dot{\tih}+\dot{S}\partial_\theta\tih)\cdot (\del^S\tih )^{-1} ,\label{tilde_v_defined}
\end{align}
and  $\fint$ is defined by $\fint g(\theta)\,d\theta = (2\pi)^{-1}\int_0^{2\pi}g(\theta)\,d\theta$.
\end{definition}

In terms of $\widetilde{\mathsf{L}}$ and $\ell\mc{C}_0$, the proper reformulation of Proposition \ref{extLBEP_vp} is the following.

\begin{theorem}\label{theorem_3}
Let $\gamma:[t_1,t_2]\rightarrow \ell\mc{C}_0\times C^\infty(Q,S^1)$ be a smooth curve with components $\gamma = (\tih,\tp,\trho,\tchi,S)$. The curve $\gamma$ is a (fixed-endpoint) critical point of the functional 
\begin{align}
\widetilde{\mathsf{A}}(\gamma) = \int_{t_1}^{t_2}\widetilde{\mathsf{L}}(\gamma(t),\partial_t\gamma(t))\,dt\label{extLBEP_action}
\end{align}
if and only if the component functions $(\tih,\tp,\trho,\tchi,S)$ satisfy Eqs.\,\eqref{extension_labels}-\eqref{extension_chi}.
\end{theorem}

\noindent This theorem can be deduced directly from Proposition \ref{extLBEP_vp} by unpacking definitions. However, in order to more clearly highlight the mechanisms underlying the variational formulation of the extLBEP equations, we will give a direct proof of Theorem \ref{theorem_3} that proceeds without recourse to Proposition \ref{extLBEP_vp}. 

Before proceeding with the proof, we will first establish a generalization of the famous Lin constraint formula\cite{Newcomb_1962,Bretherton_1970} from variational hydrodynamics. 

\begin{definition}\label{def:phase_shift}
Given any set $T$ and a mapping $\psi: Q\times S^1\rightarrow T$, define \emph{the phase shift of $\psi$ by $S$}, $\psi^S: Q\times S^1\rightarrow T$, using the formula
\begin{align}
\psi^S(\bm{x},\theta) = \psi(\bm{x},\theta+S(\bm{x})).\label{expS}
\end{align}
Note that $\psi^S$ is \emph{not} the usual exponentiation operation used in elementary arithmetic.
\end{definition}

\begin{lemma}[WKB Lin constraint formula]
Let $(t,\epsilon)\mapsto \tih_{t,\epsilon}\in \ell \text{Diff}(Q,Q_0)$ be a smooth 2-parameter family of maps $S^1\rightarrow \text{Diff}(Q,Q_0)$. Let $(t,\epsilon)\mapsto S_{t,\epsilon}\in C^\infty(Q,S^1)$ be a smooth 2-parameter family of maps $Q\rightarrow S^1$. Then the pair of parameter velocities,
\begin{align}
\tv_{t,\epsilon} =& -(\pd_t\tih_{t,\epsilon}+\pd_tS_{t,\epsilon}\,\pd_\theta\tih_{t,\epsilon})\cdot (\del\tih_{t,\epsilon}+\del S_{t,\epsilon}\otimes \pd_\theta\tih_{t,\epsilon})^{-1}\\
\txi_{t,\epsilon} =& -(\pd_\epsilon\tih_{t,\epsilon}+\pd_\epsilon S_{t,\epsilon}\,\pd_\theta\tih_{t,\epsilon})\cdot (\del\tih_{t,\epsilon}+\del S_{t,\epsilon}\otimes \pd_\theta\tih_{t,\epsilon})^{-1},
\end{align}
satisfies the identity
\begin{gather}
\bigg(\pd_\epsilon\tv_{t,\epsilon}+\pd_\epsilon S_{t,\epsilon} \,\pd_\theta \tv_{t,\epsilon}\bigg)-\bigg(\pd_t\txi_{t,\epsilon}+\pd_tS_{t,\epsilon}\,\pd_\theta\txi_{t,\epsilon}\bigg)\nonumber\\
= [\tv_{t,\epsilon},\txi_{t,\epsilon}]+\tv_{t,\epsilon}\cdot\del S_{t,\epsilon}\,\pd_\theta\txi_{t,\epsilon}-\txi_{t,\epsilon}\cdot\del S_{t,\epsilon}\,\pd_\theta \tv_{t,\epsilon},\label{wkb_lin_formula}
\end{gather}
where $[\tv_{t,\epsilon},\txi_{t,\epsilon}]=\tv_{t,\epsilon}\cdot\del\txi_{t,\epsilon} - \txi_{t,\epsilon}\cdot\del\tv_{t,\epsilon}$ denotes the $\theta$-wise vector field commutator.
\end{lemma}

\begin{proof} 
Let $\bm{F}_{\lambda_1,\lambda_2}\in \text{Diff}(Q,Q_0)$ be any smooth two-parameter family of diffeomorphisms. The ordinary Lin constraint formula says that the parameter velocities $\bm{w}_k = -(\partial_{\lambda_k}\bm{F})\cdot (\del \bm{F})^{-1}$ satisfy
\begin{align}
\partial_{\lambda_2}\bm{w}_1 - \partial_{\lambda_1}\bm{w}_2 = [\bm{w}_1,\bm{w}_2].\label{lin_constraint_ordinary}
\end{align}
In the formula \eqref{lin_constraint_ordinary} set $\lambda_1 = t$, $\lambda_2 = \epsilon$, and $\bm{F}_{t,\epsilon} = \tih_{t,\epsilon}^{S_{t,\epsilon}}$. (Regard $\theta$ as a third parameter that comes along for the ride.) We then have
\begin{align}
\pd_\epsilon \tv^S - \pd_t\txi^S = [\tv^S,\txi^S],\label{proto_wkb_lin}
\end{align}
with
\begin{align}
\tv^S=&-\pd_t\tih^S\cdot (\del\tih^S)^{-1} = -\left([\pd_t\tih+\pd_t S \,\pd_\theta\tih]\cdot[\del\tih+\del S\otimes \pd_\theta\tih]^{-1}\right)^{S} \\
\txi^S=&-\pd_\epsilon\tih^S\cdot (\del\tih^S)^{-1} = -\left([\pd_\epsilon\tih+\pd_\epsilon S \,\pd_\theta\tih]\cdot[\del\tih+\del S\otimes \pd_\theta\tih]^{-1}\right)^{S} .
\end{align}
Phase shifting the formula \eqref{proto_wkb_lin} by $-S$ and applying the chain rule then leads to Eq.\,\eqref{wkb_lin_formula}. 

\end{proof}
\begin{remark}\label{useful_remark}
In the above proof, if we had instead set $\lambda_1 = t,\lambda_2 = \theta$ and applied the usual Lin constraint formula, the resulting identity would have been
\begin{align}
\pd_\theta\tv - \pd_t^S\widetilde{\bm{\zeta}} = [\tv,\widetilde{\bm{\zeta}}] + \tv\cdot\del S\,\pd_\theta\widetilde{\bm{\zeta}} - \widetilde{\bm{\zeta}}\cdot\del S \,\pd_\theta\tv,\label{other_identity}
\end{align}
where $\widetilde{\bm{\zeta}} = -\pd_\theta\tih\cdot(\del^S\tih)^{-1}$ is the $\theta$-parameter velocity. This identity will be used in the proof of Theorem \ref{theorem_3}.
\end{remark}
\begin{proof}[proof of Theorem 3]
According to the WKB Lin constraint formula \eqref{wkb_lin_formula}, the first variation of the velocity $\tv$ is given by
\begin{align}
\delta\tv = &-\delta S\,\partial_\theta\tv +\left(\pd_t\txi + \partial_t S\,\pd_\theta \txi\right)+[\tv,\txi] + \tv\cdot\del S\,\pd_\theta\txi -\txi\cdot\del S\,\pd_\theta\tv,
\end{align}
where the loop of vector fields $\txi = -(\delta\tih+\delta S\,\pd_\theta\tih)\cdot (\del\tih+\del S\otimes \pd_\theta\tih)^{-1}$. The first (fixed-endpoint) variation of the action $\widetilde{\mathsf{A}}$ is therefore given by
\begin{align}
\delta\widetilde{\mathsf{A}} =& \int_{t_1}^{t_2} \fint\int_Q \left(\tv-\frac{\pd\mc{H}_{\text{EP}}}{\pd\bm{p}}\right)\cdot\delta\tp\,d^3\bm{x}\,d\theta\nonumber\\
+&\int_{t_1}^{t_2} \fint\int_Q\left(\pd^S_t\tchi +\tv\cdot\del^S\tchi-\frac{\pd\mc{H}_{\text{EP}}}{\pd \rho}+\del^S\cdot\frac{\pd\mc{H}_{\text{EP}}}{\pd\del\rho}\right)\delta\trho\,d^3\bm{x}\,d\theta\nonumber\\
-&\int_{t_1}^{t_2} \fint\int_Q\left(\pd_t^S\trho + \del^S\cdot(\tv\trho)\right)\delta\tchi\,d^3\bm{x}\,d\theta\nonumber\\
-&\int_{t_1}^{t_2} \fint\int_Q\bigg(\pd_t^S \widetilde{\bm{P}} + \del^S(\widetilde{\bm{P}}\cdot\tv)+(\del^S\times \widetilde{\bm{P}})\times\tv+(\del^S\cdot\tv)\widetilde{\bm{P}}\bigg)\cdot\txi\,d^3\bm{x}\,d\theta\nonumber\\
-&\int_{t_1}^{t_2} \fint\int_Q\bigg(\pd_t(\trho \pd_\theta\tchi) +\del\cdot (\tv \trho\pd_\theta\tchi)+ \widetilde{\bm{P}}\cdot\pd_\theta\tv - \del\cdot\left(\pd_\theta\trho\frac{\pd\mc{H}_{\text{EP}}}{\pd \del\rho}\right)\bigg)\,\delta S\,d^3\bm{x}\,d\theta,
\end{align}
where we have temporarily introduced the shorthand notation $\widetilde{\bm{P}} = \tp+\trho\, \del^S\tchi$. Alternatively, we may isolate all of the variations of $S$ by writing $\txi = \txi_0 + \delta S\,\widetilde{\bm{\zeta}}$ with $\widetilde{\bm{\zeta}} = -(\pd_\theta\tih)\cdot (\del\tih+\del S\otimes \pd_\theta\tih)^{-1}$ and $\txi_0 = -\delta\tih\cdot(\del^S\tih)^{-1}$, thereby obtaining
\begin{align}
\delta\widetilde{\mathsf{A}} =& \int_{t_1}^{t_2} \fint\int_Q \left(\tv-\frac{\pd\mc{H}_{\text{EP}}}{\pd\bm{p}}\right)\cdot\delta\tp\,d^3\bm{x}\,d\theta\nonumber\\
+&\int_{t_1}^{t_2} \fint\int_Q\left(\pd^S_t\tchi +\tv\cdot\del^S\tchi-\frac{\pd\mc{H}_{\text{EP}}}{\pd \rho}+\del^S\cdot\frac{\pd\mc{H}_{\text{EP}}}{\pd\del\rho}\right)\delta\trho\,d^3\bm{x}\,d\theta\nonumber\\
-&\int_{t_1}^{t_2} \fint\int_Q\left(\pd_t^S\trho + \del^S\cdot(\tv\trho)\right)\delta\tchi\,d^3\bm{x}\,d\theta\nonumber\\
-&\int_{t_1}^{t_2} \fint\int_Q\bigg(\pd_t^S \widetilde{\bm{P}} + \del^S(\widetilde{\bm{P}}\cdot\tv)+(\del^S\times \widetilde{\bm{P}})\times\tv+(\del^S\cdot\tv)\widetilde{\bm{P}}\bigg)\cdot\txi_0\,d^3\bm{x}\,d\theta\nonumber\\
-&\int_{t_1}^{t_2} \int_Q\bigg(\pd_t\fint\widetilde{\mathcal{I}}\,d\theta +\del\cdot \fint \tv \widetilde{\mathcal{I}}\,d\theta - \del\cdot\left(\fint \pd_\theta\trho\frac{\pd\mc{H}_{\text{EP}}}{\pd \del\rho}\,d\theta\right)\bigg)\,\delta S\,d^3\bm{x}\,d\theta,\label{first_variation_formula_thm_3}
\end{align}
where the \emph{specific wave action density} $\widetilde{\mc{I}}$ is given by
\begin{align}
\widetilde{\mc{I}} =\trho \,\pd_\theta\tchi  + (\tp+\trho \,\del^S\tchi)\cdot\widetilde{\bm{\zeta}}  .
	\label{eq:specific_wave_density}
\end{align}  
Because $\delta\tchi,\delta\trho,\delta\tp,\delta S,$ and $\txi_0 $ are arbitrary, $\delta\widetilde{\mathsf{A}} = 0$ if and only if 
\begin{gather}
\pd_t^S\tih = -\frac{\pd\mc{H}_{\text{EP}}}{\pd\bm{p}}\cdot \del^S\tih\label{prf_3_first}\\
\pd_t^S \widetilde{\bm{P}} + \del^S\left(\widetilde{\bm{P}}\cdot\frac{\pd\mc{H}_{\text{EP}}}{\pd \tp}\right)+(\del^S\times \widetilde{\bm{P}})\times\frac{\pd\mc{H}_{\text{EP}}}{\pd \tp}+\left(\del^S\cdot\frac{\pd\mc{H}_{\text{EP}}}{\pd \tp}\right)\widetilde{\bm{P}} =0\label{proto_momentum_thm_3}\\
\pd_t^S\trho +\del^S\cdot\left(\trho\frac{\pd\mc{H}_{\text{EP}}}{\pd \tp}\right) = 0\\ 
\pd_t^S\tchi + \frac{\pd \mc{H}_{\text{EP}}}{\pd \tp}\cdot\del^S\tchi = \frac{\pd\mc{H}_{\text{EP}}}{\pd \rho}-\del^S\cdot\frac{\pd\mc{H}_{\text{EP}}}{\pd\del\rho}\label{prf_3_chi}\\
\pd_t\fint\widetilde{\mathcal{I}}\,d\theta +\del\cdot \fint \frac{\pd\mc{H}_{\text{EP}}}{\pd \bm{p}} \widetilde{\mathcal{I}}\,d\theta = \del\cdot\left(\fint \pd_\theta\trho\frac{\pd\mc{H}_{\text{EP}}}{\pd \del\rho}\,d\theta\right).\label{prf_3_last}
\end{gather}
Notice that in moving from the first variation formula \eqref{first_variation_formula_thm_3} to Eqs.\,\eqref{prf_3_first}-\eqref{prf_3_last}, we have used the $\tv = \pd\mc{H}_{\text{EP}}/\pd \tp$ in order to eliminate $\tv$ in favor of $\pd\mc{H}_{\text{EP}}/\pd \tp$. In order to finish the proof, we will now show that Eqs.\,\eqref{prf_3_first}-\eqref{prf_3_last} are equivalent to Eqs.\,\eqref{extension_labels}-\eqref{extension_chi} by (a) demonstrating that Eq.\,\eqref{proto_momentum_thm_3} is equivalent to Eq.\,\eqref{extension_momentum}, and (b) proving that the wave action conservation law \eqref{prf_3_last} is implied by Eqs.\,\eqref{extension_labels}-\eqref{extension_chi}.
\\ \\
\noindent (a): Notice that Eq.\,\eqref{proto_momentum_thm_3} may be written as $\mathsf{M}\widetilde{\bm{P}} = 0$, where $\mathsf{M}$ is the linear differential operator whose action on any time-dependent loop of vector fields $\widetilde{\bm{w}}$ is given by
\begin{align}
\mathsf{M}\widetilde{\bm{w}} =& \pd_t^S \widetilde{\bm{w}} + \del^S(\widetilde{\bm{w}}\cdot\frac{\pd\mc{H}_{\text{EP}}}{\pd \tp})+(\del^S\times \widetilde{\bm{w}})\times\frac{\pd\mc{H}_{\text{EP}}}{\pd \tp}+(\del^S\cdot\frac{\pd\mc{H}_{\text{EP}}}{\pd \tp})\widetilde{\bm{w}}\nonumber\\
=&\pd_t^S \widetilde{\bm{w}} + \del^S(\widetilde{\bm{w}}\cdot\frac{\pd\mc{H}_{\text{EP}}}{\pd \tp})+\del^S\cdot(\frac{\pd\mc{H}_{\text{EP}}}{\pd \tp}\otimes \widetilde{\bm{w}}) - (\del^S\widetilde{\bm{w}})\cdot \frac{\pd\mc{H}_{\text{EP}}}{\pd \tp}.
\end{align}
Because $\widetilde{\bm{P}}  = \tp+\trho\del^S\tchi$, Eq.\,\eqref{proto_momentum_thm_3} is also equivalent to $\mathsf{M}\tp = -\mathsf{M}(\trho\del^S\tchi)$. By a direct calculation involving the continuity equation \eqref{extension_cont}, we have
\begin{align}
\mathsf{M}(\trho\del^S\tchi) = &\rho \del^S (\pd_t^S\tchi+\frac{\pd\mc{H}_{\text{EP}}}{\pd \tp}\cdot\del^S\tchi)\nonumber\\
=& \rho\del^S \left(\frac{\pd\mc{H}_{\text{EP}}}{\pd \rho}-\del^S\cdot\frac{\pd\mc{H}_{\text{EP}}}{\pd\del\rho}\right),
\end{align}
where we have used the Euler-Lagrange equation \eqref{prf_3_chi} in the second line. Moreover, because
\begin{align}
\del^S\mc{H}_{\text{EP}} = \left(\frac{\pd\mc{H}_{\text{EP}}}{\pd \rho}-\del^S\cdot\frac{\pd\mc{H}_{\text{EP}}}{\pd\del\rho}\right)\del^S\trho  + \del^S\cdot\left(\frac{\pd\mc{H}_{\text{EP}}}{\pd \del\rho}\otimes \del^S\trho\right) + (\del^S\tp)\cdot\frac{\pd\mc{H}_{\text{EP}}}{\pd \tp}
\end{align}
when the derivatives of $\mc{H}_{\text{EP}}$ are evaluated at $(\tp,\trho,\del^S\trho)$, the sum $\mathsf{M}(\trho\del^S\tchi)+\del^S\mc{H}_{\text{EP}}$ is given by
\begin{align}
\mathsf{M}(\trho\del^S\tchi)+\del^S\mc{H}_{\text{EP}} = &\del^S\left(\trho\frac{\pd\mc{H}_{\text{EP}}}{\pd \rho}-\trho\del^S\cdot\frac{\pd\mc{H}_{\text{EP}}}{\pd\del\rho}\right)+\del^S\cdot\left(\frac{\pd\mc{H}_{\text{EP}}}{\pd \del\rho}\otimes \del^S\trho\right)\nonumber\\
&+(\del^S\tp)\cdot\frac{\pd\mc{H}_{\text{EP}}}{\pd \tp}.\label{useful_id_thm3}
\end{align}
Using Eq.\,\eqref{useful_id_thm3} to evaluate the right-hand-side of $\mathsf{M}\tp = -\mathsf{M}(\trho\del^S\tchi)$ leads directly to Eq.\,\eqref{extension_momentum}. 
\\ \\
\noindent (b): A simple way to see that Eqs.\,\eqref{extension_labels}-\eqref{extension_chi} imply the wave action conservation law is to analyze the quantity
\begin{align}
\gamma = \pd_t^S \widetilde{\mc{I}}+\del^S\cdot\left(\frac{\pd\mc{H}_{\text{EP}}}{\pd \tp}\widetilde{\mc{I}}\right).
\end{align} 
Using Eqs.\,\eqref{extension_momentum},\eqref{extension_cont}, \eqref{extension_chi}, and the identity \eqref{other_identity} introduced in Remark \ref{useful_remark},
the quantity $\gamma$ may be written
\begin{align}
\gamma = \trho\pd_\theta\left(\frac{\pd\mc{H}_{\text{EP}}}{\pd \rho}-\del^S\cdot\frac{\pd\mc{H}_{\text{EP}}}{\pd\del\rho}\right)+\tp\cdot\pd_\theta\frac{\pd\mc{H}_{\text{EP}}}{\pd\bm{p}}.
\end{align}
Upon applying the identity
\begin{align}
\pd_\theta \mc{H}_{\text{EP}} =&\pd_\theta\tp\cdot \frac{\pd \mc{H}_{\text{EP}}}{\pd \bm{p}}+\pd_\theta\rho \frac{\pd \mc{H}_{\text{EP}}}{\pd \rho}+(\pd_\theta\del^S\trho)\cdot \frac{\pd \mc{H}_{\text{EP}}}{\pd \del\rho}\nonumber\\
=&\pd_\theta\tp\cdot \frac{\pd \mc{H}_{\text{EP}}}{\pd \bm{p}}+\pd_\theta\rho \left(\frac{\pd \mc{H}_{\text{EP}}}{\pd \rho}-\del^S\cdot\frac{\pd  \mc{H}_{\text{EP}}}{\pd \del\rho}\right)+\del^S\cdot\left(\pd_\theta\rho \frac{\pd \mc{H}_{\text{EP}} }{\pd \del\rho}\right),
\end{align}
we also have
\begin{align}
\gamma =& \pd_\theta\left(\trho\frac{\pd\mc{H}_{\text{EP}}}{\pd \rho}-\trho\del^S\cdot\frac{\pd\mc{H}_{\text{EP}}}{\pd\del\rho} +\tp\cdot \frac{\pd \mc{H}_{\text{EP}}}{\pd \bm{p}} - \mc{H}_{\text{EP}} \right) +\del^S\cdot \left(\pd_\theta \trho\,
\frac{\pd  \mc{H}_{\text{EP}}}{\pd \del\rho}\right).
\end{align}
The $\theta$-average of $\gamma$ is therefore
\begin{align}
\pd_t \fint \widetilde{\mc{I}}\,d\theta + \del\cdot \fint \frac{\pd\mc{H}_{\text{EP}}}{\pd \bm{p}} \widetilde{\mc{I}}\,d\theta =\fint \gamma\,d\theta = \del\cdot\fint \pd_\theta\trho \frac{\pd\mc{H}_{\text{EP}}}{\pd\del\rho}\,d\theta,
\end{align}
which establishes the wave action conservation law \eqref{prf_3_last}.
\end{proof}

\section{Looping the relabeling group\label{sec:five}}
One of the most remarkable features of the variational principle introduced in Theorem \ref{theorem_3} is its associated symmetry group. Let $G$ be the group of symmetries of the LBEP phase space action functional \eqref{eq:EP:act}. (Think of $G$ as the symmetry group for the LBEP equations \emph{before} applying nonlinear WKB extension.) To $G$ we may associate the \emph{loop group}\cite{Pressley_1988} $\ell G$, which comprises mappings $S^1\rightarrow G$. In this section, we will show that $\ell G$ is a group of symmetries for the action functional \eqref{extLBEP_action} from Theorem \ref{theorem_3}. We say the symmetry group $G$ becomes \emph{looped} in passing from the LBEP equations to their nonlinear WKB extension. In particular, the subgroup of $G$ given by particle relabeling transformations becomes looped when passing from the LBEP equations to their non-linear WKB extension. Using Noether's theorem, we will deduce the conserved quantity associated with loops of relabeling transformations, and thereby infer the analogue of Kelvin's circulation theorem for Eulerian nonlinear WKB. Notably, this circulation theorem represents a kind of extension of the circulation theorem discussed in Ref.\,\onlinecite{Gjaja_Holm_1996}; the latter may be seen as a consequence of symmetry under the group of \emph{mean} (i.e. $\theta$-independent) relabeling transformations, while the circulation theorem discussed in this section is a consequence of symmetry under the larger group of loops of relabeling transformations. (We will discuss the relationship between these two notions of circulation in greater detail in the next section, where we will apply our theoretical results to a concrete example of wave-mean-flow interaction.) Finally, we will use the loops of relabeling transformations to give a group-theoretic explanation for the one-way coupling between $(\tp,\trho)$ and $(\tih,\tchi)$ in the extLBEP equations. In so doing, we will have demonstrated that the nonlinear WKB extension of an Euler-Poincar\'e fluid theory fits into a general pattern that was emphasized by Marsden and Weinstein in Ref.\,\onlinecite{Marsden_1982}; many dissipation-free models from continuum mechanics arise as quotients of variational models by an appropriate symmetry group.


We begin by recalling the definition of a loop group. Let $G$ be a group with elements $g\in G$ and product $g_1g_2\in G$. The loop group $\ell G$ associated with $G$ is the set of all mappings $S^1\rightarrow G$. When $G$ carries a manifold structure, we also require the mappings to be smooth. If $\widetilde{g}$ denotes a typical element of $\ell G$, the group multiplication $\widetilde{g}_1*\widetilde{g}_2$ in $\ell G$ is given by $(\widetilde{g}_1 *\widetilde{g}_2)(\theta) = \widetilde{g}_1(\theta)\widetilde{g}_2(\theta)$. Thus, the product on $\ell G$ is given by ``parallelizing" the product on $G$ over the loop parameter $\theta\in S^1$. Accordingly, the identity and inverse map for $\ell G$ are given by $e_{\ell G}(\theta) = e_G$ and $(\widetilde{g}^{-1})(\theta) = (\widetilde{g}(\theta))^{-1}$. Note that we use the same symbol for the inverse operations in $G$ and $\ell G$.

Next we turn to establishing the main result of this section.

\begin{theorem}\label{theorem_4}
Let $G$ be a group. Suppose $\Phi:\mathcal{C}_0\times G\rightarrow \mathcal{C}_0$ is a right $G$-action on the space of frozen field configurations that leaves the action functional \eqref{eq:EP:act} invariant. Let $\widetilde{\Phi}_0$ be the right $\ell G$-action on $\ell\mathcal{C}_0$ given by ``parallelizing" the action $\Phi$, i.e.
\begin{align}
\widetilde{\Phi}_0((\tih,\tp,\trho,\tchi),\widetilde{g})(\theta) = \Phi((\tih(\theta),\tp(\theta),\trho(\theta),\tchi(\theta)),\widetilde{g}(\theta)).\label{parallelized_G_action}
\end{align} 
Then there is a right $\ell G$-action $\widetilde{\Phi}$ on $\ell \mathcal{C}_0\times C^\infty(Q,S^1)$ given by
\begin{align}
\widetilde{\Phi}((\tih,\tp,\trho,\tchi,S), \widetilde{g}) = ([\widetilde{\Phi}_0((\tih^S,\tp^S,\trho^S,\tchi^S),\widetilde{g})]^{-S},S),\label{transformed_G_action}
\end{align}
that leaves the action functional \eqref{extLBEP_action} invariant. (Recall that the notation $\cdot^S$ was defined in Eq.\,\eqref{expS}.)
\end{theorem}

\begin{proof}
Given a smooth curve $\hat{\gamma}:[t_1,t_2]\rightarrow \ell\mc{C}_0\times C^\infty(Q,S^1)$, introduce the component curves $\hat{\gamma}_1,\hat{\gamma}_2$ satisfying $\hat{\gamma}(t) = (\hat{\gamma}_1(t),\hat{\gamma}_2(t))\in \ell\mc{C}_0\times C^\infty(Q,S^1)$ for all $ t\in [t_1,t_2]$. Consider the action functional $\widetilde{\mathsf{A}}_0$ defined on the space of such curves by the formula
\begin{align}
\widetilde{\mathsf{A}}_0(\hat{\gamma}) = \int_{t_1}^{t_2}\fint\mathsf{L}(\hat{\gamma}_1(t,\theta),\partial_t\hat{\gamma}_1(t,\theta))\,d\theta\,dt.\label{parallelized_action}
\end{align}
We recall that $\mathsf{L}$ is the parameter-independent phase space Lagrangian for the LBEP equations introduced in Eq.\,\eqref{psl_no_param}. The intuition behind Eq.\,\eqref{parallelized_action} is as follows. For each $\theta\in S^1$, we may evaluate the action $\mathsf{A}$ in Eq.\,\eqref{eq:EP:act} on the curve $t\mapsto \hat{\gamma}_1(\theta,t)\in \mathcal{C}_0$, thereby obtaining the real number $\mathsf{A}(\theta)$. The value of $\widetilde{\mathsf{A}}_0(\hat{\gamma})$ is then given by averaging $\mathsf{A}(\theta)$ over $S^1$. Because $\mathsf{A}(\theta)$ is $G$-invariant for each $\theta\in S^1$, it follows that the ``parallelized" $G$-action $\widetilde{\Phi}_0$ leaves the action $\widetilde{\mathsf{A}}_0$ invariant.

As is generally true in Lagrangian mechanics, equivalent formulations of the variational problem $\delta \widetilde{\mathsf{A}}_0 = 0 $ may be obtained by applying invertible transformations to the ``generalized coordinates," which in this case may be identified with the space $\ell\mc{C}_0\times C^\infty(Q,S^1)$. In particular, we may apply the mapping $T:(\hat{\tih},\hat{\tp},\hat{\trho},\hat{\tchi},S)\mapsto (\tih,\tp,\trho,\tchi,S)$, where
\begin{align}
(\tih,\tp,\trho,\tchi) = (\hat{\tih},\hat{\tp},\hat{\trho},\hat{\tchi})^{-S}.
\end{align}
After applying the transformation $T$, the action functional $\widetilde{\mathsf{A}}_0$ is transformed into the action functional $\widetilde{\mathsf{A}}_0^*$, whose value at $\gamma = (\gamma_1,\gamma_2) : [t_1,t_2]\rightarrow \ell\mc{C}_0\times C^\infty(Q,S^1)$ is given by
\begin{align}
\widetilde{\mathsf{A}}_0^*(\gamma) =& \widetilde{\mathsf{A}}_0(\hat{\gamma})\nonumber\\
=&\int_{t_1}^{t_2}\fint\mathsf{L}({\gamma}_1^S(t,\theta),\partial_t{\gamma}_1^S(t,\theta))\,d\theta\,dt.\label{transformed_parallelized_action}
\end{align}

Because $\widetilde{\mathsf{A}}_0(\hat{\gamma})$ is by hypothesis invariant under the transformation 
\begin{align}
\hat{\gamma}&\mapsto \hat{\gamma}\cdot\widetilde{g}\nonumber\\
(\hat{\gamma}\cdot\widetilde{g})(t) &= \bigg(\widetilde{\Phi}_0(\hat{\gamma}_1(t),\widetilde{g}),\hat{\gamma}_2(t)\bigg)
\end{align}
for each $\widetilde{g}\in \ell G$, the quantity $\widetilde{\mathsf{A}}_0^*({\gamma})$ must be invariant under the transformation  given by
\begin{align}
\gamma&\mapsto \gamma\bullet \widetilde{g}\nonumber\\
(\gamma\cdot\widetilde{g})(t)&=  \bigg(T\circ\bigg((T^{-1}\circ\gamma)\cdot \widetilde{g}\bigg) \bigg)(t)\nonumber\\
& = \bigg(\widetilde{\Phi}_0({\gamma}_1^S(t),\widetilde{g})^{-S},\gamma_2(t)\bigg)\nonumber\\
& = \widetilde{\Phi}(\gamma(t),\widetilde{g})\label{invariance_id_thm4}
\end{align}
for each $\widetilde{g}\in \ell G$. Note that we have recognized the definition \eqref{transformed_G_action} of $\widetilde{\Phi}$ in the last line of Eq.\,\eqref{invariance_id_thm4}. We have therefore shown that the $G$-action $\widetilde{\Phi}$ leaves the action functional $\widetilde{\mathsf{A}}_0^*$ invariant.

In order to complete the proof, we will now show by direct calculation that $\widetilde{\mathsf{A}}_0^*$ is in fact equal to the action defined in Eq.\,\eqref{extLBEP_action}. Write $\hat{\gamma}_1 = ({\tih},{\tp},{\trho},{\tchi}) $ and ${\gamma}_1 = S$. According to Eq.\,\eqref{transformed_parallelized_action} and \eqref{psl_no_param}, the value of $\widetilde{\mathsf{A}}_0^*(\gamma)$ is given by
\begin{align}
\widetilde{\mathsf{A}}_0^*(\gamma) = & -\int_{t_1}^{t_2}\fint \int_Q  \pd_t \tih^S\cdot(\del\tih^S)^{-1}\cdot \widetilde{\bm{p}}^S\,d^3\bm{x}\,d\theta\,dt \nonumber\\
&+ \int_{t_1}^{t_2}\fint\int_Q (\pd_t{\tchi}^S-\del\tchi^S\cdot (\del \tih^S)^{-1}\cdot\pd_t\tih^S)\,\rho^S\,d^3\bm{x}\,d\theta\,dt\nonumber\\
&-\int_{t_1}^{t_2}\fint\int_Q \mathcal{H}_{\text{EP}}(\bm{p}^S,\rho^S,\del\rho^S)\,d^3\bm{x}\,d\theta\,dt.
\end{align}
Using the derivative identities
\begin{align}
\pd_t\tih^S(\bm{x},\theta) &= \pd_t\tih(\bm{x},\theta+S(\bm{x})) + \pd_t S(\bm{x})\, \pd_\theta\tih(\bm{x},\theta+S(\bm{x}))\nonumber\\
\del\tih^S &=\del\tih(\bm{x},\theta+S(\bm{x})) + \del S(\bm{x})\otimes \pd_\theta\tih(\bm{x},\theta+S(\bm{x})), 
\end{align}
along with similar identities for $\tchi$ and $\trho$, we may also write
\begin{align}
\widetilde{\mathsf{A}}_0^*(\gamma) = & \int_{t_1}^{t_2}\fint \int_Q  \tp^S\cdot\tv^S\,d^3\bm{x}\,d\theta\,dt \nonumber\\
&+ \int_{t_1}^{t_2}\fint\int_Q ([\pd^S_t{\tchi}]^S+\tv^S\cdot[\del^S\tchi]^S)\,\rho^S\,d^3\bm{x}\,d\theta\,dt\nonumber\\
&-\int_{t_1}^{t_2}\fint\int_Q \mathcal{H}_{\text{EP}}(\tp,\trho,\del^S\trho)^S\,d^3\bm{x}\,d\theta\,dt,
\end{align}
where we have defined $\tv = - (\pd^S_t\tih)\cdot (\del^S\tih)^{-1}$ as was done earlier in Eq.\,\eqref{tilde_v_defined}.
Now apply the integral identity $\fint \int_Q f^S(\bm{x},\theta)\,d^3\bm{x}\,d\theta = \fint\int_Q f(\bm{x},\theta)\,d^3\bm{x}\,d\theta$, which is valid for all integrable $f:Q\times S^1\rightarrow \mathbb{R}$, to obtain
\begin{align}
\widetilde{\mathsf{A}}_0^*(\gamma) = & \int_{t_1}^{t_2}\fint \int_Q  \tp\cdot\tv\,d^3\bm{x}\,d\theta\,dt \nonumber\\
&+ \int_{t_1}^{t_2}\fint\int_Q (\pd^S_t{\tchi}+\tv\cdot\del^S\tchi)\,\rho\,d^3\bm{x}\,d\theta\,dt\nonumber\\
&-\int_{t_1}^{t_2}\fint\int_Q \mathcal{H}_{\text{EP}}(\tp,\trho,\del^S\trho)\,d^3\bm{x}\,d\theta\,dt.\label{final_thm4}
\end{align}
By Eq.\,\eqref{extended_lag}, Eq.\,\eqref{final_thm4} is just the formula \eqref{extLBEP_action} defining the action functional $\widetilde{\mathsf{A}}$.
\end{proof}

Theorem \ref{theorem_4} says that the symmetry group of the phase-space action for the LBEP equations becomes \emph{looped} when applying the nonlinear WKB extension procedure. As a consequence, we should expect that momentum maps for the LBEP equations should be looped by nonlinear WKB extension as well. The next proposition shows that this is true.

\begin{proposition}\label{looped_momentum_prop}
Endow $\mathcal{C}_0$ with the symplectic form $-\mathbf{d}\Theta$, where the $1$-form $\Theta$ is given by
\begin{align}\label{LBEP_Theta}
\Theta[\delta\bm{h},\delta\bm{p},\delta\rho,\delta\chi] = \int_Q \bm{p}\cdot\bm{\xi}\,d^3\bm{x} + \int_Q \rho (\delta\chi+\bm{\xi}\cdot\del\chi)\,d^3\bm{x},
\end{align}
with $\bm{\xi} = -\delta\bm{h}\cdot(\del\bm{h})^{-1}$. Endow $\ell\mathcal{C}_0\times C^\infty(Q,S^1)$ with the presymplectic form $-\mathbf{d}\widetilde{\Theta}$, where 
\begin{align}
\widetilde{\Theta}[\delta\tih,\delta\tp,\delta\trho,\delta\tchi,\delta S]  = \fint \int_Q \tp\cdot \txi\,d^3\bm{x}\,d\theta + \fint \int_Q \trho\left(\delta\tchi+\delta S\pd_\theta\tchi+\txi\cdot\del^S\tchi\right)\,d^3\bm{x}\,d\theta,
\end{align}
with $\txi = - (\delta\tih +\delta S\pd_\theta\tih)\cdot(\del^S\tih)^{-1}$.
Let $G$ be a Lie group. Suppose there is a right $G$-action on $\mathcal{C}_0$ that preserves $\Theta$, and therefore admits an $\text{Ad}^*$-equivariant momentum map $\mu:\mathcal{C}_0\rightarrow \mathfrak{g}^*$ given by
\begin{align}
\langle \mu , X \rangle = \Theta[X_{\mathcal{C}_0}]
\end{align}
for each $X\in\mathfrak{g}$. (The vector field $X_{\mathcal{C}_0}$ is the infinitesimal generator on $\mathcal{C}_0$ in the direction $X\in\mathfrak{g}$.) Then the looped $\ell G$-action given by Theorem \ref{theorem_4} admits an $\text{Ad}^*$-equivariant presymplectic momentum map $\widetilde{\mu}:\ell\mathcal{C}_0\times C^\infty(Q,S^1)\rightarrow \ell\mathfrak{g}^*$ given by
\begin{align}
\widetilde{\mu}(\tih,\tp,\trho,\tchi,S)(\theta) = \mu(\tih^S(\theta),\tp^S(\theta),\trho^S(\theta),\tchi^S(\theta))\label{desried_momentum_map}
\end{align}
\end{proposition} 

\begin{proof}
Given any group action $\Psi: M\times H\rightarrow M$, where $M$ is a set and $H$ is the group, it will be convenient to introduce the maps $\Psi_h:M\rightarrow M$ for each $h\in H$, where $\Psi_h(m) = \Psi(m,h)$.
Let $\Phi:\mathcal{C}_0\times G\rightarrow \mathcal{C}_0$ be the right $G$-action that preserves $\Theta$, $\widetilde{\Phi}_0: \ell\mathcal{C}_0\times \ell G\rightarrow \ell\mathcal{C}_0$ the parallelization of $\Phi$ defined by Eq.\,\eqref{parallelized_G_action}, and $\widetilde{\Phi}:\ell\mathcal{C}_0\times C^\infty(Q,S^1)\times \ell G\rightarrow \ell\mathcal{C}_0\times C^\infty(Q,S^1)$ the $\ell G$ action provided by Theorem \ref{theorem_4}. By hypothesis, the action $\Phi$ preserves the $1$-form $\Theta$ in the sense that $\Phi_g^*\Theta = \Theta$ for each $g\in G$.

Let  $T: \ell\mathcal{C}_0\times C^\infty(Q,\mathbb{R})\rightarrow \ell\mathcal{C}_0\times C^\infty(Q,\mathbb{R})$ be the diffeomorphism given by $T:(\hat{\tih},\hat{\tp},\hat{\trho},\hat{\tchi},S)\mapsto (\tih,\tp,\trho,\tchi,S)$, with
\begin{align}
(\tih,\tp,\trho,\tchi) = (\hat{\tih},\hat{\tp},\hat{\trho},\hat{\tchi})^{-S},
\end{align}
and $\widetilde{\Theta}_0$ the $1$-form on $\ell\mathcal{C}_0\times C^\infty(Q,S^1)$ defined by
\begin{align}
\widetilde{\Theta}_0[\delta\tih,\delta\tp,\delta\trho,\delta\tchi,\delta S] = \fint \Theta[\delta\tih(\theta),\delta\tp(\theta),\delta\trho(\theta),\delta\tchi(\theta)]\,d\theta.\label{theta0}
\end{align}
The proof of Theorem \ref{theorem_4} shows that the pullback of $\widetilde{\Theta}$ along $T$
is given by $\widetilde{\Theta}_0$. 

Because $\Phi_g^*\Theta = \Theta$ for each $g\in G$, we have $(\widetilde{\Phi}_{0\widetilde{g}}\times I)^*\widetilde{\Theta}_0 = \widetilde{\Theta}_0$ for each $\widetilde{g}\in \ell G$, where $I:C^\infty(Q,S^1)\rightarrow C^\infty(Q,S^1)$ is the identity mapping on the space of phase functions. Therefore the mapping $\widetilde{\mu}_0: \ell\mathcal{C}_0\times C^\infty(Q,S^1)\rightarrow \ell\mathfrak{g}^*$ defined by
\begin{align}
\langle\widetilde{\mu}_0 , \widetilde{X}\rangle= \widetilde{\Theta}_0[\widetilde{X}_{\ell\mathcal{C}_0}\oplus 0] = \fint \langle \mu(\tih(\theta),\tp(\theta),\trho(\theta),\tchi(\theta)),\widetilde{X}(\theta)\rangle\,d\theta\label{useful_prop3}
\end{align}
for each $\widetilde{X}\in\ell\mathfrak{g}$ is an $\text{Ad}^*$-equivariant presymplectic momentum map with respect to the presymplectic form $-\mathbf{d}\widetilde{\Theta}_0$. (Note that we have used the same notation for the pairings between $\mathfrak{g},\mathfrak{g}^*$ and $\ell\mathfrak{g},\ell\mathfrak{g}^*$.) In other words, we have
\begin{align}
\mathbf{d}\langle\widetilde{\mu}_0,\widetilde{X}\rangle = - \iota_{\widetilde{X}_{\ell\mathcal{C}_0}\oplus 0}\mathbf{d}\widetilde{\Theta}_0,\label{presymplectic_mom_map}
\end{align}
for each $\widetilde{X}\in\ell\mathfrak{g}$.

The pushforward of Eq.\,\eqref{presymplectic_mom_map} along $T$ is
\begin{align}
\mathbf{d}\langle\widetilde{\mu}_0\circ T^{-1},\widetilde{X}\rangle = -\iota_{T_*(\widetilde{X}_{\ell\mathcal{C}_0}\oplus 0)}\mathbf{d}\widetilde{\Theta}.
\end{align}
But because $\widetilde{\Phi}_{\widetilde{g}} = T\circ(\widetilde{\Phi}_{0\widetilde{g}}\times I)\circ T^{-1}$, the infinitesimal generator of $\widetilde{X}$ with respect to the group action $\widetilde{\Phi}$ is just $\widetilde{X}_{\ell\mathcal{C}_0\times C^\infty(Q,S^1)} = T_*(\widetilde{X}_{\ell\mathcal{C}_0}\oplus 0)$. Therefore $\widetilde{\mu} = \widetilde{\mu}_0\circ T^{-1}$ is a presymplectic momentum map with respect to $-\mathbf{d}\widetilde{\Theta}$. In order to show that $\widetilde{\mu}$ is the same as $\widetilde{\mu}$ given in the statement of the proposition, it is enough to note that Eq.\,\eqref{useful_prop3} implies
\begin{align}
\widetilde{\mu}_0(\tih,\tp,\trho,\tchi,S) = \mu(\tih(\theta),\tp(\theta),\trho(\theta),\tchi(\theta)).
\end{align}
\end{proof}

Theorem \ref{theorem_4} and Proposition \ref{looped_momentum_prop} apply to any (Lie) subgroup of the symmetry group for the action functional \eqref{eq:EP:act} whatsoever. They apply in particular to the group of isometries of the fluid container $Q$, which corresponds to momentum conservation. From the perspective of dissipation-free fluid models, however, a more interesting subgroup is the group of particle relabeling transformations of the reference fluid container $Q_0$. Before applying nonlinear WKB extension, this group of symmetries is responsible for the well-known Kelvin circulation theorem. Let us now use Proposition \ref{looped_momentum_prop} to describe what happens to Kelvin's circulation theorem \emph{after} applying the nonlinear WKB extension procedure.

\begin{proposition}\label{prop4}
The mapping $\widetilde{\mu}: \ell\mathcal{C}_0\times C^\infty(Q,S^1)\rightarrow (\ell\mathfrak{X}(Q_0)\times C^\infty(Q_0,\mathbb{R}))^*$ given by
\begin{align}
\widetilde{\mu}(\tih,\tp,\trho,\tchi) = \bigg(\tih^S_*\left[\tp^S\cdot d\bm{x}\otimes d^3\bm{x} + \mathbf{d}\tchi^S\otimes\trho^S\,d^3\bm{x} \right],\tih^S_*\left[\trho^S\,d^3\bm{x}\right]\bigg)
\end{align}
is a $(\ell\mathfrak{X}(Q_0)\times C^\infty(Q_0,\mathbb{R}))^*$-valued first-integral of the extLBEP equations.
\end{proposition}

\begin{corollary}(Kelvin's theorem for Eulerian WKB)\label{corollary_1}
Given a family of closed curves $C_0(\theta)\subset Q_0$ parameterized by $\theta\in S^1$ and a solution $(\tih,\tp,\trho,\tchi)$ of the extLBEP equations, the integral
\begin{align}
\oint _{C(\theta)} \frac{\tp^S(\theta)}{\trho^S(\theta)} \cdot d\bm{x}\label{circulation_family}
\end{align}
is constant in time for each $\theta\in S^1$, where $C(\theta) = [\tih^S(\theta)]^{-1}(C_0(\theta))$.
\end{corollary}

\begin{remark}
See the remark after Lemma \ref{relabeling_lemma} for the argument that proves this Corollary. 
\end{remark}

Before proving Proposition \ref{prop4}, we will first review the corresponding result for the (pre-WKB extension) LBEP equations that was proved, for instance, in Ref.\,\onlinecite{Cotter_Holm_2012}. For this, we introduce the following group, which contains the particle relabeling group $\text{Diff}(Q_0)$ as a subgroup.
\begin{definition}\label{def:semi_direct}
The infinite-dimensional group $\mathcal{G} = \text{Diff}(Q_0)\ltimes C^\infty(Q_0)$ consists of pairs $(\bm{\eta},\tau)\in \text{Diff}(Q_0)\times C^\infty(Q_0)$ with the group product given by
\begin{align}
(\bm{\eta}_1,\tau_1)*(\bm{\eta}_2,\tau_2) = (\bm{\eta}_1\circ \bm{\eta}_2, \tau_1+\bm{\eta}_{1*}\tau_2).
\end{align}
\end{definition}

\begin{lemma}\label{relabeling_lemma}
There is a right $\mathcal{G} = \text{Diff}(Q_0)\ltimes C^\infty(Q_0)$-action on $\mathcal{C}_0$ that leaves the $1$-form $\Theta$ in Eq.\,\eqref{LBEP_Theta} invariant. The associated momentum map is given by
\begin{align}
\mu(\bm{h},\bm{p},\rho,\chi) = \bigg(\bm{h}_*\left[\bm{p}\cdot d\bm{x}\otimes d^3\bm{x} + \mathbf{d}\chi\otimes\rho\,d^3\bm{x} \right],\bm{h}_*\left[\rho\,d^3\bm{x}\right]\bigg).
\end{align}
\end{lemma}

\begin{remark}
This result was established using Noether's theorem in Ref.\,\onlinecite{Cotter_Holm_2012}. Because the LBEP Hamiltonian $\int_Q \mathcal{H}(\bm{p},\rho,\del\rho)\,d^3\bm{x}$ is $\mathcal{G}$-invariant, standard arguments imply that $\mu$ is constant in time along solutions of Eqs.\,\eqref{eq:extended_first}-\eqref{eq:extended_last}. In particular, because $\bm{h}_*\left[\rho\,d^3\bm{x}\right]$ is constant in time, and $\rho$ is non-vanishing, the $1$-form
\begin{align}
\frac{\bm{h}_*\left[\bm{p}\cdot d\bm{x}\otimes d^3\bm{x} + \mathbf{d}\chi\otimes\rho\,d^3\bm{x} \right]}{\bm{h}_*\left[\rho d^3\bm{x}\right]} = \bm{h}_*\left[\frac{\bm{p}}{\rho}\cdot d\bm{x} + \mathbf{d}\chi\right]
\end{align}
is constant in time. Therefore the integral
\begin{align}
\oint_{C_0}\bm{h}_*\left[\frac{\bm{p}}{\rho}\cdot d\bm{x} + \mathbf{d}\chi\right] = \oint_{\bm{h}^{-1}(C_0)}\frac{\bm{p}}{\rho}\cdot d\bm{x}
\end{align}
is constant in time for any closed curve $C_0\in Q_0$. This is the usual statement of Kelvin's circulation theorem.
\end{remark}

We may now prove \ref{prop4} by directly applying Proposition \ref{looped_momentum_prop} with $G = \mathcal{G}$.

\begin{proof}[proof of Proposition \ref{prop4}]
By Lemma \ref{relabeling_lemma}, there is a right $\mathcal{G}$-action on $\mathcal{C}_0$ that preserves $\Theta$ and admits an $\text{Ad}^*$-equivariant momentum map.  Proposition \ref{looped_momentum_prop} therefore implies that
\begin{align}
\widetilde{\mu}(\tih,\tp,\trho,\tchi,S)(\theta) &= \mu(\tih^S(\theta),\tp^S(\theta),\trho^S(\theta),\tchi^S(\theta)) \nonumber\\
& = \bigg(\tih^S(\theta)_*\left[\tp^S(\theta)\cdot d\bm{x}\otimes d^3\bm{x} + \mathbf{d}\tchi^S(\theta)\otimes\trho^S(\theta)\,d^3\bm{x} \right],\tih^S(\theta)_*\left[\trho^S(\theta)\,d^3\bm{x}\right]\bigg)
\end{align}
defines a presymplectic $\text{Ad}^*$-equivariant momentum map on $(\ell\mathcal{C}_0\times C^\infty(Q,S^1),-\mathbf{d}\widetilde{\Theta})$. Moreover, because the Hamiltonian functional $\fint\int _Q \mathcal{H}(\tp,\trho,\del^S\trho)\,d^3\bm{x}\,d\theta$ is $\ell\mathcal{G}$-invariant, it follows that $\widetilde{\mu}$ is constant in time along solutions of Eqs.\,\eqref{extension_labels}-\eqref{extension_chi}.
\end{proof}

We will now conclude this Section by giving a group-theoretic explanation for the $1$-way coupling between $\tp,\trho,S$ and $\tih,\tchi$ in the Euler-Lagrange equations associated with the action functional \eqref{extLBEP_action}. Because $\ell\mathcal{G}$ leaves the action functional \eqref{extLBEP_action} invariant, solutions of the corresponding Euler-Lagrange equations are mapped into other solutions by $\ell\mathcal{G}$. The following proposition shows that the quotient of the space of solutions of the Euler-Lagrange equations equations by $\ell\mathcal{G}$ may be identified with the space of a solutions of the extLBEP equations. This ``explains" the one-way coupling as a consequence of $\ell\mathcal{G}$-invariance.

\begin{proposition}
Let $\widetilde{\mathcal{C}}_{\widetilde{\mathsf{A}}}$ denote the space of solutions of the Euler-Lagrange equations associated with the action functional \eqref{extLBEP_action}. Let $\widetilde{\mathcal{C}}_{\text{extLBEP}}$ denote the space of solutions of the extLBEP equations. There is a canonical bijection
\begin{align}
\widetilde{\mathcal{C}}_{\widetilde{\mathsf{A}}}/\ell\mathcal{G} \approx \widetilde{\mathcal{C}}_{\text{extLBEP}}.
\end{align}
\end{proposition}
\begin{proof}
According to Theorem \ref{theorem_4} (see Eq.\,\eqref{transformed_G_action}), the right action of $\ell\mathcal{G}$ on $\ell\widetilde{\mathcal{C}}_0\times C^\infty(Q,S^1) $ that leaves the action $\widetilde{\mathsf{A}}$ invariant is given by
\begin{align}
\widetilde{\Phi}\bigg((\tih,\tp,\trho,\tchi,S),(\widetilde{\bm{\eta}},\widetilde{\tau})\bigg) = \bigg([\widetilde{\bm{\eta}}^{-1}\circ\tih^S]^{-S},\tp,\trho,\tchi+[(\tih^S)^*\widetilde{\tau}]^{-S},S\bigg).
\end{align}
Apparently the quotient of $\ell\widetilde{\mathcal{C}}_0\times C^\infty(Q,S^1) $ by $\ell\mathcal{G}$ may be identified with triples $(\tp,\trho,S)\in \ell\mathfrak{X}(Q)\times \ell C_+^\infty(Q)\times C^\infty(Q,S^1)$ using the quotient map $\pi : \ell\widetilde{\mathcal{C}}_0\times C^\infty(Q,S^1)\rightarrow  \ell\mathfrak{X}(Q)\times \ell C_+^\infty(Q)\times C^\infty (Q,S^1): (\tih,\tp,\trho,\tchi,S)\mapsto (\tp,\trho,S)$.

If $t\mapsto \Gamma(t) = (\tih(t),\tp(t),\trho(t),\tchi(t),S(t))$ is a solution of the Euler-Lagrange equations associated with $\widetilde{\mathsf{A}}$, i.e. Eqs.\,\eqref{extension_labels}-\eqref{extension_chi}, then $\gamma = \pi\circ \Gamma$ is a solution of the extLBEP equations because the extLBEP equations are a subset of the Euler-Lagrange equations. Thus there is a mapping $\Pi:\widetilde{\mathcal{C}}_{\widetilde{\mathsf{A}}}\rightarrow \widetilde{\mathcal{C}}_{\text{extLBEP}}$. If we can show that $\Pi$ is in fact a quotient map for the $\ell\mathcal{G}$-action on $\widetilde{\mathcal{C}}_{\widetilde{\mathsf{A}}}$, the proof will be complete.

To that end, suppose that $t\mapsto \gamma(t) = (\tp(t),\trho(t),S(t))$ is a solution of the extLBEP equations. Given $(\tih_0,\tchi_0)\in \ell\text{Diff}(Q,Q_0)\times \ell C^\infty(Q)$, the method of characteristics gives a unique curve $t\mapsto (\tih(t),\tchi(t))$ with $(\tih(0),\tchi(0)) = (\tih_0,\tchi_0)$ satisfying Eqs.\,\eqref{extension_labels} and \eqref{extension_chi} with $S$ and the derivatives of $\mc{H}_{\text{EP}}$ evaluated along the solution $\gamma$. Therefore the mapping $\Pi$ is surjective, and the preimage of $\gamma$ under $\Pi$ may be identified with the space of initial values $(\tih_0,\tchi_0)\in \ell\text{Diff}(Q,Q_0)\times \ell C^\infty(Q)$. The latter space is an entire $\mathcal{G}$-orbit in $\widetilde{\mathcal{C}}_{\widetilde{\mathsf{A}}}$ for if $(\tih_0,\tchi_0)$ and $(\tih_0^\prime,\tchi_0^\prime)$ are two elements of $\Pi^{-1}(\{\gamma\})$, then $(\tih_0^\prime,\tchi_0^\prime) = \widetilde{\Phi}\bigg( (\tih_0,\tchi_0),(\widetilde{\bm{\eta}},\widetilde{\tau})\bigg)$ provided we set
\begin{align}
\widetilde{\bm{\eta}} &= \tih_0^S\circ [(\tih_0^\prime)^S]^{-1}\\
\widetilde{\tau}& = (\tih_0^S)_*(\tchi_0^\prime - \tchi_0 )^S.
\end{align}
It follows that $\Pi$ is a quotient map for the $\mathcal{G}$-action on $\widetilde{\mathcal{C}}_{\widetilde{\mathsf{A}}}$.
\end{proof}
\section{Example: Eulerian variational NL-WKB for isothermal fluids}
\label{sec:six}

In this section, we present a pedagogical example of how the methods developed so far can be useful for obtaining reduced, asymptotic models describing wave--mean-flow interactions. Specifically, here we study the time-averaged interaction between a small-amplitude, high-frequency acoustic wave and a slowly-varying isothermal perfect fluid. Due to the somewhat involved calculations given in this section, we present our results in three parts. In the first part, we introduce the governing Eulerian equations of motion and perform an intuitive asymptotic expansion up to leading order in some asymptotic parameter. We give an elementary proof that the resulting leading-order equations describing wave--mean-flow interactions are variational. In the second part, we combine the theory developed in this work with results from slow-manifold theory\cite{Fenichel_1979, Verhulst:2005et} to explain why the wave--mean-flow equations ought to be variational. Finally, in the third section, we present additional details of a systematic derivation of the variational principle describing wave--mean-flow interactions in isothermal fluids.

\subsection{Governing equations for isothermal fluids and intuitive asymptotic expansion}
\label{sec:six_intuitive}

The governing equations for an isothermal fluid are given by
\begin{gather}
	\partial_t\rho +\del\cdot \vec{p}	 = 0
		\label{eq:example_cont}\\
	\partial_t\vec{p} 
		+ 	\del \cdot \left(  \frac{\vec{p} \otimes \vec{p} }{\rho} \right) =	- c_s^2 \del   \rho ,
		\label{eq:example_mom}
\end{gather}
where $c_s\in \mathbb{R}$ is the sound speed. Since we are interested in studying the effects of a high-frequency acoustic wave, let us explicitly introduce a scale separation into the equations. To do this, we use the NL-WKB extension of the equations above. Hence, we write
\begin{gather}
	\partial_t^{S/\ep} \trho +\del^{S/\ep} \cdot \tp	 = 0
		\label{eq:example_ext_cont}	\\
	\partial_t^{S/\ep} \tp	+ 	\del^{S/\ep} \cdot \left(  \frac{\tp \otimes \tp }{\trho} \right) 
		=	- c_s^2 \del^{S/\ep}   \trho .
		\label{eq:example_ext_mom}
\end{gather}
In the above, we have explicitly denoted the scale separation by rescaling the phase function $S$ such that $S\mapsto S/\ep$, where $\ep\ll 1$ is a small dimensionless parameter that represents the ratio of the wave period (or wavelength) to the characteristic timescale (or length scale) of the mean flow.

Since we are interested in linear, small-amplitude waves, we parameterize the density and momentum-density fields as follows:
\begin{equation}
	\fluctrho =\meanrho + \ep \hatrho, 	
	\quad \quad
	 \fluctp = \meanp+ \ep \hatp.
	 \label{eq:example:para}
\end{equation}

Here $\meanrho$ and $\meanp$ respectively represent the slowly-varying density and momentum-density fields of the background fluid. Note that $\meanrho$ and $\meanp$ are independent of $\theta$ and thus are the $\theta$-averaged fields. In contrast, $\hatrho$ and $\hatp$ are the fluctuating density and momentum-density fields, respectively. (To make the above parameterization unique, we assume that the $\theta$-averages of $\hatrho$ and $\hatp$ are zero.) Since we consider small-amplitude waves, we scale the amplitude of the fluctuations according to the small parameter $\ep \ll 1$ in \Eq{eq:example:para}. With these asumptions, we can then deduce the following.
\begin{proposition}\label{prop:eigen}
	To lowest order in $\ep$, the fields $(\hatp,S)$ introduced in \Eqs{eq:example_ext_cont} and \eq{eq:example_ext_mom} satisfy
	\begin{gather}
		\pd_t S \,\pd_\theta \hatrho +\del S \cdot \pd_\theta \hatp	 = 0
			\label{eq:example_cont_osc}		\\
		\left[ \left( \pd_t S + \frac{\meanp }{\meanrho} \cdot \del S \right) \mathbb{I} + \frac{\meanp \otimes \del S}{\meanrho} \right] \cdot \pd_\theta \hatp
				=	\left( \frac{\del S \cdot \meanp}{\meanrho^2} \, \meanp -c_s^2 \del S \right)  \pd_\theta \hatrho .
			\label{eq:example_mom_osc}	
	\end{gather}
	Solutions $(\hatp, S)$ corresponding to linear acoustic oscillations are given by
	\begin{gather}
		\hatp = \meanp \, \frac{\hatrho}{\meanrho} +c_s \hatrho \vec{e}_{\vec{k}} 
			\label{eq:example_hatp_sol} \\
		\pd_t S + \frac{\meanp}{\meanrho} \cdot \del S + c_s |\del S | =0 ,
			\label{eq:example_S_sol}
	\end{gather}
	where $\vec{e}_{\vec{k}}  = \nabla S/|\nabla S|$.
\end{proposition}

\begin{proof}
	After inserting \Eqs{eq:example:para} into \Eq{eq:example_ext_cont}, one can see that the $\theta$-dependent part must satisfy \eq{eq:example_cont_osc} to lowest order in $\ep$. Also when inserting \Eqs{eq:example:para} into \Eq{eq:example_ext_mom}, one finds that the averaged momentum density $\meanp$ satisfies
	\begin{equation}
		\pd_t \meanp + \del \cdot \overline{  \left(  \frac{\tp \otimes \tp }{\trho}  \right) }	=	- c_s^2 \del   \meanrho ,
		\label{eq:mom_aux}
	\end{equation}
	where the overline denotes an average over the $\theta$ variable; \eg $\overline{\hatrho \, \hatp} = \fint_{S^1}  \hatrho \hatp\, \mathrm{d}\theta $. Subtracting \Eq{eq:mom_aux} from \Eq{eq:example_ext_mom} leads to \Eq{eq:example_mom_osc} to lowest-order in $\ep$.

	Solutions corresponding to acoustic oscillations are obtained by projecting \Eq{eq:example_mom_osc} by $\del S$. Then, we substitute \Eq{eq:example_cont_osc} and obtain
	\begin{equation}
		\left( \left( \pd_t S + \frac{\meanp}{\meanrho} \cdot \del S \right)^2 - c_s^2 |\del S|^2  \right) \pd_\theta \hatrho =0.
	\end{equation}
	For acoustic waves, $\hatrho \neq 0$. Hence, the term inside the brackets must be zero. Taking the positive root leads to the dispersion relation in \Eq{eq:example_S_sol}. Then, it is straightforward to verify that the expression for $\hatp$ given in \Eq{eq:example_hatp_sol} satisfies \Eqs{eq:example_cont_osc} and \eq{eq:example_mom_osc}.
\end{proof}

\begin{remark}
It is to be noted that the lowest-order (in $\ep$) equations for the density and momentum-density fluctuations [\Eqs{eq:example_cont_osc} and \eq{eq:example_mom_osc}] lead to a time-evolution equation for the phase $S$ (which was not present in the NL WKB extension of the original fluid equations) and a constraint equation for fluctuations in the momentum density so that $\hatp=\widehat{\vec{p}}^\star(\meanrho, \meanp, \hatrho, \del S)$. In order to fully describe the temporal wave dynamics, we must also deduce a time-evolution equation for $\hatrho$. We shall come back to this point later. 
\end{remark}

As is well known from hydrodynamic theory, high-frequency waves can exert a ponderomotive (or time-averaged) force on a slowly-varying bulk fluid. This effect typically appears in the form of a Reynold-stress term in the momentum equation. In the following, we shall deduce the time-evolution equation for $\meanrho$ and $\meanp$ while taking into account the lowest-order corrections due to wave--mean-flow interactions.

\begin{proposition}
	The governing equations for $\meanrho$ and $\meanp$ with leading-order effects due to wave interactions are
\begin{gather}
	\partial_t \meanrho +\del \cdot \meanp	 = 0 
		\label{eq:example_cont_mean}\\
	\partial_t \meanp    
			+ \del \cdot \left( \frac{\meanp \otimes \meanp}{\meanrho} 
										+ \ep^2 \mc{I} \frac{\del S \otimes \del S }{|\del S| }
										 + c_s^2 \, \meanrho \, \mathbb{I} \right) = 0,
		\label{eq:example_mom_mean}
\end{gather}
where $\mc{I}\colon Q \to \mathbb{R}$ is the wave action density
\begin{equation}
	\mc{I} \doteq \fint_{S^1} 
						\meanrho \, \frac{c_s}{|\del S| }
									\left( \frac{\hatrho}{\meanrho} \right)^2		d\theta.
		\label{eq:example_wave_action}
\end{equation}
\end{proposition}

\begin{proof}
	Since \Eq{eq:example_ext_cont} is linear in $\fluctrho$ and $\fluctp$, $\theta$-averaging \Eq{eq:example_ext_cont} immediately leads to \Eq{eq:example_cont_mean}. To obtain \Eq{eq:example_mom_mean}, one first inserts \Eq{eq:example:para} into \Eq{eq:mom_aux}. After Taylor expanding up to $\mc{O}(\ep^2)$, one obtains
	\begin{equation}
		\partial_t \meanp    
			+ \del \cdot \left( \frac{\meanp \otimes \meanp}{\meanrho} 
										+ \ep^2 T_{\rm Reynolds} 
										+ c_s^2 \, \meanrho \, \mathbb{I} \right) = 0,
	\end{equation}
	where $T_{\rm Reynolds}$ is the Reynolds stress tensor
	\begin{equation}
		T_{\rm Reynolds}
		\doteq 		\fint \left[ \frac{\hatp \otimes \hatp}{\meanrho} 
						-	\frac{\hatp \otimes \meanp}{\meanrho} \frac{\hatrho}{\meanrho} 
						-	\frac{\meanp \otimes \hatp}{\meanrho} \frac{\hatrho}{\meanrho} 
						+	\frac{\meanp \otimes \meanp}{\meanrho} \left( \frac{\hatrho}{\meanrho} \right)^2 	\right] d \theta .
		\label{eq:example_Reynolds_stress}
	\end{equation}
	Finally, substituting the expression for $\hatp$ in \Eq{eq:example_hatp_sol} into $T_{\rm Reynolds}$ leads to \Eq{eq:example_mom_mean}.
\end{proof}

With the above equations, we can now deduce a dynamical equation for the wave action density $\mc{I}$. This is given in the next proposition.

\begin{proposition}\label{prop:wact}
The governing equation to leading order for the wave action density $\mc{I}$ is
\begin{gather}
	\pd_t \mc{I} + \del \cdot ( \vec{v}_{g} \mc{I} ) =0 ,
\end{gather}
where $\vec{v}_{\rm g}\doteq   \meanp/\meanrho + c_s \vec{e}_{\vec{k}}$ is the wave group velocity. 
\end{proposition}
\begin{remark}
\noindent One can in principle prove this result by calculating the time derivative of $\mc{I}$ and then substituting the governing equations for the mean and fluctuating quantities [\Eqs{eq:example_cont_osc}, \eq{eq:example_mom_osc}, \eq{eq:example_cont_mean}, and \eq{eq:example_mom_mean}], as well as the dispersion relation \eq{eq:example_S_sol}. The ensuing calculation is tedious since one must take into account corrections to the leading-order solution for the momentum density [\Eq{eq:example_hatp_sol}]. Alternatively, the proof may be constructed as a straightforward corollary of results presented in \Sec{sec:four} and later in the present Section. Because this alternative approach is simpler, we will postpone the proof until we have proved Proposition \ref{prop:fast_slow}.
\end{remark}


We may now summarize the results obtained so far as the following set of equations governing the leading-order wave-mean-flow interaction between a sound wave and a bulk isothermal flow.

\begin{definition}[Wave--mean-flow equations]
	The governing equations describing a high-frequency, small-amplitude acoustic wave interacting with a slowly-varying, isothermal bulk fluid are
	\begin{gather}
		\partial_t \meanrho +\del \cdot \meanp	 = 0 
			\label{eq:example:cont_final}\\
		\partial_t \meanp    
			+ \del \cdot \left( \frac{\meanp \otimes \meanp}{\meanrho} 
										+ \ep^2 \mc{I} \frac{\del S \otimes \del S }{|\del S| }
										 + c_s^2 \, \meanrho \, \mathbb{I} \right) = 0
			\label{eq:example:mom_final} \\
		\pd_t \mc{I} + \del \cdot ( \vec{v}_{g} \mc{I} ) =0 
			\label{eq:example:act_final} \\							 
		\pd_t S + \frac{\meanp}{\meanrho} \cdot \del S + c_s |\del S | =0 ,
			\label{eq:example:phase_final}			
	\end{gather}
	where $\mc{I}$ is the wave action density \eq{eq:example_wave_action} and $\vec{v}_{\rm g}\doteq \meanp/\meanrho + c_s \vec{e}_{\vec{k}}$ is the wave group velocity.
\end{definition}

As written, \Eqs{eq:example:cont_final}--\eq{eq:example:phase_final} are closed in the sense that they possess a (formally) well-posed initial value problem. Perhaps surprisingly, these equations also follow from a variational principle! This is shown in the theorem below.

\begin{theorem}[Effective action for wave--mean-flow interactions]\label{theo:wave_mean_flow}
	Let $\overline{\mathcal{C}}_0 = \mathcal{C}_0 \times C^\infty_+(Q) \times C^\infty(Q,S^1) $. (The space $\mc{C}_0$ is introduced in Definition \ref{frozen_field_def}.) That is, $\overline{\mathcal{C}}$ comprises maps $Q\ni \bm{x}\mapsto (\meanh(\bm{x}),\meanp(\bm{x}),\meanrho(\bm{x}),\meanchi(\bm{x}), \mc{I}(\bm{x}), S(\bm{x}) )\in Q_0\times\mathbb{R}^3\times\mathbb{R}\times\mathbb{R}\times \mathbb{R}\times S^1$, where $\meanh$ is a diffeomorphism and $\meanrho(\bm{x})>0$, $\mc{I}(\bm{x})>0$ for all $\bm{x}\in Q$. Consider the action $\overline{\msf{A}}_{\rm T}$ defined on the space of paths $[t_1,t_2]\rightarrow \overline{\mathcal{C}}_0$ by the formula
	\begin{equation}
			\overline{\msf{A}}_{\rm T}
				= 			\int_{t_1}^{t_2} 
							\overline{\msf{L}}_{\rm T}
								( \meanh, \dot{\meanh}, \meanp, \dot{\meanp}, \meanrho, \dot{\meanrho},  \meanchi  , \dot{ \meanchi  } , 
								\mc{I},\dot{\mc{I}} , S, \dot{S} ) 	d t,
			\label{eq:example_act_reduced}
	\end{equation}
	where the Lagrangian $\overline{\msf{L}}_{\rm T}$ is given by
	\begin{align}
			\overline{\msf{L}}_{\rm T}
				( \meanh, \dot{\meanh}, \meanp, \dot{\meanp},  \meanrho, \dot{\meanrho},  \meanchi  , \dot{ \meanchi  } , 
								\mc{I},\dot{\mc{I}} , S, \dot{S} )
				&	= 		\int_Q 
							\left( 	\meanp \cdot \meanv 	
										+ \meanrho	\big( \pd_t \meanchi	+ \meanv  \cdot   \del    \meanchi  \big) \right) 	\, d^3 \vec{x} 
							- \int_Q 	\mc{H}_{\rm T}(\meanp, \meanrho) 	\, d^3 \vec{x}  \notag \\
				&	\quad 	- \ep^2 	 \int_Q \mc{I}		\left( \pd_t S + \meanv \cdot \del S + c_s |\del S| \right)\,d^3\bm{x}.
		\label{eq:example_lagr_reduced}
	\end{align}
	Here the mean velocity $\meanv$ is defined as $\meanv\doteq - \pd_t \meanh \cdot (\del  \meanh)^{-1}$, and the Hamiltonian for the isothermal fluid is given by
	\begin{equation}
		\mc{H}_{\rm T} (\meanp,\meanrho) \doteq \frac{|\meanp|^2}{2 \meanrho} + c_s^2 \meanrho \ln \left( \frac{\meanrho}{\rho_0} \right). 
		\label{eq:example_ham_isothermal}
	\end{equation}
	The parameters $c_s$ and $\rho_0$ are the sound speed and reference mass density, respectively.
	Equations \eq{eq:example:cont_final}--\eq{eq:example:phase_final} are embedded in the Euler--Lagrange equations obtained when varying the action $\overline{\msf{A}}_{\rm T}$ with respect to $\meanh$, $\meanp$, $\meanrho$, $\meanchi$, $\mc{I}$, and $S$.
\end{theorem}

\begin{proof}
	Since $\overline{\msf{A}}_{\rm T}$ is functional of paths $[t_1,t_2]\rightarrow \overline{\mathcal{C}}_0$, we can vary the fields $(\meanh,\meanp,\meanrho,\meanchi, \mc{I}, S )$ independently. Moreover, varying the action with respect to the fields $(\meanh,\meanp,\meanrho,\meanchi)\in \mc{C}_0$ follows an almost identical procedure as that given in the proof of Proposition \ref{extended_vp}. Varying the action with respect to $\meanp$ leads to $\meanv= \meanp / \meanrho$. Varying the action with respect to the scalar field $\meanchi$ gives
	\begin{equation}
		\pd_t \meanrho + \del \cdot ( \meanv \meanrho)=0.
		\label{eq:example:ELE_meanrho}
	\end{equation}
	Substituting $\meanv= \meanp / \meanrho$ into the equation above trivially leads to \Eq{eq:example:cont_final}. Varying the action with respect to $\meanrho$ gives
	\begin{equation}
		\pd_t \meanchi + (\meanv \cdot \del)  \meanchi
			=  - \frac{|\meanp|^2}{2\meanrho^2} +c_s^2 \ln\left( \frac{\meanrho}{\rho_0} \right) +c_s^2.
		\label{eq:example:ELE_varphi}
	\end{equation}
	Varying $\meanh$ leads to
	\begin{equation}
		\pd_t (\vec{p} + \meanrho \del \meanchi - \mc{I} \del S )
			+ \del \cdot  (\vec{v} \otimes (\vec{p} + \meanrho \del \meanchi - \mc{I} \del S )) 
			+ (\del \bm{v})\cdot (\vec{p} + \meanrho \del \meanchi - \mc{I} \del S ) =0.
		\label{eq:example:ELE_h}
	\end{equation}
	By substituting \Eqs{eq:example:ELE_meanrho} and \eq{eq:example:ELE_varphi} into \Eq{eq:example:ELE_h} and following a similar algebraic manipulation as in the proof of Proposition \ref{extended_vp}, one can recover \Eq{eq:example:mom_final} from \Eq{eq:example:ELE_h}. Finally, varying the action $\overline{\msf{A}}_{\rm T}$ with respect to $\mc{I}$ and $S$ leads to \Eqs{eq:example:act_final} and \eq{eq:example:phase_final}, respectively. Thus, we have shown that \Eqs{eq:example:cont_final}--\eq{eq:example:phase_final} are embedded in the Euler--Lagrange equations associated to the action $\overline{\msf{A}}_{\rm T}$.
\end{proof}

\subsection{Why are the wave--mean-flow equations variational?}
\label{sec:six_why}

The proof of Theorem \ref{theo:wave_mean_flow} gave no indication as to why the leading-order wave-mean-flow equations arise from a variational principle. We now want to give a principled explanation for this result using the machinery developed in this paper.
One indication that the wave--mean-flow equations might be variational is that the parent isothermal fluid equations \eq{eq:example_cont} and \eq{eq:example_mom}, where we started our asymptotic analysis, also come from a variational principle. This can be easily proven because the isothermal fluid equations comprise a special case of the LBEP equations discussed in \Sec{sec:three}. Hence, we can readily write a variational principle for \Eqs{eq:example_cont} and \eq{eq:example_mom}, which is given below.

\begin{corollary}
		Let $\mathsf{A}_{\rm T}$ be the action functional defined on the space of paths $[t_1,t_2]\rightarrow \mathcal{C}_0$ such that
	\begin{align}
		\mathsf{A}_{\rm T}
			= \int_{t_1}^{t_2}\mathsf{L}_{\rm T}(\bm{h}(t),\dot{\bm{h}}(t),\bm{p}(t),\dot{\bm{p}}(t),\rho(t),\dot{\rho}(t),\chi(t),\dot{\chi}(t))\,dt,
		\label{eq:example_act}
	\end{align}
	where 
	\begin{align}
		\mathsf{L}_{\rm T}(\bm{h},\dot{\bm{h}},\bm{p},\dot{\bm{p}},\rho,\dot{\rho},\chi,\dot{\chi}) 
			= \int_Q \bm{p}\cdot\bm{v}\,d^3\bm{x} 
				+ \int_Q (\dot{\chi}+\bm{v}\cdot\del\chi)\,\rho\,d^3\bm{x}
				-\int_Q \mathcal{H}_{\text{T}}(\bm{p},\rho)\,d^3\bm{x}
		\label{eq:example_lagr}
	\end{align}
	and the Hamiltonian $\mathcal{H}_{\text{T}}$ is defined in \Eq{eq:example_ham_isothermal}.
	
	A path $t\mapsto (\bm{h}(t),\bm{p}(t),\rho(t),\chi(t))\in \mathcal{C}_0$ is a critical point of the action functional \eq{eq:example_act} for isothermal fluids if and only if $\bm{h}$, $\bm{p}$, $\rho$, and $\chi$ satisfy the following system of PDEs:
	\begin{gather}
	\partial_t\bm{h}=-\frac{\vec{p}}{\rho}	\cdot		\del\bm{h}
		\label{eq:ELE_v}\\
	\partial_t\vec{p} 
			+ 	\del \cdot \left(  \frac{\vec{p} \otimes \vec{p} }{\rho} \right) 
			=	- c_s^2 \del   \rho 
		\label{eq:ELE_mom}\\
	\partial_t\rho +\del\cdot \vec{p}	 = 0
		\label{eq:ELE_cont}\\
	\partial_t \chi + \frac{\vec{p} }{\rho}\cdot \del  \chi 
			=  	- \frac{|\vec{p}|^2}{2\rho^2} +c_s^2 \ln \left( \frac{\rho}{\rho_0} \right) +c_s^2.
		\label{eq:ELE_chi}
	\end{gather}	
\end{corollary}

\begin{proof}
	This is directly verified by substituting \Eqs{eq:example_act}--\eq{eq:example_lagr} into \Eqs{eq:extended_first}--\eq{eq:extended_last} in Proposition~\ref{extended_vp}.
\end{proof}

\begin{remark}
It is clear that the isothermal fluid equations \eq{eq:example_cont} and \eq{eq:example_mom} are simply embedded into the LBEP equations above. 
\end{remark}

The next indication is that, in order to explicitly introduce a scale separation into the fluid equations, the next step we used in the analysis of \Sec{sec:six_intuitive} was passing to the NL-WKB extension of the isothermal fluid equations [see \Eqs{eq:example_ext_cont} and \eq{eq:example_ext_mom}].  Following our results from \Sec{sec:four}, these equations are variational as well! The precise statement of this observation is as follows.

\begin{corollary}\label{cor:isothermal_variational}
	Let $\gamma:[t_1,t_2]\rightarrow \ell\mc{C}_0\times C^\infty(Q,S^1)$ be a smooth curve with components $\gamma = (\tih,\tp,\trho,\tchi,S)$. Let the functional $\widetilde{\mathsf{A}}(\gamma) = \int_{t_1}^{t_2}\widetilde{\mathsf{L}}_{\rm T}(\gamma(t),\partial_t\gamma(t))\,dt$ be defined such that
	\begin{align}
	\widetilde{\mathsf{L}}_{\rm T}(\tih,\tp,\trho,\tchi,S,\dot{\tih},\dot{\trho},\dot{\tchi},\dot{S}) 
			=	& \fint\int_Q \left(\tp\cdot \tv + \trho\,\dot{\tchi}+ \ep^{-1} \trho\,\dot{S}\partial_\theta\tchi 
					+ \trho\tv\cdot\del^{S/\ep}\tchi \right)\,d^3\bm{x}\,d\theta\nonumber\\
				&-\fint\int_Q \mc{H}_{\rm T}(\tp,\trho)\,d^3\bm{x}\,d\theta,
	\label{eq:extagr_isothermal}
\end{align}
	where $\widetilde{\bm{v}}$ is defined in Eq.\,\eqref{tilde_v_defined}. Then, the curve $\gamma$ is a (fixed-endpoint) critical point of $\widetilde{\mathsf{A}}(\gamma)$ if and only if the component functions $(\tih,\tp,\trho,\tchi,S)$ satisfy
	\begin{gather}
		\partial_t^{S/\ep} \tih=-\frac{\tp}{\trho}	\cdot		\del^{S/\ep} \tih
			\label{eq:extELE_v}	\\
		\partial_t^{S/\ep} \tp	+ 	\del^{S/\ep} \cdot \left(  \frac{\tp \otimes \tp }{\trho} \right) 
			=	- c_s^2 \del^{S/\ep}   \trho 
			\label{eq:extELE_mom}\\
		\partial_t^{S/\ep} \trho +\del^{S/\ep} \cdot \tp	 = 0
			\label{eq:extELE_cont}	\\
		\partial_t^{S/\ep} \tchi + \frac{\tp }{\trho}\cdot \del^{S/\ep}  \tchi 
			=  	- \frac{|\tp|^2}{2\trho^2} +c_s^2 \ln \left( \frac{\trho}{\rho_0} \right) +c_s^2.
			\label{eq:extELE_chi}
	\end{gather}
\end{corollary}

\begin{proof}
	This is immediately verified by substituting \Eq{eq:extagr_isothermal} into the result in Theorem~\ref{theorem_3}.
\end{proof}

The question that now remains to be answered is whether the variational structure underlying the NL-WKB extension of the isothermal fluid equations is somehow compatible with the asymptotics leading to the wave--mean-flow equations \eq{eq:example:cont_final}--\eq{eq:example:phase_final}. A geometrically satisfying way to address this question is through the application of a dynamical systems tool known as \emph{slow-manifold reduction}.

The concept of slow manifolds originated from the theory of fast-slow dynamical systems, which essentially are singularly perturbed dynamical systems.\cite{Fenichel_1979, Verhulst:2005et} Before explaining the role played by slow manifolds in our example, we will first give a quick overview of slow manifold theory.

\begin{definition}[Fast-slow dynamical system]\label{def:slow_fast}
	Let $X,Y$ be Banach spaces and $\ep \ll1$. A fast-slow dynamical system is an ODE on $X\times Y$ of the form
	\begin{equation}
		\ep \dot{y} = f_\epsilon(x,y), \qquad
		\dot{x} = g_\epsilon(x,y),
	\end{equation}
	with $D_y f_0(x,y)\colon Y \to Y $ an isomorphism when $(x,y) \in f_0^{-1}(\{0\})$. The functions $f_\epsilon$ and $g_\epsilon$ are required to depend smoothly on $\epsilon$ in such a manner that $f_\epsilon,g_\epsilon = O(1)$ as $\epsilon\rightarrow 0$.
\end{definition}

By convention, the variable $y$ is called the ``fast" variable, while the variable $x$ is called the ``slow" variable. For fast-slow dynamical systems, it then follows that invariant manifolds given as graphs over the slow variables satisfy a nonlinear (functional) PDE. This is illustrated below.

\begin{lemma}
	Suppose a fast-slow dynamical systems admits an invariant manifold $I_{\ep}$ of the form $I_\ep = \{ (x,y) \in X \times Y | y = y_\ep^\star(x)\}$ for some smooth map $y_\ep^\star\colon X \to Y$. Then, 
	\begin{equation}
		\ep Dy^\star_\ep(x)[g(x,y^\star_\ep(x))] = f(x,y^\star_\ep(x)),\label{invariance_equation}
	\end{equation}
	for each $x\in X$. 
\end{lemma}

\begin{proof}
	Supposing $y = y_\ep^\star(x)$, one then inserts this into $\ep \dot{y} = f(x,y)$. Using the chain rule and substituting the time-evolution equation for $x$ leads to the claimed result.
\end{proof}

\begin{definition}[Slow manifold]
	If $I_\ep$ is an invariant manifold given as the graph of $y^\star_\ep \colon X \to Y$, $I_\ep$ is a slow manifold when $y^\star_\ep(x)$ is a formal power series solution of Eq.\,\eqref{invariance_equation}.
\end{definition}


Of the invariant manifolds given as graphs, slow manifolds play a special role for several reasons. First of all, slow manifolds are unique; \ie if $I_\ep$ and $I_\ep'$ are two slow manifolds, then $I_\ep = I_\ep'$. Moreover, the formal power series expansion of the graphing function $y^\star_\ep(x)$ may be obtained using explicit formulas. In addition, dynamics restricted to the slow manifold is indeed slow; it is simple to check that the time derivatives of both the fast and slow variables are $O(1)$ on the slow manifold. The slow manifold may therefore be interpreted intuitively as the region in phase space where the fast degrees of freedom are not excited.

For the purposes of the present discussion, a crucial result on slow-manifold dynamics is that they inherit Hamiltonian structure from the parent fast-slow system whenever the larger system has such a structure. One way to state this fact precisely is as follows.

\begin{theorem}[Inheritance of Hamiltonian structure]\label{theo:slow_manifold_reduction}
	Consider a fast-slow system satisfying the variational principle
	\begin{equation}
		\delta A 
			= \delta \int_{t_1}^{t_2} \bigg( 
				\Theta_\ep(x(t),y(t)) [\dot{x}(t),\dot{y}(t)] - H_\ep(x(t),y(t)) \bigg) dt = 0,
	\end{equation}
	with $\delta(x(t_1),y(t_1)) = \delta(x(t_2),y(t_2)) =0$. Here $\Theta_\ep$ is an $\epsilon$-dependent one-form on $X\times Y$, and $H_\ep$ is an $\epsilon$-dependent smooth function on $X\times Y$.  Suppose a slow manifold exists where $y=y^\star_\ep(x)$. Then, the slow dynamics for the variable $x\in X$ satisfy the variational principle
	\begin{equation}
		\delta A_{\rm slow} 
			= \delta \int_{t_1}^{t_2} \left( 
					\Theta_{\rm slow}(x(t)) [\dot{x}(t)] - H_{\rm slow}(x(t)) \right) dt = 0,
	\end{equation}
	with $\delta x(t_1) = \delta x(t_2) =0$. Here $\Theta_{\rm slow}(x) [\delta x] \doteq  \Theta(x,y_\ep^\star(x)) [\delta x, Dy_\ep^\star(x)[\delta x]]$, and $H_{\rm slow}(x) \doteq H(x,y_\ep^\star(x))$.
\end{theorem}

\begin{proof}
	With the given boundary conditions, it is clear that $\delta A =0$ holds even if the trajectory $t\mapsto (x(t),y(t))$ that is subject to variations lies in the slow manifold, i.e. $y(t) = y^\star_\ep (x(t))$, since $\delta y(t_{1,2})= Dy^\star_\ep (x(t_{1,2})) \, \delta x(t_{1,2})=0$. In particular, $\delta A = 0$ when variations are constrained to lie along the slow manifold. This is equivalent to saying that the first variation of $A_{\text{slow}}$, which is $A$ restricted to paths contained in the slow manifold, is zero along a solution contained in the slow manifold. After restriction to the slow manifold, the two terms in the integrand of $A$ may be written	%
	\begin{equation}
		\Theta_\ep(x,y_\ep^\star(x)) \left[\dot{x}, \frac{d}{dt} y^\star_\ep(x) \right]
			=	\Theta_\ep(x,y_\ep^\star(x)) \left[\dot{x}, Dy^\star_\ep(x) \, \dot{x}\right]
			=	\Theta_{\rm slow}(x) [\dot{x}],
	\end{equation}
	and $H_\ep (x,y_\ep^\star(x)) = H_{\text{slow}}(x)$. (We have omitted writing the time dependence explicitly.) Thus, $A_{\text{slow}}$ may be written as in the Theorem statement. Moreover, we have already argued $\delta A_{\text{slow}} = 0$ along any solution of the fast-slow system contained in the slow manifold. This completes the proof.
\end{proof}

We will now argue that Theorem \ref{theo:slow_manifold_reduction} may be used to systematically derive the variational principle for the leading-order wave--mean-flow equations \eq{eq:example:cont_final}--\eq{eq:example:phase_final}.
The first thing to be verified is that the NL-WKB extension of the isothermal fluid equations \eq{eq:extELE_v}--\eq{eq:extELE_chi}, together with the dispersion relation \eqref{eq:example:phase_final} for specifying the dynamics of $S$, indeed form a fast-slow system. If this is the case, then by Theorem \ref{theo:slow_manifold_reduction}, slow-manifold reduction will allow us to construct a variational principle for the slow, wave--mean-flow system. Our argument will then be complete if we can show that the variational principle given by Theorem \ref{theo:slow_manifold_reduction} reproduces the variational principle from Theorem \ref{theo:wave_mean_flow}. The rest of this subsection will be devoted to establishing that Eqs. \eq{eq:extELE_v}--\eq{eq:extELE_chi}, together with the dispersion relation \eqref{eq:example:phase_final}, comprise a fast-slow system. The following subsection will sketch the details of manipulating $A_{\text{slow}}$ from Theorem \ref{theo:slow_manifold_reduction} in order to produce $\overline{\msf{A}}_{\rm T}$ from Theorem \ref{theo:wave_mean_flow}.

As in the previous section, we consider only high-frequency, small-amplitude waves. Hence, we adopt the parameterization given in \Eqs{eq:example:para} for the fields $\fluctrho$ and $\fluctp$. Although it is not technically necessary for the slow-manifold analysis, we shall also parameterize the back-to-labels map $\flucth$ by following the generalized-Lagrangian-mean (GLM) approach proposed by Andrew and McIntyre.\cite{Andrews:1978fg} For more information on this approach, we recommend reading as well the works by Holm and Gjaja,\cite{Holm:2002ju,Holm:2002kh,Gjaja_Holm_1996,Holm:2002ex} as well as Buhler's accessible book.\cite{Buhler_2009} In GLM theory, one introduces a space $\overline{Q}$ that is diffeomorphic to $Q$ and that is interpreted as the collection of ``mean" Eulerian positions. Then $\flucth$ is written as the composition of a mean component $\meanh\in\text{Diff}(\overline{Q},Q_0)$ and a fluctuating component $\hattau\in \text{Diff}(Q,\overline{Q})$, i.e. $\flucth = \meanh \circ \hattau$. Additionally, and in order to uniquely specify $\hattau$, we consider $\hattau$ to be a near-identity transformation of the form
\begin{equation}
	\hattau(\bm{x}) =\bm{x}  + \ep^2 \hatalpha(\bm{x}),
	\label{eq:example:para_II}
\end{equation}
where $\hatalpha : Q\rightarrow \mathbb{R}^3$ satisfies $\fint \hatalpha \,d\theta = 0$. Finally, we shall parameterize the Lagrange multiplier $ \tchi $ according to
\begin{equation}
	\tchi = \overline{\chi} + \ep^2 \hatchi,
	\label{eq:example:para_III}
\end{equation}
where $\fint \hatchi\,d\theta = 0$.

\begin{proposition}\label{prop:fast_slow}
	With the parameterizations given in \Eqs{eq:example:para}, \eq{eq:example:para_II}, and \eq{eq:example:para_III}, \Eqs{eq:extELE_v}--\eq{eq:extELE_chi}, together with the dispersion relation \eqref{eq:example:phase_final}, are equivalent to a fast-slow dynamical system.
\end{proposition}

\begin{proof}
	We begin by inspecting the equations of motion for the mean fields. Upon $\theta$-averaging \Eqs{eq:extELE_v}--\eq{eq:extELE_chi}, we immediately obtain
	\begin{gather}	
		\pd_t \meanh= -  \frac{\meanp}{\meanrho}	\cdot		\del  \meanh + \mc{O}(\ep^2)	
			\label{eq:example_slowfast_meanh}  \\
		\partial_t \meanp	=	- 	\del \cdot \overline{\left(  \frac{\tp \otimes \tp }{\trho} \right) }	- c_s^2 \del \meanrho 	\\
		\partial_t \meanrho  = - \del \cdot \meanp				\\
		\partial_t \meanchi 
			=  - \overline{ \left( \frac{\tp }{\trho} \right) } \cdot \del  \meanchi  
				- \ep  \overline{ \frac{\tp }{\trho}\cdot \del S \pd_\theta  \hatchi } 
				- \overline{\frac{|\tp|^2}{2\trho^2}} 
				+c_s^2 \overline{ \ln \left( \frac{\trho}{\rho_0} \right) } +c_s^2,	
			\label{eq:example_slowfast_meanchi} 
	\end{gather}
	where in the first equation we used $\overline{ \partial_t^{S/\ep} \tih } = \pd_t \meanh + \mc{O}(\ep^4)$ and $\overline{ (\tp /\trho)	\cdot		\del^{S/\ep} \tih	 } = - (\meanp / \meanrho) \cdot \del \meanh + \mc{O}(\ep^2)$. We have also introduced the shorthand notation $\overline Q = \fint Q\,d\theta$ for denoting averages over $\theta$. When comparing \Eqs{eq:example_slowfast_meanh}--\eq{eq:example_slowfast_meanchi} to Definition \ref{def:slow_fast}, it so far seems that the variables $(\meanh, \meanp, \meanrho,\meanchi) \in \mc{C}_0$ should be included amongst the slow variables. Additionally, according to the dispersion relation \eqref{eq:example:phase_final}, the time derivative $\partial_t S = O(1)$, suggesting that $S$ should be a slow variable.
		
Let us next examine the dynamical equations for the fluctuating quantities. A straightforward calculation leads to
\begin{gather}
	\ep \pd_t \hatalpha =
				-	\left( \pd_t S + \frac{\meanp}{\meanrho} \cdot \del S \right) \pd_\theta \hatalpha
				- 	\left(  \frac{\hatp}{\meanrho} - \frac{\meanp}{\meanrho} \frac{\hatrho}{\meanrho} \right)
				+ \mc{O}(\ep)
		\label{eq:example_slowfast_flucth}	\\
	\ep \pd_t  \hatp 
			= 	- 	\left[ \left( \pd_t S + \frac{\meanp }{\meanrho} \cdot \del S \right) \mathbb{I} 
				+ 	\frac{\meanp \otimes \del S}{\meanrho} \right] \cdot \pd_\theta \hatp
				+	\left( \frac{\del S \cdot \meanp}{\meanrho^2} \, \meanp -c_s^2 \del S \right)  \pd_\theta \hatrho 	
				+\mc{O}(\ep) \label{hatp_evol}\\
	\ep \pd_t \hatrho = - \pd_t S \,\pd_\theta \hatrho -\del S \cdot \pd_\theta \hatp	 + \mc{O}(\ep)	\\
	\ep \pd_t \hatchi
		=	-	\left( \pd_t S + \frac{\meanp }{\meanrho} \cdot \del S\right) \pd_\theta \hatchi 
			- 	\left(	\frac{\hatp}{\meanrho} - \frac{\meanp}{\meanrho} \frac{\hatrho}{\meanrho} \right) \cdot \del  \meanchi
			- 	\frac{\meanp \cdot \hatp}{\meanrho^2}  	+ \frac{|\meanp|^2}{\meanrho^2}  \frac{\hatrho}{\meanrho}+c_s^2 \, \frac{\hatrho}{\meanrho} 
			+	\mc{O}(\ep),
		\label{eq:example_slowfast_chi}
\end{gather}	
where we have omitted $\mc{O}(\ep)$ terms related to nonlinearities in the fluctuations and $\mc{O}(\ep)$ terms involving spatial derivatives. These omissions are motivated by the fact that, in order to prove the singularly-perturbed dynamical system $\epsilon \dot{y} = f_\epsilon(x,y)$, $\dot{x} = g_\epsilon(x,y)$ is in fact a fast-slow system, it is enough to check that $f_\epsilon,g_\epsilon = O(1)$ and that $D_yf_0$ is invertible along the zero level of $f_0$.
	
At first glance, \Eqs{eq:example_slowfast_flucth}--\eq{eq:example_slowfast_chi} seem to suggest that $\hatalpha,\hatp,\hatrho$ and $\hatchi$ should be fast variables. Indeed, the time derivative of each of these fields is generically $O(\epsilon^{-1})$. However, there happens to be a non-trivial combination of these quantities whose time derivative is $\mc{O}(\ep)$. It is straightforward to verify that the field $\widehat{\lambda} \colon Q \times S_1 \to \mathbb{R}$ given by
\begin{align}\label{lambda_def}
	\widehat{\lambda} 
		= \hatrho + \frac{\hatp\cdot\nabla S}{ c_s |\nabla S| - \meanp\cdot\nabla S/\meanrho}
\end{align}
satisfies $\partial_t\widehat{\lambda} = O(1)$. This suggests that a viable set of slow variables might be $x = (\meanh,\meanp,\meanrho,\meanchi,\widehat{\lambda})$ with corresponding fast variables $y = (\hatalpha,\hatp,\hatchi)$. The rest of the proof will be devoted to showing that, when expressed in terms of $x$ and $y$, \Eqs{eq:extELE_v}--\eq{eq:extELE_chi}, together with the dispersion relation \eqref{eq:example:phase_final}, do in fact comprise a fast-slow dynamical system.

In order to write \Eqs{eq:extELE_v}--\eq{eq:extELE_chi} and the dispersion relation \eqref{eq:example:phase_final} in terms of $x$ and $y$, it is only necessary to exchange the dependent variable $\hatrho$ with the new dependent variable $\widehat{\lambda}$. Because this change of dependent variables is independent of $\epsilon$, our calculations so far already demonstrate that $dx/dt = O(\epsilon)$. In order to prove that we have identified the correct fast and slow variables, we therefore only have to show that $\epsilon dy/dt = f_0(x,y) + O(\epsilon)$ and that $D_y f_0(x,y)$ is invertible along the zero level of $f_0$.

In order to identify $f_0(x,y)$, we substitute the definition of $\widehat{\lambda}$ given by Eq.\,\eqref{lambda_def} into Eqs.\,\eqref{eq:example_slowfast_flucth}, \eqref{hatp_evol}, and \eqref{eq:example_slowfast_chi}, thereby obtaining
\begin{align}
	\epsilon \partial_t \hatalpha  
		= & c_s |\nabla S| \partial_\theta \hatalpha 
				- \left[\mathbb{I}+ \frac{1}{1-\bm{e}_{\bm{k}}\cdot{(\meanp/\meanrho)}/{c_s}}\frac{(\meanp/\meanrho)}{c_s}\otimes \bm{e}_{\bm{k}}\right]\cdot \frac{\hatp}{\meanrho}
				+\frac{\widehat{\lambda}}{\meanrho}\frac{\meanp}{\meanrho}+O(\epsilon)\label{fast_slow_first}\\
\epsilon \partial_t \widehat{\bm{p}} 
		=& \left[c_s |\nabla S|\left(\mathbb{I}
				+\bm{e}_{\bm{k}}\otimes\bm{e}_{\bm{k}}\right)-\frac{(\mathbb{I}
				-\bm{e}_{\bm{k}}\otimes\bm{e}_{\bm{k}})\cdot (\meanp/\meanrho)\otimes \nabla S}{1-\bm{e}_{\bm{k}}\cdot{(\meanp/\meanrho)}/{c_s}}\right]\cdot \partial_\theta\widehat{\bm{p}}\nonumber\\
			&-\nabla S\cdot \left[c_s^2 \mathbb{I}
								- \frac{\meanp\otimes\meanp}{\meanrho^2}\right]\partial_\theta\widehat{\lambda}+O(\epsilon)\\
\epsilon\partial_t\hatchi 
		= & c_s |\nabla S|\partial_\theta\hatchi 
				- \left[\frac{\meanp}{\meanrho}+\nabla\meanchi 
				+ \frac{\left[(\meanp/\meanrho)\cdot\nabla\meanchi 
						+ |\meanp/\meanrho|^2+c_s^2 \right]}{c_s - (\meanp/\meanrho)\cdot\bm{e}_{\bm{k}}}\bm{e}_{\bm{k}}\right]\cdot\frac{\hatp}{\meanrho} \nonumber\\
			&+ \left[(\meanp/\meanrho)\cdot\nabla\meanchi + |\meanp/\meanrho|^2+c_s^2 \right]\frac{\widehat{\lambda}}{\meanrho}.\label{fast_slow_last}
\end{align}
These expressions show that $f_0(x,y)$ is of the form $f_0(x,y) =A(x)[y] +C(x) $, where $A(x):Y\rightarrow Y$ is a linear map and $C(x)\in Y$ is independent of $y$. In particular, the derivative of $f_0$ with respect to $y$ is given by $D_y f_0(x,y) = A(x)$. Thus, for any $(x,y)\in X\times Y$, $D_y f_0(x,y)$ is invertible if and only if $A(x)$ is invertible.

We will now complete the proof by showing that $A(x)$ is invertible for all $x\in X$ that satisfy $\nabla S(\bm{x})\neq 0$ for all $\bm{x}\in Q$. Fix $\delta y = (\delta\hatalpha,\delta\hatp,\delta\hatchi)\in Y$. If there is a $y = (\hatalpha,\hatp,\hatchi)$ that solves the equation $A(x)[y] = \delta y$, then, by Eqs.\,\eqref{fast_slow_first}-\eqref{fast_slow_last}, $y$ must satisfy
\begin{align}
\delta\hatalpha = &c_s |\nabla S| \partial_\theta \hatalpha - \left[\mathbb{I}+ \frac{1}{1-\bm{e}_{\bm{k}}\cdot{(\meanp/\meanrho)}/{c_s}}\frac{(\meanp/\meanrho)}{c_s}\otimes \bm{e}_{\bm{k}}\right]\cdot \frac{\hatp}{\meanrho}\label{invertibility_aeqn}\\
\delta\hatp = &  \left[c_s |\nabla S|\left(\mathbb{I}+\bm{e}_{\bm{k}}\otimes\bm{e}_{\bm{k}}\right)-\frac{(\mathbb{I}-\bm{e}_{\bm{k}}\otimes\bm{e}_{\bm{k}})\cdot (\meanp/\meanrho)\otimes \nabla S}{1-\bm{e}_{\bm{k}}\cdot{(\meanp/\meanrho)}/{c_s}}\right]\cdot \partial_\theta\widehat{\bm{p}}\label{invertibility_peqn}\\
\delta\hatchi = &  c_s |\nabla S|\partial_\theta\hatchi - \left[\frac{\meanp}{\meanrho}+\nabla\meanchi + \frac{\left[(\meanp/\meanrho)\cdot\nabla\meanchi + |\meanp/\meanrho|^2+c_s^2 \right]}{c_s - (\meanp/\meanrho)\cdot\bm{e}_{\bm{k}}}\bm{e}_{\bm{k}}\right]\cdot\frac{\hatp}{\meanrho}.\label{invertibility_chieqn}
\end{align}
By decomposing Eq.\,\eqref{invertibility_peqn} into components parallel and perpendicular to $\nabla S$, it is straightforward to show that $\partial_\theta\hatp$ must be given by
\begin{align}
\partial_\theta\hatp =& \frac{1}{2 c_s |\nabla S|}\left[2\mathbb{I}-\bm{e}_{\bm{k}}\otimes \bm{e}_{\bm{k}} +[\mathbb{I}-\bm{e}_{\bm{k}}\otimes\bm{e}_{\bm{k}}]\cdot \frac{(\meanp/\meanrho)\otimes \bm{e}_{\bm{k}}}{c_s - \bm{e}_{\bm{k}}\cdot (\meanp/\meanrho)}\right]\cdot\delta\hatp\nonumber\\
\equiv& \frac{1}{c_s|\nabla S|}{\mathbb{T}}\cdot \delta\hatp,
\end{align}
which implies $\hatp =(c_s |\nabla S|)^{-1} \mathbb{T}\cdot I[\delta\hatp]$, where $I[\delta\hatp]$ is the unique $\theta$-antiderivative of $\delta\hatp$ with zero mean, i.e.
\begin{align}
I[\delta\widehat{\bm{p}}](\theta) = \int_0^\theta \widehat{\bm{p}}(\overline{\theta}) \,d\overline{\theta} - \fint \left(\int_0^\theta  \widehat{\bm{p}}(\overline{\theta}) \,d\overline{\theta} \right)\,d\theta.\label{anti_op}
\end{align}
By substituting this expression for $\hatp$ into Eqs.\,\eqref{invertibility_aeqn} and \eqref{invertibility_chieqn}, it follows that $\hatalpha$ and $\hatchi$ must be given by
\begin{align}
	\hatalpha = & \frac{1}{c_s |\nabla S|} I[\delta\hatalpha] + \frac{1}{(c_s|\nabla S|)^2}\left[\mathbb{I}+\frac{(\meanp/\meanrho)\otimes \bm{e}_{\bm{k}}}{c_s-\bm{e}_{\bm{k}}\cdot (\meanp/\meanrho)}\right]\cdot\mathbb{T}\cdot I^2[\delta\hatp/\meanrho]\\
	\hatchi = & \frac{1}{c_s |\nabla S|}I[\delta\hatchi] + \frac{1}{(c_s|\nabla S|^2)}\left[\frac{\meanp}{\meanrho}+\nabla\meanchi + \frac{\left[(\meanp/\meanrho)\cdot\nabla\meanchi + |\meanp/\meanrho|^2+c_s^2 \right]}{c_s - (\meanp/\meanrho)\cdot\bm{e}_{\bm{k}}}\bm{e}_{\bm{k}}\right]\cdot \mathbb{T}\cdot I^2 \left[\frac{\delta\hatp}{\meanrho}\right],
\end{align}
where $I^2[\delta\hatp/\meanrho]= I[ I[ \delta\hatp/\meanrho] ]$ denotes the antiderivative operator applied two times.
Thus, if there is a $y$ satisfying $A(x)[y] = \delta y$, then that $y$ is unique. Conversely, by substituting the above expressions for $y$ back into $A(x)[y] = \delta y$, we conclude that a solution $y$ exists for any $\delta y$. Therefore $A(x)$ is invertible as claimed.
\end{proof}

\begin{remark}
In the above Proposition, the dispersion relation \eqref{eq:example:phase_final} plays two important roles. First, it supplies an evolution equation for $S$. This is necessary for the Proposition to work because if the dispersion relation was not imposed, then \Eqs{eq:extELE_v}--\eq{eq:extELE_chi} would not specify a dynamical system, let alone a fast-slow dynamical system. Second, it ensures that the system supports wave motion whose asymptotic behavior is captured by the NL-WKB ansatz.
\end{remark}

\begin{remark}
As is true of all fast-slow systems, \Eqs{eq:extELE_v}--\eq{eq:extELE_chi}, together with the dispersion relation \eqref{eq:example:phase_final}, admit a slow manifold. In terms of the fast and slow variables identified in the proof of the Proposition, the slow manifold is a subset of $X\times Y$ of the form $I_\epsilon = \{(x,y)\in X\times Y\mid y = y^\star_\epsilon(x)\}$, where $y^\star_\epsilon$ is the so-called slaving function.
Using the expressions for the inverse of $A(x)$ from the proof, it is straightforward to find the leading-order term in slaving function $y^\star_0(x) = (\hatalpha_0^\star,\hatp_0^\star,\hatchi_0^\star)$. We have
\begin{align}
\hatalpha_0^{\star} = & \, \frac{1}{2|\nabla S|} \left[1 - \frac{\bm{e}_{\bm{k}}\cdot (\meanp/\meanrho)}{c_s} \right] I[\widehat{\lambda}/\meanrho]\bm{e}_{\bm{k}}\\
\hatp_0^\star =& \, \frac{1}{2} \widehat{\lambda} c_s \left[1 - \left(\frac{\bm{e}_{\bm{k}}\cdot (\meanp/\meanrho)}{c_s}\right)^2\right] \bm{e}_{\bm{k}} + \frac{1}{2}\widehat{\lambda} c_s \left[1- \left(\frac{\bm{e}_{\bm{k}}\cdot (\meanp/\meanrho)}{c_s}\right)\right][\mathbb{I}-\bm{e}_{\bm{k}}\otimes \bm{e}_{\bm{k}}]\cdot\frac{\meanp}{c_s\meanrho}\\
\hatchi_0^\star = & \, \frac{1}{2|\nabla S|}\left[\bm{e}_{\bm{k}}\cdot (\meanp/\meanrho+\nabla\meanchi) - c_s\right]\left[1-\frac{\bm{e}_{\bm{k}}\cdot (\meanp/\meanrho)}{c_s}\right] I[\widehat{\lambda}/\meanrho].
\end{align}

While it was convenient to introduce the dependent variable $\widehat{\lambda}$ for the sake of showing equivalence with a fast-slow system, now that the existence of the slow manifold has been established, we are free to express the slow manifold in terms of $\hatrho$ instead of $\widehat{\lambda}$. By a slight abuse of notation, the slow manifold may written in terms of $\hatrho$ as 
\begin{align}
I_\epsilon = \{(\meanh,\meanp,\meanchi,\hatrho,S,\hatalpha,\hatp,\hatchi)\mid \hatalpha = \hatalpha^\star_\epsilon,\,\hatp = \hatp^\star_\epsilon,\,\hatchi = \hatchi_\epsilon^\star\},
\end{align}
where now $y_\epsilon^* = (\hatalpha^\star_\epsilon,\hatp^\star_\epsilon,\hatchi_\epsilon^\star)$ is a function of $(\meanh,\meanp,\meanchi,\hatrho,S)$.
In this alternate representation, the leading-order terms in the slaving functions are given by:
\begin{align}
		\hatalpha_0^\star	=& \frac{\vec{e}_{\vec{k}}}{|\del S|} I[\hatrho/\meanrho]
			\label{eq:alpha_sol} \\
		\hatp_0^\star =& \meanp \, \frac{\hatrho}{\meanrho} +c_s \hatrho \,\vec{e}_{\vec{k}} 	
			\label{eq:hatp_sol}\\
		\hatchi_0^\star =& \frac{1}{|\del S|}
						\left( \vec{e}_{\vec{k}} \cdot \del \meanchi  + \frac{\meanp}{\meanrho} \cdot \vec{e}_{\vec{k}} - c_s \right) 
						I[\hatrho/\meanrho]
			\label{eq:hatchi_sol}.
\end{align}
We remind the reader that the antiderivative operator $I$ was defined in Eq.\,\eqref{anti_op}.

\end{remark}

\begin{remark}
We are now in a good position to prove Proposition \ref{prop:wact}. 

\begin{proof}[proof of Proposition \ref{prop:wact}]
In Corollary \ref{cor:isothermal_variational}, we demonstrated that the NL--WKB extension of the isothermal fluid equations \eq{eq:example_ext_cont} and \eq{eq:example_ext_mom} is variational. Additionally, Theorem~\ref{theorem_3} shows that all extLBEP fluid equations imply a wave-action conservation equation \eq{prf_3_last}. Upon substituting the action \eq{eq:extagr_isothermal} and the Hamiltonian \eq{eq:example_ham_isothermal} into \Eq{prf_3_last}, we obtain
\begin{equation}
	\pd_t	\fint	\widetilde{\mathcal{I}}	\,d\theta 
		+\del	\cdot \fint \frac{\fluctp}{\fluctrho}\,  \widetilde{\mathcal{I}}\,d\theta = 0.
	\label{eq:wact_aux}
\end{equation}
The specific wave action density $\widetilde{\mathcal{I}}$ is given by \Eq{eq:specific_wave_density}, which we rewrite below for clarity:
\begin{equation}
	\widetilde{\mc{I}} =\trho \,\pd_\theta\tchi  + (\tp+\trho \,\del^S\tchi)\cdot\widetilde{\bm{\zeta}}  .
	\label{eq:swd_aux}
\end{equation}
Here $\fluctzeta = -(  \pd_\theta \flucth ) \cdot ( \del \flucth + \ep^{-1} \del S \otimes \pd_\theta \flucth )^{-1}$. Let us now calculate the terms in \Eqs{eq:wact_aux} and \eq{eq:swd_aux} by substituting the leading-order slaving functions \eq{eq:alpha_sol}--\eq{eq:hatchi_sol}. Specifically, when inserting $\flucth = \meanh \circ \hattau$ and $\hattau(\bm{x}) \simeq \bm{x}  + \ep^2 \hatalpha_0^\star(\bm{x})$ into $\fluctzeta$, we obtain
\begin{equation}
	\fluctzeta
		=	- \ep^2 \pd_\theta \hatalpha_0^\star 
			+ \ep^3 ( \pd_\theta \hatalpha_0^\star \cdot \del S ) \pd_\theta \hatalpha_0^\star 
			+ \mc{O}(\ep^4).
\end{equation}
We then substitute this result as well as the parameterizations \eq{eq:example:para}, \eq{eq:example:para_II}, and \eq{eq:example:para_III} and the leading-order slaving functions \eq{eq:alpha_sol}--\eq{eq:hatchi_sol} into the first term in \Eq{eq:wact_aux}. We obtain
\begin{equation}
	\fint	\widetilde{\mathcal{I}}	\,d\theta  
		= - \ep^3 \fint \meanrho \, \frac{c_s}{|\del S|} \left(\frac{\hatrho}{\meanrho} \right)^2 + O(\epsilon^4)
		= - \ep^3 \mc{I} + \mc{O}(\ep^4),
	\label{eq:wact_aux_II}
\end{equation}
where $\mc{I}$ is the wave action density defined in \Eq{eq:example_wave_action}. For the second term in \Eq{eq:wact_aux}, a similar calculation leads to
\begin{align}
	\fint \frac{\fluctp}{\fluctrho}\,  \widetilde{\mathcal{I}}\,d\theta
		&	= \fint \left( \frac{\meanp}{\meanrho} 
							+ \ep \frac{\hatp^\star_0}{\meanrho} 
							- \ep \frac{\meanp}{\meanrho} \frac{\hatrho}{\meanrho} 
							+ \mc{O}(\ep^2) \right)   \widetilde{\mathcal{I}}\,d\theta 
				\notag \\
		&	=	\frac{\meanp}{\meanrho}  \fint \widetilde{\mathcal{I}}\,d\theta
							+ \ep \fint  \left( \frac{\hatp^\star_0}{\meanrho} - \frac{\meanp}{\meanrho} \frac{\hatrho}{\meanrho}  + \mc{O}(\ep) \right) 
							 \widetilde{\mathcal{I}}\,d\theta 
				\notag	\\
		&	=	- \ep^3 \left( \frac{\meanp}{\meanrho}  + c_s \vec{e}_{\vec{k}} \right) \mc{I} + \mc{O}(\ep^4).
	\label{eq:wact_aux_III}
\end{align}
Finally, inserting \Eqs{eq:wact_aux_II} and \eq{eq:wact_aux_III} into \Eq{eq:wact_aux} leads to our claim in Proposition \ref{prop:wact}.
\end{proof}
\end{remark}

\subsection{Calculation of the effective action $\overline{\msf{A}}_{\rm T}$}
\label{sec:six_how}

In the previous section \Sec{sec:six_why}, we gave the general arguments explaining why the wave--mean-flow equations \eq{eq:example:cont_final}--\eq{eq:example:phase_final} are variational. We also broadly discussed how to calculate the effective action \eq{eq:example_act_reduced} for wave--mean-flow interactions by using slow-manifold reduction (see Theorem \ref{theo:slow_manifold_reduction}). In this section, we present some of the technical details in obtaining \Eq{eq:example_act_reduced}.

We start from the NL--WKB extended action \eq{eq:extagr_isothermal} for the isothermal fluid equations. As was done in the proof of Theorem \ref{theorem_4}, we apply a phase shift to the fields. This gives
\begin{align}
	\widetilde{\msf{L}}_{\rm T}
		&	= 		\int_Q  \fint
						\left( \fluctp^{S/\ep} \cdot  	\fluctv^{S/\ep} 
						+ \fluctrho^{S/\ep} \big( \pd_t \fluctchi^{S/\ep} 
							+ \fluctv^{S/\ep}  \cdot   \del    \fluctchi^{S/\ep}  \big) \right)	\, d \theta \, d \vec{x}  
				-	\int_Q  \fint 	\mc{H}_{\rm T} (\fluctp^{S/\ep} , \fluctrho^{S/\ep} )  \,  d\theta \,  d \vec{x} ,
	\label{eq:lagr_non_averaged}
\end{align}
where the superscript ``$S/\ep$" denotes that $\theta$ angle is shifted by $S/\ep$ (see Definition \ref{def:phase_shift}).

As it was explained in Theorem \ref{theo:slow_manifold_reduction}, after a slow manifold $I_\ep$ has been identified, one can restrict the action of the parent fast-slow system onto the slow manifold $I_\ep$ in order to obtain an effective action for the slow variables only. Following this same procedure, we substitute the parameterizations \eq{eq:example:para}, \eq{eq:example:para_II}, and \eq{eq:example:para_III} into \Eq{eq:lagr_non_averaged}. We then restrict the fast variables $(\hatalpha, \hatp, \hatchi)$ to the slow manifold by using \Eqs{eq:alpha_sol}--\eq{eq:hatchi_sol}. This leads to
\begin{align}
	\widetilde{\msf{L}}_{\rm T}
		&	= 		\int_Q  \fint
						\left( \fluctp^{\star S/\ep}_0 \cdot  	\fluctv^{ \star S/\ep}_0 
						+ \fluctrho^{S/\ep} \big( \pd_t \fluctchi^{\star S/\ep}_0 
							+ \fluctv^{\star S/\ep}_0  \cdot   \del    \fluctchi^{\star S/\ep}_0  \big) \right)	\, d \theta \, d \vec{x}  
				-	\int_Q  \fint 	\mc{H}_{\rm T} (\fluctp^{\star S/\ep}_0 , \fluctrho^{S/\ep} )  \,  d\theta \,  d \vec{x} ,
	\label{eq:lagr_non_averaged_II}
\end{align}
where $\fluctp^{\star S/\ep}_0 = \meanp + \ep \hatp^{\star S/\ep}_0$ and similarly for the rest of the variables restricted to the slow manifold. {From hereon, we shall consider the Lagrangian \eq{eq:lagr_non_averaged_II} restricted to the lowest-order slaving functions. Hence, to simplify our notation, we shall omit the ``0" subscript when referring to the lowest-order slaving functions \eq{eq:alpha_sol}--\eq{eq:hatchi_sol}.}

Before explicitly substituting the expressions for $\hatalpha^\star$, $\hatp^\star$, and $\hatchi^\star$ into \Eq{eq:lagr_non_averaged_II}, it convenient to first perform a variable transformation. (The following transformation is closely related to the well-known oscillation-center transform used in kinetic theories for plasma--wave interactions.\cite{Dewar:73aa}) First, we note that the velocity
\begin{equation}
		\fluctv^{\star S/\ep} 
			= - \pd_t \widetilde{\vec{h}}^{\star S/\ep} \cdot ( \del \flucth^{\star S/\ep} )^{-1}
			= - \pd_t (\meanh \circ \flucttau^{\star S/\ep} ) \cdot [ \del  (\meanh \circ \flucttau^{\star S/\ep} )]^{-1}		
\end{equation} 
can be written as
\begin{equation}
	\fluctv^{\star  S/\ep} = (\widetilde{\tau}^{\star S / \ep})^* ( \meanv - \fluctnu^{\star S/\ep} ),
	\label{eq:example_vel_osc}
\end{equation}
where $(\widetilde{\tau}^{\star S / \ep})^*$ is the pullback associated to $\flucttau^{\star S / \ep}$, $\meanv$ is the mean Lagrangian velocity
\begin{equation}
	\meanv \doteq  -\pd_t \meanh \cdot (\del  \meanh)^{-1},
\end{equation}		
and $\fluctnu^{\star S/\ep}$ is the velocity associated to $\flucttau^{\star S / \ep}$:
\begin{equation}
	\fluctnu^{\star S/\ep}	\doteq	\pd_t ( \flucttau^{\star S / \ep} ) \circ (\flucttau^{\star S / \ep})^{-1}.
	\label{eq:nu}
\end{equation}
In order to simplify the expression for the velocity \eq{eq:example_vel_osc},  it is convenient to apply the pushforward $\widetilde{\tau}^{\star S / \ep}_*$ to the symplectic part of the Lagrangian \eq{eq:lagr_non_averaged_II}. This leads to
\begin{align}
	\widetilde{\msf{L}}_{\rm T}
		&	= 		\int_Q  \fint
						\left(  \widetilde{\tau}^{\star S / \ep}_* \fluctp^{\star S/\ep}  \cdot 
								\left( 	\meanv - \widetilde{\vec{\nu}}^{\star S/\ep} 	\right)
							+ 	\widetilde{\tau}^{\star S / \ep}_* \fluctrho^{S/\ep} 
								\big( \pd_t \widetilde{\varphi}^{\star S/\ep} 
											+ \meanv  \cdot   \del    \widetilde{\varphi}^{\star S/\ep}  \big)  \right)\,
								d \theta \, d^3 \vec{x}  
					\notag \\
		&			\qquad
					-	\int_Q \fint  	\mc{H} (\fluctp^{\star S/\ep},  \fluctrho^{S/\ep}  )  \, d \theta \,  \mathrm{d}^3 \vec{x} ,
	\label{eq:example:act_transformed}
\end{align}
Some remarks should be given on the symbols appearing in \Eq{eq:example:act_transformed}. First, the term $\widetilde{\tau}^{\star S / \ep}_* \fluctp^{S/\ep}$ is understood as the $\widetilde{\tau}^{\star S / \ep}_*$ acting on $\fluctp^{\star S/\ep} $ which is treated as a one-form density in the domain $Q$. Similarly, $\widetilde{\tau}^{\star S / \ep}_* \fluctrho^{\star S/\ep}$ is interpreted as the pull-back $\widetilde{\tau}^{\star S / \ep}_*$ acting on $\fluctrho^{S/\ep} $ treated as a three-form in the domain $Q$. We also introduced a new variable
\begin{equation}
	 \widetilde{\varphi}^{\star S / \ep} 
		= \widetilde{\tau}^{\star S / \ep}_* \fluctchi^{\star S / \ep} ,
	\label{eq:example_varphi}
\end{equation}
which is a scalar in the domain $Q$. { Finally, we used the Lie derivative theorem to establish the identities $\pd_t \fluctchi^{\star S / \ep} = \pd_t [ (\widetilde{\tau}^{\star S / \ep})^* \widetilde{\varphi}^{\star S / \ep}  ]  = (\widetilde{\tau}^{\star S / \ep})^* \pd_t   \widetilde{\varphi}^{\star S / \ep}  + (\widetilde{\tau}^{\star S / \ep})^* \mathfrak{L}_{\fluctnu^{\star S/\ep}} \widetilde{\varphi}^{\star S / \ep} $ and $\mathfrak{L}_{\fluctv^{\star  S/\ep}}  \fluctchi^{\star S / \ep} =  (\widetilde{\tau}^{\star S / \ep})^* \mathfrak{L}_{\meanv - \fluctnu^{\star S/\ep}  } \widetilde{\varphi}^{\star S / \ep}  $.} Note that, by applying the pushforward $\widetilde{\tau}^{\star S / \ep}_*$, we were able to replace the oscillating velocity $\fluctv^{\star S/\ep}$ appearing in the $\fluctv^{\star S/\ep}  \cdot   \del    \fluctchi^{\star S/\ep}$ term of \Eq{eq:lagr_non_averaged} with the mean Lagrangian velocity $\meanv$. This was originally the main motivation for the transformation.

The next step is to explicitly calculate the terms $\widetilde{\vec{\nu}}^{\star  S/\ep}$, $\widetilde{\tau}^{\star S / \ep}_* \fluctp^{\star S/\ep}$, $\widetilde{\tau}^{\star S / \ep}_* \fluctrho^{S/\ep}$, and $\widetilde{\varphi}^{\star S / \ep}$  appearing in \Eq{eq:example:act_transformed}. Let us first start by calculating $\widetilde{\vec{\nu}}^{\star S/\ep}$ in \Eq{eq:nu}. Since $(\flucttau^{\star S / \ep})^{-1}$ is a near-identity transformation, one can verify that $(\flucttau^{\star S / \ep})^{-1}(\vec{x}) = \vec{x} - \ep^2 \hatalpha^{\star S/\ep}(\vec{x}) + \mc{O}(\ep^3)$. Substituting this into \Eq{eq:nu} gives
\begin{align}
	\widetilde{\vec{\nu}}^{\star S/\ep} 
		&	=	[ \ep (\pd_t S) \pd_\theta \hatalpha^{\star S/\ep} + \ep^2 (\pd_t \hatalpha)^{\star S/\ep} ] \circ
				( \vec{id} - \ep^2 \hatalpha^{\star S/\ep}  ) + \mc{O}(\ep^3)  
			\notag \\
		&	=	( 1 - \ep^2 \hatalpha^{\star S/\ep} \cdot \del  ) \left(  \ep (\pd_t S) \pd_\theta \hatalpha^{\star S/\ep}	
				+ \ep^2 (\pd_t \hatalpha)^{\star S/\ep} \right) + \mc{O}(\ep^3)  
			\notag \\
		&	=	\ep (\pd_t S) \pd_\theta \hatalpha^{\star S/\ep}	
				+ \ep^2 (\pd_t \hatalpha)^{\star S/\ep} 
				- \ep^2 ( \del S \cdot \hatalpha^{\star S/\ep}  ) (\pd_t S) \pd_\theta^2 \hatalpha^{\star S/\ep}	 
				+ \mc{O}(\ep^3)  
			\notag \\
		&	=	\ep (\pd_t S) \pd_\theta \hatalpha^{\star S/\ep}	
				+ \ep^2 (\del S \cdot \pd_\theta \hatalpha^{\star S/\ep} )  (\pd_t S)  \pd_\theta \hatalpha^{\star S/\ep} 	\notag \\
		&		\quad
				+ \ep^2 (\pd_t \hatalpha^\star )^{S/\ep}  
				- \ep^2 \pd_\theta \left( ( \del S \cdot \hatalpha^{\star S/\ep}  ) (\pd_t S) \pd_\theta \hatalpha^{\star S/\ep} \right) 
				+ \mc{O}(\ep^3)  
			\notag \\
		&	=	\ep (\pd_t S) \pd_\theta \hatalpha^{\star S/\ep}	
				+ \ep^2 (\del S \cdot \pd_\theta \hatalpha^{\star S/\ep} )  (\pd_t S)  \pd_\theta \hatalpha^{\star S/\ep} 
				+ \ep^2 \mathrm{Osc}.
	\label{eq:example_tau_exp}
\end{align}
Here the term ``$\ep^2 \mathrm{Osc}$" means that we have neglected $\mc{O}(\ep^3)$ terms and that we have omitted writing fluctuating terms that are $\mc{O}(\ep^2)$ whose $\theta$-average is zero. Since we are only calculating the effective Lagrangian up to $\mc{O}(\ep^2)$, it is safe to omit those terms since they will not contribute anything once the Lagrangian \eq{eq:example:act_transformed} is explicitly $\theta$-averaged. More specifically, the term $\ep (\pd_t S) \pd_\theta \hatalpha^{\star S/\ep}$ is kept because it could later multiply another $\mc{O}(\ep)$ term in the Lagrangian \eq{eq:example:act_transformed}. The term $\ep^2 (\del S \cdot \pd_\theta \hatalpha^{\star S/\ep} )  (\pd_t S)  \pd_\theta \hatalpha^{\star S/\ep} $ is also kept because it is quadratic in $\hatrho$, so it has a non-zero $\theta$-average. The term $\ep^2 (\pd_t \hatalpha^\star)^{S/\ep}$ is omitted because it is oscillatory so its $\theta$-average is zero.  Also, the term $\ep ^2 \pd_\theta ( ( \del S \cdot \hatalpha^{\star S/\ep}  ) (\pd_t S) \pd_\theta \hatalpha^{\star S/\ep} )$ is omitted because it is as a total derivative in $\theta$, which will vanish when integrating over $\theta$. Finally, substituting the expression for $\hatalpha^{\star S/\ep} $ in \Eq{eq:alpha_sol} gives
\begin{equation}
	\widetilde{\vec{\nu}}^{\star S/\ep}
		=	\ep \vec{e}_{\vec{k}} \frac{\pd_t S}{|\del S|} \frac{\hatrho^{S/\ep}}{\meanrho} 
			+ \ep^2 \vec{e}_{\vec{k}} \frac{\pd_t S}{|\del S|} \left( \frac{\hatrho^{S/\ep}}{\meanrho} \right)^2,
\end{equation}			
where $\vec{e}_{\vec{k}} \doteq \del S / |\del S| $. 

Let us now proceed by calculating the term $\widetilde{\tau}^{\star S / \ep}_* \fluctrho^{S/\ep}$ appearing in \Eq{eq:example:act_transformed}. Remembering that the density should be considered as a 3-form, we obtain
\begin{align}
	\widetilde{\tau}^{\star S / \ep}_*	 ( \fluctrho^{S/\ep}  d^3 \vec{x} )
		&	=	 \fluctrho^{S/\ep} \circ ( \flucttau^{\star S / \ep})^{-1} 
				 \, \det ( \del   (\flucttau^{\star S / \ep})^{-1} ) \, d^3 \vec{x}  \notag \\
		&	=	 \fluctrho^{S/\ep} \circ ( \vec{id} - \ep^2 \hatalpha^{\star S/\ep}  )
				 \, \det ( \mathbb{I} - \ep^2 \del   \hatalpha^{\star S/\ep}    ) \, d^3 \vec{x}
				 +\mc{O}(\ep^3)\notag \\
		&	=	(\fluctrho^{S/\ep} - \ep^2 \hatalpha^{\star S/\ep} \cdot \del \fluctrho^{S/\ep} )
				( 1 - \ep^2 \del  \cdot \hatalpha^{\star S/\ep}  )  \, d^3 \vec{x}
				+\mc{O}(\ep^3) \notag \\
		&	=	\fluctrho^{S/\ep} d^3 \vec{x} - \ep^2 \del   \cdot  ( \hatalpha^{\star S/\ep}\fluctrho^{S/\ep} ) \, d^3 \vec{x}	
				+\mc{O}(\ep^3) \notag \\					
		&	=	\meanrho \, d^3 \vec{x} + \ep \hatrho^{S/\ep} d^3 \vec{x} 
				- \ep^2 \del \cdot ( \hatalpha^{\star S/\ep} \meanrho) \, d^3 \vec{x} 
				- \ep^3 \del \cdot ( \hatalpha^{\star S/\ep} \hatrho^{S/\ep}) \, d^3 \vec{x}   
				+\mc{O}(\ep^3) \notag \\
		&	=	 \meanrho \, d^3 \vec{x}  + \ep \hatrho^{S/\ep} d^3 \vec{x}  
						 - \ep \del S \cdot ( \pd_\theta \hatalpha^{\star S/\ep} \meanrho)  \, d^3 \vec{x} 
						- \ep^2 \del S \cdot \pd_\theta ( \hatalpha^{\star S/\ep} \hatrho^{S/\ep})  \, d^3 \vec{x} 
						+\mc{O}(\ep^3)
				\notag \\
		&	=	\meanrho  \, d^3 \vec{x}  + \ep^2 \mathrm{Osc},
	\label{eq:example_push_rho}
\end{align}
where in the last line, we substituted \Eq{eq:alpha_sol} so that $ \del S \cdot ( \pd_\theta \hatalpha^{\star S/\ep} \meanrho)   = \hatrho^{S/\ep}$.  We also used the well-known formula for the determinant of a near-identity matrix:
\begin{align}
	\det ( \mathbb{I} & - \ep^2 \del   \hatalpha^{\star S/\ep}    ) \notag \\
		&	=	1 - \ep^2 \mathrm{Tr}( \del   \hatalpha^{\star S/\ep}  )
				 	- \frac{\ep^4}{2} \left(  \mathrm{Tr}^2 ( \del   \hatalpha^{\star S/\ep} ) 
				 			- \mathrm{Tr}(\del   \hatalpha^{\star S/\ep} )^2 \right)
				 	+\mc{O}(\ep^6)  
				\notag \\
		&	=	1 - \ep^2 \del \cdot \hatalpha^{\star S/\ep} 
				 	- \frac{\ep^4}{2} \left( ( \del  \cdot  \hatalpha^{\star S/\ep} )^2 
				 			- ( \del    \hatalpha^{\star S/\ep}: \del   \ \hatalpha^{\star S/\ep} ) \right)
				 	+\mc{O}(\ep^6)  
				 \notag \\
		&	=	1 - \ep^2 \del  \cdot \hatalpha^{\star S/\ep} 
				 	- \frac{\ep^2}{2} \left(( \del S \cdot \pd_\theta \hatalpha^{\star S/\ep} )^2 
				 			- ( \del S \otimes \pd_\theta  \hatalpha^{\star S/\ep}: \del S \otimes \pd_\theta  \hatalpha^{\star S/\ep} ) \right) 
				 	+\mc{O}(\ep^4)  
				 \notag \\
		&	=	1 - \ep^2 \del  \cdot \hatalpha^{\star S/\ep} +\mc{O}(\ep^4)  .
\label{eq:example_det}
\end{align}
where at the end, the terms in parentheses cancel because $\hatalpha^\star $ is parallel to $\del S$.

In a similar manner, we can calculate the term $\widetilde{\tau}^{\star S / \ep}_* \fluctp^{\star S/\ep}$. Note, however, that we should consider $\fluctp^{\star S/\ep}$ as a one-form density so that the above is written as $\widetilde{\tau}^{S\star  / \ep}_*( \fluctp^{\star S/\ep} \cdot d \vec{x} \otimes d^3 \vec{x} )$. A direct calculation leads to
\begin{align}
& \widetilde{ \tau}^{\star S / \ep}_*  ( \fluctp^{\star S/\ep} \cdot d \vec{x} \otimes d^3 \vec{x} ) \notag \\
	&	=		d \vec{x} \cdot \del  (( \widetilde{\tau}^{\star S / \ep})^{-1} )  
				\cdot (\fluctp^{\star S/\ep} \circ  (\flucttau^{\star S / \ep})^{-1} )
				\otimes \det ( \vec{D} (\flucttau^{\star S / \ep})^{-1} ) \, d^3 \vec{x}
				\notag \\
	&	=		d \vec{x} \cdot ( \mathbb{I} - \ep^2 \del  \hatalpha^{\star S / \ep} ) 
				\cdot (\fluctp^{\star S/\ep} \circ  (\vec{id} - \ep^2 \hatalpha^{\star S / \ep} ) )
				\otimes \det ( \mathbb{I}  - \ep^2 \del   \hatalpha^{\star S/\ep}    )  \, d^3 \vec{x}
				+\mc{O}(\ep^3)
				\notag \\
	&	=		d \vec{x} \cdot ( \mathbb{I} - \ep^2 \del  \hatalpha^{\star S / \ep} ) 
				\cdot (\fluctp^{\star S/\ep} - \ep^2 (\hatalpha^{\star S / \ep} \cdot \del) \fluctp^{\star S/\ep}  )
				\otimes ( 1  - \ep^2 \del \cdot  \hatalpha^{\star S/\ep}    )  \, d^3 \vec{x}
				+\mc{O}(\ep^3)
				\notag \\
	&	=		d \vec{x} \cdot ( \mathbb{I} - \ep \del S \otimes \pd_\theta \hatalpha^{\star S / \ep} ) 
				\cdot (\fluctp^{\star S/\ep} - \ep (\hatalpha^{\star S / \ep} \cdot \del S) \pd_\theta \fluctp^{\star S/\ep}  )
				\otimes ( 1  - \ep \del S \cdot  \pd_\theta \hatalpha^{\star S/\ep}    )  \, d^3 \vec{x}
				+\ep^2 \mathrm{Osc}
				\notag \\					
	&	=		\bigg( \meanp 
					+ \ep \hatp^{\star S/\ep}
					- \ep (\meanp \cdot \pd_\theta \hatalpha^{\star S/\ep} ) \del S 
					- \ep ( \del S  \cdot \pd_\theta \hatalpha^{\star S/\ep} )\meanp 
					+ \ep^2 \pd_\theta ( (\hatalpha^{\star S/\ep} \cdot \del S ) \hatp^{\star S/\ep} 			)
				 \notag \\
	&			\quad					
					- \ep^2 (\hatp^{\star S/\ep} \cdot \pd_\theta \hatalpha^{\star S/\ep} ) \del S 		
					+ \ep^2 (\meanp   \cdot \pd_\theta \hatalpha^{S/\ep})  ( \del S  \cdot \pd_\theta \hatalpha^{\star S/\ep} ) \del S \bigg)
				\cdot d \vec{x}  \otimes d^3 \vec{x}
				+ \ep^2 \mathrm{Osc}	
				\notag \\
	&	=		\bigg( \meanp 
					+ \ep c_s \hatrho^{S/\ep} \vec{e}_{\vec{k}} 
					- \ep (\meanp \cdot \vec{e}_{\vec{k}}  ) (  \hatrho^{S/\ep} / \meanrho ) \, \vec{e}_{\vec{k}} 
					-\ep^2 c_s \meanrho (  \hatrho^{S/\ep} / \meanrho )^2 \vec{e}_{\vec{k}}  \bigg)
				\cdot d \vec{x}  \otimes d^3 \vec{x}
				+  \ep^2 \mathrm{Osc}	,
\label{eq:example_push_p}
\end{align}
where in the last line, we substituted the expressions for $\hatalpha^\star$ and $\hatp^\star$ in \Eqs{eq:alpha_sol} and \eq{eq:hatp_sol}.

Finally, a far simpler calculation of $\widetilde{\varphi}^{\star S / \ep} $ introduced in \Eq{eq:example_varphi} gives
\begin{align}
 \widetilde{\varphi}^{\star S / \ep} 
	&	= (\widetilde{\tau}^{\star S / \ep} )_* \fluctchi^{\star S / \ep} \notag \\
	&	=		\fluctchi^{\star S / \ep} \circ  (\flucttau^{\star S / \ep})^{-1} 
				\notag \\
	&	=		\widetilde{\chi}^{\star S / \ep} \circ  (\vec{id} - \ep^2 \hatalpha^{\star S / \ep} ) 
				+\mc{O}(\ep^3)
				\notag \\
	&	=		\widetilde{\chi}^{\star S / \ep} 
				-\ep^2 \hatalpha^{\star S / \ep} \cdot \del \widetilde{\chi}^{\star S / \ep}
				+\mc{O}(\ep^3)
				\notag \\
	&	=		\meanchi + \ep^2 \hatchi^{\star S / \ep} 
				-\ep^2 \hatalpha^{\star S / \ep} \cdot \del  ( \meanchi + \ep^2 \hatchi^{\star S / \ep}  )
				+\mc{O}(\ep^3)
				\notag \\
	&	=		\meanchi 
				+  \ep^2 \mathrm{Osc}	.
\label{eq:example_push_varphi}
\end{align}

We now insert \Eqs{eq:example_tau_exp}, \eq{eq:example_push_rho}, \eq{eq:example_push_p}, and \eq{eq:example_push_varphi} into the Lagrangian \eq{eq:example:act_transformed}. Starting from the first integral in \Eq{eq:example:act_transformed}, we substitute the obtained expressions for $\widetilde{\tau}^{S / \ep}_* \fluctp^{S/\ep}$ and $\fluctnu^{S/\ep}$. We then Whitham average, or $\theta$ average, the Lagrangian and only keep terms up to $\mc{O}(\ep^2)$. We obtain
\begin{equation}
	\int_Q \fint 	 
						( \widetilde{\tau}^{\star S / \ep}_* \fluctp^{\star S/\ep} ) \cdot 
						(	\meanv - \fluctnu^{\star S/\ep} 	) \,
						d \theta \,  d^3 \vec{x}  
			= 		\int_Q 
						\left( \meanp \cdot \meanv - \ep^2 \mc{I} (\pd_t S + \meanv \cdot \del S ) 	\right)  	\, d^3 \vec{x} 
						+\mc{O}(\ep^3) 		,
	\label{eq:example_act_sym}
\end{equation}
where $\mc{I}$ is the wave action density introduced in \Eq{eq:example_wave_action}. For the next term of the Lagrangian \eq{eq:example:act_transformed}, we substitute \Eqs{eq:example_push_rho} and \eq{eq:example_push_varphi}. This leads to
\begin{equation}
	\int_Q \fint \widetilde{\tau}^{\star S / \ep}_*	 \fluctrho^{S/\ep} 
					\big( \pd_t \widetilde{\varphi}^{\star S/\ep}_\theta 
							+ \meanv  \cdot   \del   \widetilde{\varphi}^{\star S/\ep}_\theta   \big) \,
						d\theta \, d^3 \vec{x} 
		= 	\int_Q \meanrho
					\left( \pd_t \meanchi	+ \meanv  \cdot   \del  \meanchi  \right)  \, d^3 \vec{x}
			+\mc{O}(\ep^3) .
	\label{eq:example_act_lagr}
\end{equation}
In a similar manner, substituting \Eqs{eq:example:para} and \eq{eq:hatp_sol} gives the following for the $\theta$-averaged Hamiltonian:
\begin{align}
	\int_Q  \fint	\bigg( \frac{|\fluctp^{\star S / \ep}|^2}{2\fluctrho^{S / \ep}} + c_s^2 \fluctrho^{S / \ep} \ln \bigg( \frac{\fluctrho^{S / \ep}}{\rho_0} \bigg) 
						\bigg)  
						d \theta \,  d^3 \vec{x} 	
			= 		\int_Q 
						\left( \frac{|\meanp|^2}{2\meanrho} + c_s^2 \meanrho \ln \left( \frac{\meanrho}{\rho_0} \right)  + \ep^2 c_s \mc{I} |\del S| 	\right)
							 \, \mathrm{d}^3 \vec{x} 
						+\mc{O}(\ep^3) .
	\label{eq:example_act_ham}
\end{align}
When combining the results in \Eqs{eq:example_act_sym}--\eq{eq:example_act_ham}, we obtain the effective action given in \Eq{eq:example_act_reduced}, which we rewrite below for convenience:
\begin{align}
			\overline{\msf{L}}_{\rm T}
				( \meanh, \dot{\meanh}, \meanp, \dot{\meanp},  \meanrho, \dot{\meanrho},  \meanchi  , \dot{ \meanchi  } , 
								\mc{I},\dot{\mc{I}} , S, \dot{S} )
				&	= 		\int_Q 
							\left( 	\meanp \cdot \meanv 	
										+ \meanrho	\big( \pd_t \meanchi	+ \meanv  \cdot   \del    \meanchi  \big) \right) 	\, d^3 \vec{x} 
							- \int_Q 	\mc{H}_{\rm T}(\meanp, \meanrho) 	\, d^3 \vec{x}  \notag \\
				&	\quad 	- \ep^2 	 \int_Q \mc{I}		\left( \pd_t S + \meanv \cdot \del S + c_s |\del S| \right). 
\end{align}

In summary, in this section we have presented additional details for calculating the effective action for wave--mean-flow interactions in an isothermal fluid. Our method was primarily based on slow-manifold reduction, whose general ideas were presented in \Sec{sec:six_why}. Our derivation followed to two main steps. First, we restricted the variational principle to the slow manifold. This is essentially done by substituting the expressions obtained for the fast variables. Second, we identified a transformation that facilitated computations of the wave--mean-flow action.

\section{Discussion\label{sec:seven}}

In this article we have identified the variational structure underlying the nonlinear WKB method\cite{Whitham_1965_eom,Miura_1974} as it applies to ideal fluid equations in the Eulerian frame. This work therefore compliments previous studies on variational nonlinear WKB in the \emph{mean} Eulerian frame.\cite{Dewar_1970,Bretherton_1971,Gjaja_Holm_1996} Our main results concern what we have termed the nonlinear WKB extension procedure, which is the technique used for generating a system of equations governing the profile functions appearing in the nonlinear WKB ansatz. Our results may be summarized as follows.
(i) Given Eulerian fluid equations arising from an Euler-Poincar\'e variational principle,\cite{Holm_1998} we have shown that the enlarged system resulting from the nonlinear WKB extension procedure also arises from a variational principle.
(ii) This new variational principle inherits a ``looped" version of the original system's symmetry group. After recognizing that a subgroup of this looped group comprises a looped version of the particle relabeling group, we have used Noether's theorem to identify a family of circulation invariants parameterized by $S^1$.
(iii) By combining the newly discovered class of variational principles with ideas from the theory of slow manifold reduction, we have presented an example of a systematic procedure for identifying variational principles governing the self-consistent interaction between (possibly nonlinear) locally-plane waves and mean flows.

Our analysis made use of several technical assumptions that are straightforward to relax. In particular, we restricted our attention to barotropic fluid equations that arise from a local Lagrangian. A more general equation of state involving an advected entropy could readily be incorporated into our discussion by an enterprising reader. Similarly, it would not be prohibitively difficult to allow for spatial non-locality in the Lagrangian. (On the other hand, temporal nonlocality would not be simple to include.)  More generally, extensions of our work to fluid systems not discussed in this paper may readily be accommodated as long as the proof of Theorem \ref{theorem_4} remains in tact.

Two key technical features that distinguish our work from much of the previous work on variational fluid mechanics are (i) our use of the inverse of the Lagrangian configuration map $\bm{h} = \bm{g}^{-1}$, and (ii) our use of the fluid phase-space Lagrangian (akin to $L = p \dot{q} - H(q,p)$). The use of $\bm{h}(\bm{x})$ instead of $\bm{g}(\bm{x}_0)$ allowed us to reformulate the Euler-Poincar\'e approach to fluid variational principles in terms of conventional classical field theory, which in turn enabled us to apply Whitham's averaged Lagrangian technique in the Eulerian frame. Our inspiration for this shift in perspective came from Ref.\,\onlinecite{Beig_2003}, which explains the use of $\bm{h}$ within the theory of relativistic elastic solids. 
Using the phase-space Lagrangian allowed us to apply Theorem \ref{theo:slow_manifold_reduction} on the inheritance of Hamiltonian structure in order to explain the variational principle underlying the interaction between small-amplitude acoustic waves and a compressible barotropic mean flow.  This same idea was used in Refs.\,\onlinecite{Burby_two_fluid_2017} and \onlinecite{Burby_Sengupta_2017_pop} to explain the Hamiltonian structures underlying magnetohydrodynamics and kinetic magnetohydrodynamics, respectively.

It is most interesting to compare the approach we have introduced here for variational modeling of wave-mean-flow interaction with earlier approaches\cite{Dewar_1970,Bretherton_1971,Gjaja_Holm_1996,Holm:2002ju,Holm:2002kh,Gjaja_Holm_1996,Holm:2002ex} based on generalized Lagrangian mean (GLM) theory.\cite{Andrews:1978fg,Buhler_2009} As an intuitively appealing way of representing waves superimposed on a mean flow, previous authors have decomposed the Lagrangian configuration map $\bm{g}$ as the composition of a mean configuration map with a fluctuating configuration map. This decomposition forms the foundation of GLM theory. The mean configuration map takes values in (and in fact defines) the mean Eulerian frame, while the fluctuating configuration map takes values in the conventional Eulerian frame. When the averaging operation is identified with WKB phase averaging, the prevailing trend has then been to express the fluctuating configuration map in terms of the WKB ansatz. This effectively amounts to applying the WKB method within the mean Eulerian frame. While this approach obfuscates the connection between wave-mean-flow dynamics and the conventional Eulerian-frame WKB method, especially at higher orders in asymptotic expansions, it is compatible with variational formulations of fluid dynamics in a simple manner. Indeed, it is straightforward to decompose the Lagrangian configuration map in an Euler-Poincar\'e variational principle using the GLM ansatz, and then apply WKB phase averaging to the result.\cite{Dewar_1970,Bretherton_1971,Gjaja_Holm_1996} In contrast, the perspective taken in our new approach is that, in principle, there is no need to introduce the mean Eulerian frame in order to identify wave-mean-flow variational principles. Instead, one can start from our new variational principle for the nonlinear WKB extension of the Eulerian-frame fluid equations, and then apply slow manifold reduction to obtain the desired reduced variational principle for wave-mean-flow interaction. Aside from maintaining a clear link with the Eulerian-frame WKB procedure, a benefit of this approach is that it systematically incorporates the closure (a.k.a. ``slaving" or ``balance") relations needed to express the rapidly-varying fluctuations in terms of slowly-varying mean quantities, thereby eliminating the risk of unwanted fast modes creeping into the variational principle. (For an example of the latter phenomenon, see the Hamiltonian models in Refs.\,\onlinecite{Burby_gvm_2015} and \onlinecite{Brizard_Tronci_2016} for low-frequency dynamics of strongly-magnetized plasmas. Those models support high-frequency electromagnetic waves that must be handled with care.) Interestingly however, our calculations have revealed that it is practically expedient to express our variational principle in terms of mean-Eulerian frame quantities. In so doing, the Lagrangian simplifies dramatically. In fact, it is not at all clear that Lagrangian expressed in terms of conventional Eulerian frame quantities behaves well with respect to truncation, i.e. when high-order terms in the asymptotic expansions of the slaving functions are dropped. Thus, GLM theory plays an important practical role in our new formalism, even though it is not a necessary ingredient at a conceptual level.

Given the dichotomy between our new method and the established mean-Eulerian frame approach, it is also interesting to ask how the family of circulation invariants given in Corollary \ref{corollary_1} relates to the mean circulation invariants of Refs.\,\onlinecite{Bretherton_1971} and \onlinecite{Gjaja_Holm_1996}. Because the family of circulation invariants identified in Corollary \ref{corollary_1} is parameterized by the angle $\theta\in S^1$, it may be averaged over $\theta$ to obtain a mean circulation invariant for the nonlinear WKB extension of the ideal barotropic fluid equations. In particular, the averaged circulation is constant along solutions that lie in the slow manifold. Thus, the $S^1$-mean of our family of circulation invariants restricted to the slow manifold is a circulation invariant for the wave-mean-flow dynamics. It is in this manner that mean circulation invariants of the types found in Refs.\,\onlinecite{Bretherton_1971} and \onlinecite{Gjaja_Holm_1996} emerge from our formalism. As an illustration of this point, we prove in Appendix \ref{appendix_A} that the average of our family of circulation invariants restricted to the slow manifold is equivalent to the circulation invariant associated with mean particle relabeling symmetry of the wave--mean-flow Lagrangian \eqref{eq:extagr_isothermal}.

In the future, we plan to use the tools developed in this article to capture the effects of harmonic generation and corrections to ray trajectories caused by space-dependent wave polarization\cite{Littlejohn_1991_phases,Ruiz:2017ij} in wave-mean-flow problems arising in fluids and plasmas.

\begin{acknowledgements}
The authors would like to thank Richard Montgomery and Cesare Tronci for a number of helpful discussions of this work at MSRI. This material is based upon work supported by the National Science Foundation under Grant No. DMS-1440140 while one of the authors (JWB) was in residence at the Mathematical Sciences Research Institute in Berkeley, California, during the Fall 2018 semester. In addition, research presented in this article was supported by (1) the Laboratory Directed Research and Development program of Los Alamos National Laboratory under project number  20180756PRD4, and (2) Sandia National Laboratories. Sandia National Laboratories is a multimission laboratory managed and operated by National Technology and Engineering Solutions of Sandia, LLC., a wholly owned subsidiary of Honeywell International, Inc., for the U.S. DOE National Nuclear Security Administration under contract DE-NA-0003525. This paper describes objective technical results and analysis. Any subjective views or opinions that might be expressed in the paper do not necessarily represent the views of the U.S. DOE or the U.S. Government.
\end{acknowledgements}

\appendix
\section{Proof of mean circulation theorem\label{appendix_A}}

\begin{corollary}[Kelvin's theorem for wave--mean-flow system]
Given a closed curve $C_0 \subset Q_0$ and a solution $(\meanh, \meanp, \meanrho, \meanchi, \mc{I},S)$ of the wave--mean-flow equations \eq{eq:example:cont_final}--\eq{eq:example:phase_final}, then
\begin{equation}
	\frac{\mathrm{d}}{\mathrm{d}t} 
		\oint_{\overline{ C }} 
			\left( \frac{\meanp}{\meanrho} - \frac{\mc{I}}{\meanrho} \del S \right)  
			\cdot \mathrm{d} \vec{x} =0.
	\label{eq:example_Kelvin}
\end{equation}
where $\overline{C}=\meanh^{-1} ( C_0)$.
\end{corollary}

Thus, in the wave--mean-flow framework developed here, the circulation theorem \eq{eq:example_Kelvin} is now a closed contour integral of the fluid momentum minus a term related to the wave momentum. The modification of Kelvin's circulation theorem due to wave effects has been noticed before.\cite{Bretherton_1971} In essence, this result shows that waves can affect the vorticity of the bulk fluid. The last term in \Eq{eq:example_Kelvin} is sometimes referred as ``wave pseudomomentum."\cite{Gjaja_Holm_1996, Buhler_2009} 

\begin{proof}
Equation \eq{eq:example_Kelvin} can be proven by following a similar procedure as that used in Lemma \ref{relabeling_lemma}. The only difference is that now the simplectic form associated to $\overline{\mc{C}}_0$ is $-\vec{\mathrm{d}} \overline{\Theta} $, where the 1-form $\overline{\Theta}$ is given by
\begin{align}
	\overline{\Theta}
		[\delta \meanh,\delta \meanp,\delta\meanrho,\delta\meanchi,\delta\mc{I},\delta S] 
			= 	\int_Q \meanp \cdot \overline{\vec{\xi}} \,d^3\bm{x} 
				+ \int_Q \meanrho (\delta \meanchi
						+\overline{\vec{\xi}} \cdot \del \meanchi)\,d^3\bm{x}
				- \int_Q \mc{I} (\delta S+ \overline{\vec{\xi}}
							\cdot\del S)\,d^3\bm{x}
\end{align}
and $\overline{\vec{\xi}} \doteq - \delta \meanh \cdot (\del \meanh)^{-1}$. Alternatively, we can also show the result in \Eq{eq:example_Kelvin} by using the general Kelvin's theorem for Eulerian WKB obtained in Corollary \ref{corollary_1}.  Indeed, since solutions along the slow manifold [see, \eg \Eqs{eq:alpha_sol}--\eq{eq:hatchi_sol}] are solutions of the NL--WKB isothermal fluid equations \eq{eq:extELE_v}--\eq{eq:extELE_chi}, then \Eq{eq:example_Kelvin} can be obtained by restricting \Eq{circulation_family} onto the slow manifold and averaging over the phase $\theta$. Specifically, we have
\begin{equation}
	0 	=\frac{\mathrm{d}}{\mathrm{d}t} 
			\overline{ \left( \oint_{C_\theta} 
				\frac{\fluctp^{\star S/\ep}}{\fluctrho^{\star S/\ep}} \cdot d\bm{x} \right) }
		=	\frac{\mathrm{d}}{\mathrm{d}t} 
				\oint_{\overline{ C }} 
				\overline{ \left( \widetilde{\tau}^{\star S / \ep}_*	 
				\frac{\fluctp^{\star S/\ep}}{\fluctrho^{\star S/\ep}} \cdot d\bm{x}\right) }.
	\label{eq:example_Kelvin_aux}
\end{equation}
Following a similar calculation as in \Eq{eq:example_push_p}, we can calculate the pushforward appearing in the integral above:
\begin{align}
\widetilde{\tau}^{\star S / \ep}_*	 
				\frac{\fluctp^{\star S/\ep}}{\fluctrho^{\star S/\ep}} \cdot d\bm{x} 
	&	=		d \vec{x} \cdot \del  (( \flucttau^{\star S / \ep})^{-1} )  
				\cdot 
				\left( \frac{\fluctp^{\star S/\ep}}{\fluctrho^{\star S/\ep}}  \right) 
						\circ  (\flucttau^{\star S / \ep})^{-1} 
				\notag \\
	&	=		d \vec{x} \cdot ( \mathbb{I} - \ep^2 \del  \hatalpha^{\star S / \ep} ) 
				\cdot 
				\bigg( 
					\frac{\meanp}{\meanrho}
					+ \ep \frac{\hatp^{\star S/\ep}}{\meanrho}		
					- \ep \frac{\meanp}{\meanrho}	\frac{\hatrho^{\star S/\ep}}{\meanrho}	
					- \ep^2 \frac{\hatp^{\star S/\ep}}{\meanrho}	
							\frac{\hatrho^{\star S/\ep}}{\meanrho}			
						\notag \\
	&			\qquad
					+ \ep^2 \frac{\meanp}{\meanrho}	
							\bigg( \frac{\hatrho^{\star S/\ep}}{\meanrho^{\star S/\ep}}\bigg)^2	
				\bigg) 
						\circ  (\vec{id} - \ep^2 \hatalpha^{\star S / \ep} ) 
				+\mc{O}(\ep^3)
				\notag \\
	&	=		d \vec{x} \cdot ( \mathbb{I} - \ep \del S \otimes \pd_\theta \hatalpha^{\star S / \ep} ) 
				\cdot 
				\bigg( 
					\frac{\meanp}{\meanrho}
					+ \ep \frac{\hatp^{\star S/\ep}}{\meanrho}		
					- \ep \frac{\meanp}{\meanrho}	\frac{\hatrho^{\star S/\ep}}{\meanrho}	
					- \ep^2 \frac{\hatp^{\star S/\ep}}{\meanrho}	
							\frac{\hatrho^{\star S/\ep}}{\meanrho}			
					\notag \\
	& \qquad
					+ \ep^2 \frac{\meanp}{\meanrho}
						\bigg( \frac{\hatrho^{\star S/\ep}}{\meanrho^{\star S/\ep}}\bigg)^2	
					- \ep^2 \del S \cdot \hatalpha^{\star S / \ep} \pd_\theta 
								\frac{\hatp^{\star S/\ep}}{\meanrho}	
					+ \ep^2 \del S \cdot \hatalpha^{\star S / \ep} \pd_\theta 
								\frac{\meanp}{\meanrho}	\frac{\hatrho^{\star S/\ep}}{\meanrho}	
				\bigg) 
				+  \ep^2 \mathrm{Osc}	,
				\notag \\
	&	=		\bigg( 
					\frac{\meanp}{\meanrho}
					- \ep^2 \frac{\hatp^{\star S/\ep}}{\meanrho}	
								\frac{\hatrho^{\star S/\ep}}{\meanrho}			
					+ \ep^2 \frac{\meanp}{\meanrho}	
								\bigg( \frac{\hatrho^{\star S/\ep}}{\meanrho^2}\bigg)^2
					- \ep   \pd_\theta \hatalpha^{\star S / \ep}   \cdot \frac{\meanp}{\meanrho} \del S
				\bigg) \cdot d \vec{x}
				+  \ep^2 \mathrm{Osc}	,
\end{align}
where in the last line, some $\mc{O}(\ep^2)$ terms were written as a total derivative of $\theta$. Since their $\theta$-average is zero, we omitted writing them and placed them under the symbol ``$\ep^2 \mathrm{Osc}	$". Finally, averaging over $\theta$ and substituting the expressions for $\hatp^{\star S/\ep}$ and $\mc{I}$ leads to
\begin{equation}
	\overline{ \left( 
		\widetilde{\tau}^{\star S / \ep}_*	 
			\frac{\fluctp^{\star S/\ep}}{\fluctrho^{\star S/\ep}} \cdot d\bm{x}  \right) }
		=	\frac{\meanp}{\meanrho} \cdot d \vec{x}
			-	\frac{\mc{I}}{\meanrho} \del S \cdot d \vec{x}.
\end{equation}
Inserting this into \Eq{eq:example_Kelvin_aux} finishes the proof.
\end{proof}


\bibliography{cumulative_bib_file.bib}



\end{document}